\documentclass[acmsmall, nonacm]{acmart}

\usepackage{subcaption}

\usepackage{wrapfig}

\usepackage{textcomp} 
\usepackage{stmaryrd} 
\SetSymbolFont{stmry}{bold}{U}{stmry}{m}{n}

\usepackage{mathtools}

\usepackage[capitalize,nameinlink]{cleveref} 
\crefformat{definition}{Definition~#2#1#3}
\crefformat{equation}{Equation~(#2#1#3)}
\crefformat{lemma}{Lemma~#2#1#3}
\crefformat{theorem}{Theorem~#2#1#3}
\crefformat{section}{Section~#2#1#3}

\usepackage{thmtools,thm-restate}

\usepackage{url}

\usepackage{mathpartir}

\usepackage{booktabs}
\usepackage{multirow}

\usepackage{xcolor}
\definecolor{ltblue}{rgb}{0,0.4,0.4}
\definecolor{dkblue}{rgb}{0,0.1,0.6}
\definecolor{dkgreen}{rgb}{0,0.35,0}
\definecolor{dkviolet}{rgb}{0.3,0,0.5}
\definecolor{dkred}{rgb}{0.5,0,0}

\usepackage{listings}
\newcommand{\haskell}[1]{\lstinline{#1}}

\usepackage{tikz} 
\usetikzlibrary{cd}
\usetikzlibrary{positioning}
\usetikzlibrary{shapes.multipart}
\usetikzlibrary{%
  arrows,%
  shapes.misc,
  shapes.arrows,%
  shapes.callouts,
  chains,%
  matrix,%
  positioning,
  scopes,%
  decorations.pathmorphing,
  decorations.text,
  shadows%
}
\usetikzlibrary{quantikz2}

\usepackage{xparse}
\NewDocumentCommand{\optionalParens}{s m m}{
    \IfBooleanTF{#2}{\left(#3\right)}{\IfBooleanTF{#1}{~#3}{#3}}
}
\NewDocumentCommand{\apply}{m s O{} m}{
    #1#3 \optionalParens*{#2}{#4}
}

\NewDocumentCommand{\applytwo}{m O{} s m s m}{ 
   #1#2~\optionalParens{#3}{#4}~\optionalParens{#5}{#6}
}

\newcommand{\letin}[3]{\text{let}~#1 = #2~\text{in}~#3}
\newcommand{\caseof}[2]{\text{case}~#1~\text{of}~\left\{#2\right\}}
\newcommand{\inl}[1]{\iota_1(#1)}
\newcommand{\inr}[1]{\iota_2(#1)}

\usepackage{xspace}

\newcommand{\ie}{\emph{i.e.,}\xspace}

\usepackage{cmll} 
\newcommand{\lolli}{\multimap}

\usepackage{braket}

\newcommand{\ZZ}{\mathbb{Z}}

\usepackage{style/draft}

\draftfalse

\newnote[JP]{jennifer}{red}
\newnote[SW]{sam}{blue}
\newnote[reviewer]{reviewers}{purple}
\newnote{fixme}{green}
\newnote[Q]{question}{purple}
\newnote{note}{green}
\newnote{todo}{green}
\newnote[Cite:]{tocite}{green}

\usepackage{enumerate}

\usepackage{doi}

\usepackage{anyfontsize} 

\AtEndPreamble{%
    \theoremstyle{acmplain}
    \newtheorem{proofpart}{Case}[theorem]
}

\NewDocumentCommand{\Zd}{st'}{\IfBooleanTF{#1}{\tfrac12\mathbb{Z}_{d'}}
                               { \IfBooleanTF{#2}{\mathbb{Z}_{d'}}
                               { \mathbb{Z}_d
                               }}}
\RenewDocumentCommand{\ZZ}{st'}{\IfBooleanTF{#1}{\tfrac12\mathbb{Z}_{4}}
                               { \IfBooleanTF{#2}{\mathbb{Z}_{4}}
                               { \mathbb{Z}_2
                               }}}

\newcommand{\unitaryGroup}{\mathcal{U}_{2^n}}

\newcommand{\Pauligroup}{\mathrm{Pauli}_{d,n}}

\newcommand{\Pcliffordgroup}{\mathrm{PClif}_{d,n}}
\newcommand{\cliffordgroup}{\mathrm{Clif}_{n}}
\newcommand{\Cliffordgroup}{\mathrm{Clif}_{d,n}}
\newcommand{\Pobj}[1]{\mathrm Q_{#1}}

\newcommand{\Punit}{\bullet}
\newcommand{\Qdn}{\mathrm{Q}_{d,n}}
\newcommand{\idmorph}[1]{\mathrm{id}_{#1}}

\NewDocumentCommand{\pcl}{s}{\IfBooleanTF{#1}{\mathit{PCl}}{\mathbf{PCl}}}
\NewDocumentCommand{\pauligroup}{s}{\IfBooleanTF{#1}{\mathit{P}}{\mathbf{P}}}
\newcommand{\spgroup}{\mathrm{Sp}}

\newcommand{\cat}[1]{\mathcal{#1}}
\newcommand{\PCat}{\cat{P}_c}
\newcommand{\CCat}{\cat{L}}
\newcommand{\SCat}{\cat{S}}

\newcommand{\lambdaC}{\lambda^{\CCat}}
\newcommand{\lambdaP}{\lambda^{\PCat}}

\renewcommand{\textsf}{\text}

\newcommand{\cprod}{\star}

\newcommand{\ptensor}{\boxtimes}

\newcommand{\qudit}{\text{qu\emph{d}it}\xspace}
\newcommand{\qudits}{\text{qu\emph{d}its}\xspace}

\newcommand{\Unit}{\ZZ} 
\newcommand{\Unitd}{\Zd}

\newcommand{\force}[1]{#1}

\newcommand{\Pauli}{\textsc{Pauli}}
\newcommand{\Clifford}{\apply{\texttt{Clifford}}}
\newcommand{\compile}{\textsf{compile}}
\newcommand{\rank}[1]{\mathrm{rank}(#1)}
\newcommand{\flatten}[1]{\lvert #1 \rvert}

\newcommand{\Clf}{\textit{Clf}}

\newcommand{\CNOT}{\textsf{CNOT}}

\NewDocumentCommand{\CQSTATE}{O{}}{\textsf{CQ}^{#1}}
\newcommand{\lift}[1]{\lceil #1 \rceil}

\newcommand{\FinNat}[1]{\textrm{Nat}(#1)}

\NewDocumentCommand\interp{sm}
        {\IfBooleanTF{#1}
                {\conjugate{\left\llbracket #2 \right\rrrbracket}}
                {\left\llbracket #2 \right\rrbracket}
        }
\newcommand{\interpL}[1]{\interp{#1}^{\scalebox{0.5}{\ensuremath{\CCat}}}}
\newcommand{\interpP}[1]{\interp{#1}^{\scalebox{0.5}{\ensuremath{\PCat}}}}
\newcommand{\up}[1]{\lceil #1 \rceil}
\newcommand{\down}[1]{\lfloor #1 \rfloor}

\NewDocumentCommand\commute{smm}
        {#2 \IfBooleanT{#1}{\not}\upmodels #3}

\newcommand{\dom}{\textsf{dom}}

\newcommand{\qft}{\textrm{qft}}

\newcommand{\sympcaseof}[3]{\caseof{#1 \vartriangleright #2}{#3}}
\newcommand{\PauliZ}{\textbf{Z}}
\newcommand{\PauliX}{\textbf{X}}
\newcommand{\PauliY}{\textbf{Y}}
\newcommand{\PauliI}{\textbf{I}}
\newcommand{\vdashS}{\vdash^{\scalebox{0.5}{$\textbf{S}$}}}
\newcommand{\vdashL}{\vdash^{\scalebox{0.5}{$\CCat$}}}
\newcommand{\vdashP}{\vdash^{\scalebox{0.5}{$\PCat$}}}
\newcommand{\vdashPL}{\vdash^{\scalebox{0.5}{$\textbf{L}$}}}
\newcommand{\vdashT}{\vdash^{\scalebox{0.5}{\textbf{N}}}}
\newcommand{\vdashC}{\vdash^{\scalebox{0.5}{$\CCat$}}}
\newcommand{\pauliToClifford}{\apply{\texttt{pauliToClifford}}}
\newcommand{\R}{\Unitd}
\NewDocumentCommand{\symplecticform}{O{}mm}{\omega_{#1}\left(#2,#3\right)}
\NewDocumentCommand{\SYMPLECTICFORM}{O{}}{\overline{\omega}_{#1}}

\newcommand{\ctxequiv}{\approx^{\text{ctx}}}
\newcommand{\lrequiv}{\equiv}
\newcommand{\ValRelation}[1]{\mathcal{V}_{#1}}
\newcommand{\ExpRelation}[1]{\mathcal{E}_{#1}}

\newcommand{\inz}{\text{inz}}
\newcommand{\inx}{\text{inx}}
\newcommand{\PauliM}{\apply{\textsc{Phase}}*}
\newcommand{\Norm}[1]{\mathcal{N}[#1]}
\newcommand{\Val}[1]{\mathcal{V}[#1]}

\newcommand{\config}[2]{\left\langle #1 \right\rangle \xspace #2}
\NewDocumentCommand{\discard}{O{}m}{\texttt{discard}_{#1}(#2)}
\NewDocumentCommand{\dup}{O{}m}{\texttt{dup}_{#1}(#2)}
\newcommand{\zero}{\textbf{0}}
\newcommand{\obj}[1]{\textsf{Obj}(#1)}
\NewDocumentCommand{\norm}{O{} m}{\left\lVert #2 \right\rVert_{#1}}

\newcommand{\sgn}[1]{\apply{\text{sgn}}*{#1}}
\newcommand{\divd}[1]{\sgn{#1}}
\newcommand{\clambda}{\underline{\lambda}}
\newcommand{\CliffordFunction}[2]{\Clifford(#1,#2)}
\newcommand{\PseudoCliffordFunction}[2]{\lift{#1 \lolli #2}}
\newcommand{\Frame}[2]{\textsf{Tableau}(#1,#2)}

\NewDocumentCommand{\pseudoinverse}{O{} m}{\textsf{pinv}_{#1}(#2)}
\NewDocumentCommand{\pseudomuinverse}{O{} m}{\textsf{minv}_{#1}(#2)}

\newcommand{\orth}{\textrm{orth}}

\newcommand{\stype}{\sigma}

\newcommand{\pow}[2]{\textsf{pow}(#1,#2)}
\newcommand{\muof}[1]{{#1}^{\mu}}
\NewDocumentCommand{\psiof}{sm}{
    \IfBooleanTF{#1}{(#2)^{\boldsymbol{\psi}}}
                    {{#2}^{\boldsymbol{\psi}}}
}

\newcommand{\ntype}{\boldsymbol{\nu}}

\draftfalse
\begin{document}
\title{Qudit Quantum Programming with Projective Cliffords}

\author{Jennifer Paykin}
\orcid{0009-0008-9502-3219}
\affiliation{%
  \institution{University of Vermont}
  \city{Burlington}
  \country{USA}
}
\affiliation{%
  \institution{Intel}
  \city{Hillsboro}
  \country{USA}
}
\email{jpaykin@uvm.edu}

\author{Sam Winnick}
\orcid{0009-0009-9511-6551}
\affiliation{%
  \institution{Simon Fraser University}
  \city{Vancouver}
  \country{Canada}
}
\affiliation{%
  \institution{University of Waterloo}
  \city{Waterloo}
  \country{Canada}
}
\email{swinnick@sfu.ca}


\begin{abstract}
    This paper introduces a novel abstraction for programming quantum operations, specifically \emph{projective Cliffords}, as functions over the \qudit Pauli group.
    Generalizing the idea behind Pauli tableaux, we introduce a type system and lambda calculus for projective Cliffords called $\lambdaP$ that captures well-formed Clifford operations via a Curry-Howard correspondence with a particular encoding of the Clifford and Pauli groups.
    In $\lambdaP$, users write functions that encode projective Cliffords $P \mapsto U P U^\dagger$, and such functions are compiled to circuits executable on modern quantum computers that transform quantum states $\ket{\varphi}$ into $U \ket{\varphi}$, up to a global phase. 
    Importantly, the language captures not just qubit operations, but \qudit operations for any dimension $d$.

    Throughout the paper we explore what it means to program with projective Cliffords through a number of examples and a case study focusing on stabilizer error correcting codes.
\end{abstract}




\maketitle

\section{Introduction}
\label{sec:introduction}
In the usual model of quantum computing, a unitary operator $U$ is characterized by the fact that it transforms the quantum state $\ket{\varphi}$ into the state $U \ket{\varphi}$.\footnote{A unitary transformation on $n$ qubits is a complex matrix $U$ of size $2^n \times 2^n$ whose inverse is its conjugate transpose, denoted $U^\dagger$. An $n$-qubit quantum state is a complex vector of size $2^n$.} \xspace
In some cases, it is more useful to describe a unitary by its conjugation action $U P U^\dag$ on elements of the multi-qubit \emph{Pauli group} $P \in \pauligroup$~\citep{gottesman1998heisenberg}. This perspective is especially important \emph{Clifford operators}: unitaries that preserve the Pauli group under conjugation, i.e., $U P U^\dag \in \pauligroup$. A Clifford's conjugation action $P \mapsto U P U^\dag$ defines an equivalence class of physically indistinguishable Clifford operators $[U]$, which we refer to collectively as a \emph{projective Clifford operator}.\footnotemark

\footnotetext{Formally, a projective Clifford $[U]$ is an equivalence class of Cliffords up to global phase: $U_1\sim U_2$ if and only if $U_2=e^{i\theta}U_1$ for some $\theta$. Two unitaries are equivalent if and only if they have the same conjugation action on the Pauli group: $U_1 P U_1^\dagger = U_2 P U_2^\dagger$ for all $P \in \pauligroup$.}

This work explores how to program projective Cliffords directly through their conjugation action on Paulis, rather than as circuits or by their action on states. 
The result is a \emph{new programming paradigm} where users write functions encoding projective Cliffords $P \mapsto U P U^\dagger$, and such functions are compiled to circuits that implement $\ket{\varphi}$ into $U \ket{\varphi}$ (up to an indistinguishable global phase). This paradigm emphasizes the mathematical intuition behind Clifford-based algorithms such as stabilizer error correction, while enabling efficient compilation to circuits implementable on current quantum hardware using Pauli tableaux~\citep{aaronson2004,farinholt2014ideal}.

This paper makes three main contributions. 
First, we present a sound and complete lambda-calculus, $\lambdaP$, that precisely captures projective Cliffords and compiles efficiently to circuits.
Second, we provide detailed case studies demonstrating this programming paradigm.
Third, we generalize these frameworks to \qudits of arbitrary dimension, with the goal of extending this work towards universal quantum computing.

The rest of this introduction illustrates the main ideas behind $\lambdaP$ through examples. For ease of presentation, we start with qubit Cliffords before moving on to qu\(d\)its in \cref{sec:why-qudits}.




\subsection{Programming with Paulis}

In order to program projective Cliffords, we first need to talk about Pauli operators.

The single-qubit Pauli matrices $X$, $Y$, $Z$, and $I$ are defined as follows:
    \begin{align*}
        X=\begin{pmatrix}
            0 & 1 \\ 1 & 0
        \end{pmatrix}
        &&
        Z = \begin{pmatrix}
            1 & 0 \\ 0 & -1
        \end{pmatrix}
        &&
        Y = i X Z = \begin{pmatrix}
                        0 & -i \\ i & 0
                    \end{pmatrix}
        &&
        I = \begin{pmatrix}
            1 & 0 \\ 0 & 1
        \end{pmatrix}
    \end{align*}
Each single-qubit Pauli matrix can be encoded as a pair of bits\footnotemark\xspace $x,z \in \ZZ$ as $\Delta_{[x,z]} = i^{x z} X^{x} Z^{z}$:
\begin{align*}
    X = \Delta_{[1,0]}
    \qquad
    Z = \Delta_{[0,1]}
    \qquad
    Y = \Delta_{[1,1]}
    \qquad
    I = \Delta_{[0,0]}
\end{align*}
At the core of $\lambdaP$ is a sub-language of vectors $v$ that encode Pauli operators $\Delta_v$.
This sub-language, $\lambdaC$, supports addition and scalar multiplication of vectors, where scalars, of type $\Unit$, correspond to the integers modulo $2$. Its type system uses linear logic to ensure functions are linear transformations in their input, similar to the linear-algebraic $\lambda$-calculus of \citet{diazcaro2024linear}.

\footnotetext{Other encodings of Paulis in the literature, such as $D_{[x,z]} = X^{x} Z^{z}$, may appear simpler, but the $\Delta$ operators are better suited for encoding projective Cliffords, as discussed by \citet{winnick2024condensed}.}

In $\lambdaC$ we define linear transformations by how they act on the basis elements of the vector space. For example, the \emph{symplectic form} function $\omega(v_1,v_2)$ returns a scalar $s$ such that $\Delta_{v_1} \Delta_{v_2} = (-1)^s \Delta_{v_2} \Delta_{v_1}$. It is defined by its action on the basis elements $\inl{1}=[1,0]$ and $\inr{1}=[0,1]$ of $\Unit \oplus \Unit$:\footnote{Code blocks use Haskell-style pattern matching syntax as pseudo-code; see \cref{appendix:glossary} in the supplimentary materials.
}
\begin{lstlisting}[basicstyle=\footnotesize]
    omega ::^L Unit oplus Unit -o Unit oplus Unit -o Unit
    omega (in1 1) (in1 1) = 0       -- X X = (-1)^0 X X
    omega (in2 1) (in2 1) = 0       -- Z Z = (-1)^0 Z Z
    omega (in1 1) (in2 1) = 1       -- X Z = (-1)^1 Z X
    omega (in2 1) (in1 1) = 1       -- Z X = (-1)^1 X Z
\end{lstlisting}
When applied to non-basis arguments, the function breaks up its arguments into a linear transformation of these basis elements. For example, we can check that $Y Y = (-1)^s Y Y$, where:
\begin{lstlisting}[basicstyle=\footnotesize]
    s = omega [1,1] [1,1] =   omega ([1,0] + [0,1]) ([1,0] + [0,1])
                          ->* omega [1,0] [1,0] + omega [1,0] [0,1] + omega [0,1] [1,0] + omega [1,0] [0,1]
                          ->* 0 + 0 + 1 + 1 = 0 IN Unit
\end{lstlisting}

\paragraph*{From Vectors to the Pauli Group}

The Pauli matrices generate a group under matrix multiplication where all elements have the form $i^r \Delta_v$ for $r \in \ZZ'$: for example, $X Z = -i Y = i^3 \Delta_{[1,1]}$.
However, Cliffords always send $\Delta_v$ to $\pm \Delta_{v'}$, never $\pm i \Delta_{v'}$, and so $\lambdaP$ really only needs to consider \emph{Hermitian} Paulis of the form $(-1)^s \Delta_v$ where $s \in \ZZ$.
We indicate this \emph{phase} $s \in \ZZ$ with the syntax $\config{s} v$ for $(-1)^s \Delta_v$. For example, $-\Delta_v$ is represented by \haskell{<1> v}.

The \emph{Hermitian product}~\citep{PaykinSchmitz2023pcoast} $\Delta_{v_1} \cprod \Delta_{v_2} = (-i)^{\omega(v_1,v_2)} \Delta_{v_1} \Delta_{v_2}$ defines a non-associative group-like structure (a loop structure) on the set of Hermitian Paulis. For example:
\begin{lstlisting}[basicstyle=\footnotesize]
    X *.* Z ->* <1> Y -- aka -Y                             Z *.* X ->* <0> Y -- aka +Y
\end{lstlisting}

\paragraph*{Multi-Qubit Paulis}

A Pauli on $n$ qubits is the tensor product of $n$ single-qubit Paulis, which we represent as $\Delta_v$ where $v \in (\ZZ \oplus \ZZ)^n$. The type of a multi-qubit Pauli is written $\Pauli \ptensor \cdots \ptensor \Pauli$, so as not to confuse it with the $\otimes$ from linear logic. For example, $-I \otimes X \otimes Y \otimes I$ can be expressed as:
\begin{lstlisting}[basicstyle=\footnotesize]
    negx2y3 :: |^ Pauli ** Pauli ** Pauli ** Pauli ^|
    negx2y3 = |^ <1> I ** X ** Y ** I ^|
\end{lstlisting}

Every Pauli type $\tau$ has an identity $I_\tau \triangleq \config{0}\zero$ defined in terms of the zero vector, as well as a Hermitian product $\cprod_\tau$. When clear from context, we omit the subscripts and just write $I$ and $\cprod$. 

\subsection{Programming with Projective Cliffords}

The Clifford group is the set of unitaries that preserves the Pauli group under conjugation:
    \[
        \cliffordgroup \triangleq \{ U \in \unitaryGroup \mid \forall P \in \pauligroup.~ U P U^\dagger \in \pauligroup \}
    \]
%
    For example, we can see that the Hadamard unitary 
    $H=\tfrac{1}{\sqrt{2}} \left( \begin{smallmatrix}
        1 & 1 \\ 1 & -1
    \end{smallmatrix}\right)$ 
    is Clifford based on how it acts on Pauli matrices by conjugation:
    \[
        H X H^\dagger = Z 
        \qquad \qquad
        H Y H^\dagger = -Y
        \qquad \qquad
        H Z H^\dagger = X
    \]
    We encode this as a projective Clifford in $\lambdaP$ as the following function, now defined by case analysis on the basis elements $X,Z$ of type $\Pauli$:
\begin{lstlisting}[basicstyle=\footnotesize]
    hadamard :: |^ Pauli -o Pauli ^|
    hadamard |^X^| = Z
    hadamard |^Z^| = X
\end{lstlisting}
The type of \lstinline{hadamard}, written \lstinline{|^Pauli -o Pauli^|}, indicates a projective Clifford function from one single-qubit \lstinline{Pauli} to another. 
The brackets \haskell{|^-^|} in the type and in the pattern matching syntax distinguish these functions from ordinary linear transformations like \haskell{omega} above.

We can evaluate \lstinline{hadamard} on $Y$ by virtue of the fact that $Y = X \cprod Z$:
\begin{lstlisting}[basicstyle=\footnotesize]
    hadamard Y = hadamard (X *.* Z) -> hadamard X *.* hadamard Z -> Z *.* X -> <1>Y
\end{lstlisting}
Note that this computation relies on the assumption that $H (X \cprod Z) H^\dagger$ is the same as $H X H^\dagger \cprod H Z H^\dagger$. Since $P_1 \cprod P_2 = (-i)^{\omega(P_1,P_2)} P_1 P_2$, this is only true if $\omega(H X H^\dagger, H Z H^\dagger) = \omega(X,Z)=1$. Indeed, this property---that a function on Pauli operators respects the symplectic form---is one of the defining features that ensures a linear transformation is in fact a projective Clifford.

The type system of $\lambdaP$ ensures that every well-typed function is not only a linear transformation, but also respects the symplectic form. When we write a case analysis like \haskell{hadamard} above, the type system ensures that the branches $t_x$ and $t_z$ of the case statement satisfy $\omega(t_x,t_z)=1$. In contrast, the following function does \emph{not} type check, since $\omega(X,X)=0 \neq \omega(X,Z)$:
\begin{lstlisting}[basicstyle=\footnotesize]
    -- Does not type check
    illTyped |^X^| = X
    illTyped |^Z^| = X
\end{lstlisting}

To define a projective Clifford over a multi-qubit Pauli $P : \tau_1 \ptensor \tau_2$, we also proceed by case analysis on the basis elements $\inl{q}=[q,\zero]$ and $\inr{q}=[\zero,q]$ of $\tau_1 \ptensor \tau_2$. For example, the following \haskell{swap} function is a valid projective Clifford:
\begin{lstlisting}[basicstyle=\footnotesize]
    swap :: |^ tau_1 ** tau_2 -o tau_2 ** tau_1 ^|
    swap |^ in1 q_1 ^| = in2 q_1
    swap |^ in2 q_2 ^| = in1 q_2
\end{lstlisting}
This case analysis has its own symplectic form check: the two branches $t_1$ and $t_2$ corresponding to \haskell{in1 q_1} and \haskell{in2 q_2} should satisfy $\omega(t_1,t_2)=\omega(\text{\haskell{in1}}~q_1,\text{\haskell{in2}}~q_2)=0$ for all $q_1$ and $q_2$.

As another example, the controlled-not matrix $CNOT$ is a Clifford over 2-qubit Paulis:
\begin{center}
\begin{minipage}[t]{0.5\textwidth}
    \begin{lstlisting}[basicstyle=\footnotesize]
  cnot :: |^ Pauli ** Pauli -o Pauli ** Pauli ^|
  cnot |^ in1 X ^| = in1 X *.* in2 X
  cnot |^ in1 Z ^| = in1 Z
  cnot |^ in2 X ^| = in2 X
  cnot |^ in2 Z ^| = in1 Z *.* in2 Z
    \end{lstlisting}
\end{minipage}
\begin{minipage}[t]{0.44\textwidth}
    \begin{align*}
    \begin{array}{ccc}
        i & CNOT X_i CNOT^\dagger & CNOT Z_i CNOT^\dagger \\ \midrule
        1 & X_1 X_2 & Z_1 \\
        2 & X_2 & Z_1 Z_2
    \end{array}
    \end{align*}
\end{minipage}
\end{center}

\subsection{Why Cliffords? Why Qu\(d\)its?}
\label{sec:why-qudits}

Why introduce programming paradigm for Clifford operations? It is well-known that Cliffords are not universal for quantum computing, and in fact are efficiently simulatable on classical computers~\citep{aaronson2004}. Nevertheless, Cliffords play key parts in almost every quantum algorithm, and form the backbone of quantum error correction and fault tolerance~\citep{gottesman2024surviving,nielsen2002quantum}. A key contribution of this work is to explore how expressing these algorithms in terms of their conjugation action on Paulis highlights reveals underlying structure and could inspire new algorithmic approaches.

The algebraic structure of projective Cliffords is already widely used via Pauli tableaux, which serve as intermediate representations in quantum compilers~\citep{Wu2023,PaykinSchmitz2023pcoast}, simulation~\citep{aaronson2004,li2022paulihedral}, circuit synthesis~\citep{schneider2023sat,schmitz2024graph}, and equivalence checking~\citep{amy2018towards,kissinger2020,borgna2021hybrid}.
Another important contribution of this work is to elevate the tableau data structure into a full programming feature, supporting variables, functions, equational reasoning, polymorphism, and more.

Neither of these justifications get us to universal quantum computing, however; for this we turn to \qudits.
A qudit is a $d$-dimensional quantum system $\alpha_0 \ket{0} + \cdots + \alpha_{d-1} \ket{d-1}$, generalizing a qubit when $d=2$. While an $n$-qubit state has dimension $2^n$, an $n$-\qudit state has dimension $d^n$. The Pauli and Clifford groups generalize naturally to $\qudit$ systems~\citep{deBeadrap2013linearized}, but for any fixed $d$, the $\qudit$ Clifford group is still not universal.

However, universality \emph{can} be achieved by combining Cliffords of different dimension. Consider the quantum Fourier transform (QFT) on $n$ qubits. By interpreting that $n$-qubit state as a single $\qudit$ where $d=2^n$, the QFT operator becomes Clifford and is simple to express in $\lambdaP$. By combining QFT with qubit Cliffords, we obtain a universal set of unitaries.

Currently, $\lambdaP$ assumes a fixed dimension $d$, but future work will extend $\lambdaP$ to a universal quantum programming language by allowing polymorphism in $d$. To prepare for this, a major contribution of this work is establishing $\lambdaP$ not just for qubits, but for arbitrary $\qudits$.

\subsection{Outline}


\cref{sec:background} begins with an overview of the qudit Pauli and Clifford groups, in particular the condensed encodings that justify the correctness of the $\lambdaP$ language.
The calculus, developed in \cref{sec:lambdaC,sec:lambdaPC}, consists of two main parts: a linear type system $\lambdaC$ (\cref{sec:lambdaC}) to describe vectors and linear transformations; and a lambda calculus for projective Cliffords $\lambdaP$ (\cref{sec:lambdaPC}) that incorporates the necessary orthogonality check based on the symplectic form. We give both operational and categorical semantics of these languages and prove them sound and complete.

\cref{sec:extensions} extends the core $\lambdaP$ calculus to include programming features such as polymorphism, higher-order functions, meta-transformations, and custom data types.
\cref{sec:examples} uses these features for a case study of stabilizer error correcting codes. Finally, \cref{sec:related-work} discusses related and future work, including how to extend $\lambdaP$ to a universal Pauli-based programming paradigm.

\section{Background}
\label{sec:background}
The type system of $\lambdaP$ is inspired by Pauli tableaux, a binary encoding of qubit Cliffords used widely for circuit simulation, compilation, and optimization~\citep{aaronson2004}. 
A tableau represents an $n$-qubit Clifford operator by its projective action on the Pauli $X$ and $Z$ operators, and is well-formed only if this action preserves the canonical commutation relations of $X$ and $Z$.

Pauli tableaux have been generalized to \qudit Cliffords, and are well-behaved for odd and odd prime dimension  $d$~\citep{deBeadrap2013linearized}. For event dimensions, however, encodings become more complex, and several alternative formulations have been proposed.

The goal of $\lambdaP$'s type system is to ensure that functions correspond exactly to valid encodings of projective Cliffords. To do this, we need an encoding of qudit projective Cliffords in arbitrary dimension $d$ whose well-formedness properties are enforceable by a type system.

For odd dimensions this is relatively straight-forward: encodings are pairs consisting of a linear map and a symplectomorphism (a linear map that respects the symplectic form). 
Linearity can be captured via techniques from linear type systems~\citep{diazcaro2024linear}, and the symplectomorphism property is compositional in a way that is suitable for type-checking.

When $d$ is even, encodings of projective Cliffords are significantly more complicated. Projective Cliffords can be expressed as $X^xZ^z\mapsto e^{i\pi r(x,z)/d}X^{\psi_1(x,z)}Z^{\psi_2(x,z)}$ where $(r,\psi)$ satisfy certain constraints, but the constraint on $r$ depends on $\psi$, and $r$ need not even be linear.
In \citep{winnick2024condensed} we present \textit{condensed encodings}, which characterize projective Cliffords in any dimension as pairs of a linear map and symplectomorphism, with an additional subtle \emph{phase correction}. In $\lambdaP$, this correction is handled internally in the $\beta$-reduction rules and is invisible to the programmer. 


With this motivation in mind, the rest of this section reviews qudits in both even and odd dimensions and summarizes the main properties of condensed encodings.

\subsection{Qu\(d\)its for Even and Odd \(d\)}

A $d$-level quantum system, for an integer $d\geq2$, is called a \qudit. Throughout this work we fix $d$ and work with a system of multiple qudits. We use $\zeta$ to refer to a fixed primitive complex $d$th root of unity, meaning $d$ is the least positive integer such that $\zeta^d=1$. Further, we let $\tau$ be a primitive $d'$th root of unity squaring to $\zeta$, where $d'=d$ if $d$ is odd, and $d'=2d$ if $d$ is even. For example, if $d=2$, then $\zeta=-1$ and we may take $\tau=i$.

In the condensed encodings of \cref{sec:encodings}, some calculations take place in $\Zd$ while others take place in $\Zd'$.
We write $\overline{\,\cdot\,}:\Zd'^{2n}\to\Zd^{2n}$ for the reduction mod $d$ homomorphism and $\underline{\,\cdot\,}:\Zd^{2n}\to\Zd'^{2n}$ for the inclusion function, which is not a homomorphism. These subtleties only matter when $d$ is even; when $d$ is odd we have $d'=d$ and $\overline{\,\cdot\,}$ and $\underline{\,\cdot\,}$ are each the identity function.

The \emph{sign} of $r' \in \Zd'$ is $0$ if $0 \leq r' < d$ and $1$ if $d \le r' < d'$; \ie $\sgn{r'}=\tfrac{1}{d}(r'-\overline{~\underline{r'}~}) \in \mathbb{Z}_{d'/d}$.

We also introduce the additive group:
\begin{align*}\label{eqn-12zd'}
    \Zd*=\{0,\tfrac12,1,\tfrac32,\cdots,\tfrac{2d'-1}2\}=(1,\tfrac12\,|\,\underbrace{1+\cdots+1}_d=0\textrm{ and }\tfrac12+\tfrac12=1)
\end{align*}
In other words, $\Zd*$ is obtained from $\Zd$ by adding a new element $\tfrac12$ (and closing under addition) if and only if $2$ does not already have a multiplicative inverse, \ie $d$ is even. So if $d$ is odd, then $\Zd=\Zd'=\Zd*$. In either case, there is a group isomorphism $\tfrac12:\Zd'\to\Zd*$, which allows for uniform treatment of the even and odd cases. To do this, we interpret half-element exponents $t\in\Zd*$ of $\zeta$ using the square root $\tau$, that is, $\zeta^t=\tau^{2t}$, where $2:\Zd* \to \Zd'$ is the inverse of $\tfrac12$. 

\subsection{Encodings of the Qudit Pauli Group}

The single-qubit Pauli operators generalize to single-qudit Paulis in the following way, where blank spaces in the matrix indicate the value $0$:
    \begin{align*}
            X = \begin{psmallmatrix}
                    &&&1\\
                    1\\
                    &\cdots\\
                    &&1
                \end{psmallmatrix}
            && Z = \begin{psmallmatrix}
                    1\\
                    &\zeta\\
                    &&\cdots\\
                    &&&\zeta^{d-1}
                \end{psmallmatrix}
            && Y = \tau X Z
    \end{align*}
The $n$-qudit Pauli group $\Pauligroup$ is generated by $\langle X_i, Y_i, Z_i \rangle$, where $P_i = I_{2^{i-1}} \otimes P \otimes I_{2^{n-i}}$. 

Given $v=[[x_1,z_1],\ldots,[x_n,z_n]] \in (\Zd' \oplus \Zd')^{n} \cong \Zd'^{2n}$, we define $\Delta_v=\bigotimes_i \tau^{x_i z_i} X^{x_i} Z^{z_i}$. Observe that the product in the exponent of $\tau$ is that of $\Zd'$ and not $\Zd$.

\begin{example}
\begin{align*}
    \Delta_{[1,0]} &= \tau^0 X^1 Z^0 = X 
    &
    \Delta_{[0,1]} &= \tau^0 X^0 Z^1 = Z
    &
    \Delta_{[1,1]} &= \tau^{1} X^1 Z^1 = Y
\end{align*}
\end{example}

\begin{proposition}
    Every element of the $n$-qudit Pauli group $\Pauligroup$ can be expressed in the form $\zeta^r \Delta_v$ where $r \in \Zd*$ and $v \in \Zd'^{2n}$.
\end{proposition}

For $v\in \Zd'^{2n}$, it is often useful to convert between $\Delta_{v}$ and $\Delta_{\underline{\overline{v}}}$ as follows:
\begin{align}
    \Delta_{v} = (-1)^{\sgn{v}} \Delta_{\underline{\overline{v}}}
  \qquad \qquad
  \text{where}~
    \sgn{v}=\sgn{\omega'(v,\overline{\underline v})}=\tfrac{1}{d} \omega'(v, \underline{\overline{v}})
\end{align}
and where $\omega' : \Zd'^{2n}\otimes\Zd'^{2n}\to\Zd'$ is the \textit{extended symplectic form} 
$\omega'([x_1,z_1],[x_2,z_2])=z_1\cdot x_2 - z_2\cdot x_1$.
That is, the extended symplectic form $\omega'$ is computed using arithmetic mod $d'$, whereas the standard symplectic form $\omega : \Zd^{2n} \otimes \Zd^{2n} \to \Zd$ uses arithmetic mod $d$.

\subsection{Encodings of Qudit Projective Cliffords}
\label{sec:encodings}

Similar to qubits, the qudit Clifford operators take the qudit Pauli group to itself under conjugation.
\begin{align*}
        \Cliffordgroup \triangleq \{ U \in \mathcal{U}_{d^n} \mid \forall P \in \Pauligroup.~ U P U^\dagger \in \Pauligroup \}
\end{align*}
Two Cliffords $U_1,U_2\in\Cliffordgroup$ are \emph{projectively equivalent} if they have the same conjugation action for all Paulis: $U_1 P U_1^\dagger = U_2 P U_2^\dagger$. This is the case exactly when $U_2=e^{i \theta} U_1$, meaning that $U_1$ and $U_2$ are quantum-mechanically indistinguishable. The equivalence classes $[U]$ of $U\in\Cliffordgroup$ form a group $\Pcliffordgroup$ with composition $[U_2][U_1]=[U_2U_1]$.

In order to program projective Cliffords, we draw on a common practice of \emph{encoding} each class $[U]$ as a pair of functions that act on the vectors $v$ in the representation of $\Delta_v$.

\begin{definition}[\citep{winnick2024condensed}]
    The \emph{condensed encoding} of a projective Clifford is a pair of functions $(\mu,\psi)$ where $\psi: \Zd^{2n} \rightarrow \Zd^{2n}$ is a symplectomorphism (respects the symplectic form on $\Zd^{2n}$), and $\mu : \Zd^{2n} \rightarrow \Zd$ is a linear transformation. This encoding corresponds to a projective Clifford $[U]\in\Pcliffordgroup$ defined by its conjugation action on $\Delta_{\underline b}$ for every standard basis vector $b\in\Zd^{2n}$ as follows:
    \begin{align*}
        U\Delta_{\underline b}U^\dag = \zeta^{\mu(b)}\Delta_{\underline{\psi(b)}}
    \end{align*}
\end{definition}

Let us write $V^*$ for the group of homomorphisms $\mu:\Zd^{2n}\to\Zd$ and $\spgroup(\Zd^{2n})$ for the group of symplectomorphisms $\psi:\Zd^{2n}\to\Zd^{2n}$ \ie for all $v_1,v_2 \in \Zd^{2n}$, $\omega(\psi(v_1),\psi(v_2))=\omega(v_1,v_2)$.

\begin{theorem}[\citep{winnick2024condensed}]
    The assignment of the condensed encoding $(\mu,\psi)$ to each projective Clifford $[U]$ is a bijection between the underlying sets of $\Pcliffordgroup$ and $V^*\times\spgroup(\mathbb Z_d^{2n})$.
\end{theorem}

The condensed encoding is not the only choice of encoding, but it has several advantages. First,  both the condition that $\mu$ is a linear transformation and the condition that $\psi$ is a symplectomorphism can be enforced by the type system, leading to a robust Curry-Howard correspondence. Second, the condensed encoding is applicable for all dimensions of $d$ in a uniform way, not just odd or prime instances, and only involves $\Zd$, not $\Zd'$. As a result, the programmer does not have to worry about the irregularity between $\Zd$ and $\Zd'$ in the even case.



With this motivation in mind, we now describe the relevant structure of these encodings. For proofs, the reader is referred to \citep{winnick2024condensed}.


To \emph{evaluate} a condensed encoding $(\mu,\psi)$ on an arbitrary Pauli $\zeta^t \Delta_v$ for $v \in \Zd'^{2n}$, we use the following evaluation formula:
\begin{align}\label{eqn-evaluation-formula}
    U\Delta_vU^\dag = \zeta^{\mu(\overline v)+\frac d2K^\psi(v)}\Delta_{\underline{\psi(\overline v)}}=(-1)^{K^\psi(v)}\zeta^{\mu(\overline v)}\Delta_{\underline{\psi(\overline v)}}
\end{align}
The correction term $K^\psi:\Zd'^{2n}\to\mathbb Z_{d'/d}$ in \cref{eqn-evaluation-formula} is given by 
\begin{align}\label{eqn-K}
    K^\psi(v) 
    &= \frac1d \sum_{i=1}^n \left(z_i x_i
            + x_i z_i \omega'(\underline{\psi b_i^x},\underline{\psi b_i^z})
            + x_i \omega'(\underline{\psi b_i^x}, \underline{\psi(\overline{v})})
            + z_i \omega'(\underline{\psi b_i^z}, \underline{\psi(\overline v)})
            \right)
\end{align}
where $v=([x_1,z_1],\ldots [x_n,z_n]) \in \Zd'^{2n}$ and where $b_i^x \in \Zd'^n$ (respectively $b_i^z$) is the standard basis vector that is $[1,0]$ (respectively $[0,1]$) at index $i$ and $0$ elsewhere.
Note that $K^\psi$ is \emph{not} linear on its input, and must be calculated independently for each $v$.
When $v \in \Zd^{2n}$, we write $\kappa^\psi(v)$ for $K^\psi(\underline v)$.

As can be inferred from their types, the arithmetic operations in \cref{eqn-K} are operations on $\Zd'$. Note that the parenthesized part is either $0$ or $d\in\Zd'$, so the result is well-defined in $\mathbb{Z}_{d'/d}$.

The proofs of the following lemmas follow from \cref{eqn-K}.
\begin{lemma} \label{lem:kappa-basis-0}
    On standard basis vectors $b_i^x$ and $b_i^z$, the function $\kappa^\psi$ is $0$.
\end{lemma}

\begin{lemma} \label{lem:kappa-R'-symplectomorphism}
    Suppose $\psi\in\spgroup(\Zd^{2n})$ has the property that for all $v_1,v_2\in\Zd^{2n}$,
    $\omega'(\underline{\psi(v_1)}, \underline{\psi(v_2)}) = \omega'(\underline{v_1},\underline{v_2})$.
    Then $\kappa^{\psi}(v) = 0$ for all $v\in\Zd^{2n}$.
\end{lemma}

\subsubsection{Composition and Inverses}

Let $(\mu_1,\psi_1)$ and $(\mu_2,\psi_2)$ be the condensed encodings of $[U_1]$ and $[U_2]$ respectively. The condensed encoding of $[U_2U_1]$ is $(\mu_3,\psi_2 \circ \psi_1)$, where $\mu_3$ is a linear map acting on each standard basis vector $b\in\Zd^{2n}$ by:
\begin{align}\label{eqn-composition-formula}
    \mu_3(b) = \mu_1(b) + \mu_2(\psi_1(b)) + \tfrac d2K^{\psi_2}(\underline{\psi_1(b)})
\end{align}
We caution that \cref{eqn-composition-formula} is only valid when applied to standard basis vectors, not arbitrary $v\in\Zd^{2n}$. From the composition formula one may work out that the identity element is $(0,\mathrm{id}_V)$, and that the inverse of $(\mu,\psi)$ is $(\mu_{\mathrm{inv}},\psi^{-1})$, where for each standard basis vector $b\in\Zd^{2n}$:
\begin{align}\label{eqn-inversion-formula}
    \mu_{\mathrm{inv}}(b) &= -\mu(\psi^{-1}(b))+\frac d2K^\psi(\underline{\psi^{-1}(b)}) \\
    \psi^{-1}(v) &= [[\omega(\psi(t_{n+1}),v),\omega(v,\psi(t_1))],\ldots,[\omega(\psi(t_{2n}),v),\omega(v,\psi(t_n))] \notag
\end{align}

\subsubsection{$\Qdn$ and the Condensed Product} \label{sec:background:qdn}

We know that elements of the Pauli group $\Pauligroup$ can be written uniquely as $\tau^{t'} \Delta_{\underline{v}}$ where $t' \in \Zd'$ and $v \in \Zd^{2n}$. However, both $\mu$ and $\psi$ are defined solely in terms of $\Zd$ rather than $\Zd'$. It is possible to avoid $\Zd'$ entirely by considering the following set:
\begin{align*}
    \Qdn = \{\zeta^t\Delta_{\underline v} \mid t \in \Zd ~\text{and}~ v\in\Zd^{2n}\}
\end{align*}
$\Qdn$ is a subset of $\Pauligroup$ with $d^{2n+1}$ elements. Note that since $\Qdn$ contains $X$, $Y$, and $Z$ but not $iY = \tau X Z$, it is not closed under matrix multiplication in the even case. 
Luckily, there is another operation $\cprod$ on $\Qdn$ that can be used as a replacement, the \emph{condensed product}:
\begin{align} \label{eqn:condensed-product-defn}
    \zeta^{r} \Delta_{\underline{u}} \cprod \zeta^{s} \Delta_{\underline{v}} 
    &\triangleq \zeta^{r+s} \tau^{-\underline{\omega(u,v)}} \Delta_{\underline{u}} \Delta_{\underline{v}}
    = \zeta^{r+s} (-1)^{\sgn{\omega'(\underline u,\underline v)}}\Delta_{\underline{u}+\underline{v}}
    \\
    &= \zeta^{r+s} (-1)^{\sgn{\omega'(\underline u,\underline v)} + \sgn{\underline{u}+\underline{v}}}\Delta_{\underline{u + v}}
        \notag
\end{align}

Lastly, we note that projective Clifford operations distribute over $\cprod$:
\begin{align} \label{lem:gamma-distributes-cprod}
    U(\Delta_u\cprod\Delta_v)U^\dag = U\Delta_uU^\dag\cprod U\Delta_vU^\dag
\end{align}

Using $\cprod$ together with the condensed encodings, we can build our type system and categorical semantics without needing to use the extended phase space $\Zd'^{2n}$.

\section{A Calculus for \(\Zd\)-Modules}
\label{sec:lambdaC}
Having established encodings $(\mu,\psi)$ for projective Cliffords, we can now design a type system for these encodings. This type system needs to check two properties: that both functions are linear in their input, and that $\psi$ respects the symplectic form \ie is a symplectomorphism. In this section we start with a type system for expressing linear maps, which we call $\lambdaC$. \cref{sec:lambdaPC} will extend $\lambdaC$ to encompass symplectomorphisms.

In $\lambdaC$, types correspond to $\Zd$-modules and expressions to $\Zd$-linear maps.\footnote{An $R$-module is just a generalization of a vector space where scalars are drawn from a ring $R$ rather than a field.} The operational semantics must be linear in the sense that function types respect addition and scalar multiplication of vectors. To achieve this we will use a type system based on linear logic but where the additive product $\&$ and additive sum $\oplus$ are combined into a single operation, which we denote by $\oplus$.

The types of $\lambdaC$ correspond to $\Zd$-modules built up inductively from the base type $\R$, the direct sum $\oplus$, and the space of linear transformations $\lolli$.
\begin{align*}
    \alpha &::= \R \mid \alpha_1 \oplus \alpha_2 \mid \alpha_1 \lolli \alpha_2
        \tag{$\CCat$-types}
\end{align*}

The syntax of linear $\CCat$-expressions is given by the following grammar:
\begin{align*}
    a &::= x \mid \letin{x}{a}{a'} 
        \tag{$\CCat$-expressions}  \\
        &\mid \zero_\alpha \mid r \mid a \cdot a' \mid a_1 + a_2 \\
        &\mid [a_1,a_2] \mid \caseof{a}{\inl{x_1} \rightarrow a_1 \mid \inr{x_2} \rightarrow a_2} \\
        &\mid \lambda x.a \mid a a' \\
    r &\in \Zd 
        \tag{constants}
\end{align*}
Here, $\zero_\alpha$ is the zero vector, $(\cdot)$ represents scalar multiplication, and $+$ represents vector addition.

The typing judgment (\cref{fig:typing-rules-C-simple}) has the form $\Delta \vdashL a : \alpha$, where $\Delta ::= \cdot \mid \Delta,x:\alpha$ is a map from linear variables to $\CCat$-types.
We write $\Delta_1, \Delta_2$ to mean the disjoint union of $\Delta_1$ and $\Delta_2$, under the condition that $\dom(\Delta_1) \cap \dom(\Delta_2)=\emptyset$. 

Because $\oplus$ is a biproduct in $\CCat$---both a product and a coproduct---its introduction rule in $\lambdaC$ mirrors the additive product rule from linear logic ($\&$), while its elimination rule mirrors the additive sum rule ($\oplus$). Intuitively, case analysis of a vector $a$ of type $\alpha_1 \oplus \alpha_2$ decomposes $a$ into the linear span of $[x_1,\zero]$ (written $\inl{x_1}$) and $[\zero,x_2]$ (written $\inr{x_2}$).

\begin{figure}
    \centering
    \[ \begin{array}{c}
        \inferrule*[right=$\CCat$-var]
            {\Delta = x:\alpha}
            {\Delta \vdashL x : \alpha}
        \qquad \qquad
        \inferrule*[right=$\CCat$-let]
            {\Delta \vdashL a : \alpha \\
             \Delta',x:\alpha \vdashL a' : \alpha'
            }
            {\Delta,\Delta' \vdashL \letin{x}{a}{a'} : \alpha'}
        \\ \\
        \inferrule*[right=$\CCat$-const]
            {r \in \Zd}
            {\cdot \vdashL r : \R}
        \qquad \qquad
        \inferrule*[right=$\CCat$-$\cdot$]
            {\Delta \vdashL a : \R \\
             \Delta' \vdashL a' : \alpha
            }
            {\Delta,\Delta' \vdashL a \cdot a' : \alpha}
        \qquad \qquad
        \inferrule*[right=$\CCat$-$\zero$]
            {~}
            {\Delta \vdashL \zero : \alpha}
        \\ \\
        \inferrule*[right=$\CCat$-$+$]
            { \Delta \vdashL a_1 : \alpha \\
              \Delta \vdashL a_2 : \alpha
            }
            { \Delta \vdashL a_1 + a_2 : \alpha }
        \qquad \qquad
        \inferrule*[right=$\CCat$-$\oplus$]
            { \Delta \vdashL a_1 : \alpha_1 \\ \Delta \vdashL a_2 : \alpha_2 }
            { \Delta \vdashL [a_1,a_2] : \alpha_1 \oplus \alpha_2 }
        \\ \\
        \inferrule*[right=$\CCat$-case]
            { \Delta \vdashL a : \alpha_1 \oplus \alpha_2 \\
              \Delta',x_1:\alpha_1 \vdashL a_1 : \alpha' \\
              \Delta',x_2:\alpha_2 \vdashL a_2 : \alpha'
            }
            { \Delta, \Delta' \vdashL \caseof{a}{\inl{x_1} \rightarrow a_1 \mid \inr{x_2} \rightarrow a_2} : \alpha' }
        \\ \\
        \inferrule*[right=$\CCat$-$\lambda$]
            { \Delta,x:\alpha \vdashL a : \alpha'}
            { \Delta \vdashL \lambda x.a : \alpha \lolli \alpha' }
        \qquad \qquad
        \inferrule*[right=$\CCat$-app]
            { \Delta_1 \vdashL a_1 : \alpha \lolli \alpha' \\
              \Delta_2 \vdashL a_2 : \alpha
            }
            { \Delta_1, \Delta_2 \vdashL a_1 a_2 : \alpha' }
    \end{array}\]
    \caption{Typing rules for $\CCat$-expressions.}
    \Description{A set of ten typing rules for $\CCat$-expressions, one for each type of expression.}
    \label{fig:typing-rules-C-simple}
\end{figure}

Notice that while non-zero constants in $\Zd$ must be typed under the empty context, the $\zero$ vector can be typed under an arbitrary context.

The small-step operational semantics, shown in \cref{fig:operational-semantics-C-simple}, has the form $a \rightarrow a'$ where $\cdot \vdashL a : \alpha$.
We write $a \rightarrow^\ast a'$ for the reflexive transitive closure of
$\rightarrow$.
In the next section we will show that the normal forms of the step relation $\rightarrow$ are the following values:
\begin{align*}
    v &::= r \mid [v_1,v_2] \mid \lambda x.a
        \tag{values}
\end{align*}

\begin{figure}
    \centering
    {\small
    \begin{align*}
      \begin{aligned}
        \letin{x}{v}{a'} 
            &\rightarrow_{\beta}
                a'\{v/x\}
            \\
        \caseof{[v_1,v_2]}{\inl{x_1} \rightarrow a_1 \mid \inr{x_2} \rightarrow a_2}
            &\rightarrow_\beta
            a_1\{v_1/x_1\} + a_2\{v_2/x_2\}
        \\
        (\lambda x.a) v &\rightarrow_\beta a\{v/x\}
       \end{aligned}
       \qquad
       \begin{aligned}
            \zero_{\R}
                &\rightarrow_\beta
                0
                \\
            \zero_{\alpha_1 \oplus \alpha_2}
                &\rightarrow_\beta
                [\zero_{\alpha_1}, \zero_{\alpha_2}]
                \\
            \zero_{\alpha \lolli \alpha'}
                &\rightarrow_\beta
                \lambda x.\zero_{\alpha'}
       \end{aligned}
    \end{align*}
    \newline
    \begin{align*}
      \begin{aligned}
        r_1 \cdot r_2
            &\rightarrow_\beta
            (r_1r_2) \in \Zd
        \\
        r \cdot [v_1, v_2]
            &\rightarrow_\beta
            [r \cdot v_1, r \cdot v_2]
        \\
        r \cdot \lambda x.a
            &\rightarrow_\beta
            \lambda x.r \cdot a
      \end{aligned}
      \qquad \qquad
      \begin{aligned}
        r_1 + r_2
            &\rightarrow_\beta
            (r_1 + r_2) \in \Zd
        \\
        [v_1,v_2] + [v_1',v_2']
            &\rightarrow_\beta
            [v_1+v_1', v_2+v_2']
        \\
        (\lambda x_1.a_1) + (\lambda x_2.a_2)
            &\rightarrow_\beta
            \lambda x. a_1\{x/x_1\} + a_2\{x/x_2\}
      \end{aligned}
    \end{align*}
    \normalsize}
    \caption{$\beta$-reduction rules for $\CCat$-expressions. The full call-by-value small-step operational semantics rules can be found in the supplementary material (\cref{app:lambdaC-reduction-rules}).}
    \Description{A set of twelve beta reduction rules for \(\CCat\)-expressions. The rules are grouped into four sets of three. The first set of rules are tradition beta reduction rules for let statements, case statements, and function application. The second set of rules describes how to reduce expressions of the form 0 for each type alpha. The third set of rules describes how to reduce expressions of the form \(r \cdot v\). The fourth set of rules describes how to reduce expressions of the form \(v1 + v2\).}
    \label{fig:operational-semantics-C-simple}
\end{figure}

\subsection{Type Safety and Other Meta-theory}

\begin{theorem}[Progress] \label{thm:progress}
    If $\cdot \vdashL a : \alpha$ then either $a$ is a value, or there is some $a'$ such that $a \rightarrow a'$.
\end{theorem}
\begin{proof}
    By induction on the typing judgment.

    To demonstrate, consider the case when $a = \caseof{a'}{\inl{x_1} \rightarrow a_1 \mid \inr{x_2} \rightarrow a_2}$.
    If $a'$ can take a step, so can $a$. If not, $a'$ is a value of type $\alpha_1 \oplus \alpha_2$, in which case $a$ can take a step via $\beta$-reduction.

    The remainder of the cases are similar.
\end{proof}

We can also prove that substitution and $\beta$-reduction preserve the typing relation.

\begin{lemma}[Substitution]
    If $\cdot \vdashL v : \alpha$ and $\Delta,x:\alpha \vdashL a : \alpha'$ then $\Delta \vdashL a\{v/x\} : \alpha'$.
\end{lemma}
\begin{proof}
    By induction on $\Delta,x:\alpha \vdashL a : \alpha'$.
\end{proof}

\begin{theorem}[Preservation] \label{thm:preservation}
    If $\cdot \vdashL a : \alpha$ and $a \rightarrow a'$ then $\cdot \vdashL a' : \alpha$.
\end{theorem}
\begin{proof}
    By induction on $a \rightarrow a'$.
\end{proof}


Finally, we prove that every closed well-typed expression does normalize to a unique value. The proof, which uses a logical relation, is given in the supplementary material (\cref{app:strong-normalization}).

\begin{theorem}[Strong normalization, \cref{app:strong-normalization}] \label{thm:normalization}
    If $\cdot \vdashL a : \alpha$ then there exists a unique value $v$ (up to the usual notions of $\alpha$-equivalence) such that $a \rightarrow^\ast v$.
\end{theorem}


\subsection{Equivalence relation}
\label{sec:equivalence-relation}

In this section we define equivalence of open terms $\Delta \vdashL a_1 \lrequiv a_2 : \alpha$ as a logical relation. Later we will show that this implies equality of the underlying $\Zd$-linear maps in the categorical semantics.

We start by defining relations on values ($\ValRelation{\alpha}$) and closed expressions $(\ExpRelation{\alpha})$.
\begin{align*}
    \ValRelation{\R} &\triangleq \{(r,r) \mid r \in \Zd\} \\
    \ValRelation{\alpha_1 \oplus \alpha_2} &\triangleq \{([v_1,v_2],[v_1',v_2']) \mid
        (v_1,v_1') \in \ValRelation{\alpha_1} \wedge (v_2,v_2') \in \ValRelation{\alpha_2} \} \\
    \ValRelation{\alpha \lolli \alpha'} &\triangleq
        \{ (\lambda x_1.a_1, \lambda x_2.a_2) \mid \forall (v_1,v_2) \in \ValRelation{\alpha}.~ (a_1\{v_1/x_1\}, a_2\{v_2/x_2\}) \in \ExpRelation{\alpha'} \} \\
    \ExpRelation{\alpha} &\triangleq \left\{(a_1,a_2) \mid \cdot \vdashL a_1 : \alpha ~\wedge~ \cdot \vdashL a_2 : \alpha ~\wedge~ \forall v_1 v_2.~ (a_1 \rightarrow^\ast v_1) \wedge (a_2 \rightarrow^\ast v_2) \Rightarrow (v_1,v_2) \in \ValRelation{\alpha} \right\}
\end{align*}

Let $\delta$ be a value context mapping variables $x:\alpha$ in $\Delta$ to values of type $\alpha$. We write $\delta(a)$ for the usual capture-avoiding substitution of each value $\delta(x)$ for $x$ in $a$.
We can define an equivalence relation on value contexts $\delta$ as follows:
\begin{align*}
    \ValRelation{\Delta} &\triangleq \{(\delta_1,\delta_2) \mid \forall x \in \dom(\Delta).~ (\delta_1(x), \delta_2(x)) \in \ValRelation{\Delta(x)} \}
\end{align*}

\begin{definition}[$\lrequiv$]
    Let $\Delta \vdashL a_1 : \alpha$ and $\Delta \vdashL a_2 : \alpha$. We say $a_1$ and $a_2$ are equivalent, written $\Delta \vdashL a_1 \lrequiv a_2 : \alpha$, when, for all $(\delta_1, \delta_2) \in \ValRelation{\Delta}$ we have $(\delta_1(a_1), \delta_2(a_2)) \in \ExpRelation{\alpha}$.
\end{definition}


\begin{theorem}[Fundamental property, \cref{app:lambdaC-equivalence-relation}] \label{thm:lrequiv-fundamental-property}
    If $\Delta \vdashL a : \alpha$ then $\Delta \vdashL a \lrequiv a : \alpha$.
\end{theorem}

\subsection{Categorical Model}
\label{sec:lambdaC-categorical-model}

The categorical semantics for $\lambdaC$ is defined in terms of free finitely-generated $\Zd$-modules and $\Zd$-linear maps.
Let $\CCat$ be the category of finitely generated free $\Zd$-modules with ordered bases.
That is, the objects of $\CCat$ consist of an underlying set $A$ together with an ordered basis $b^A_1,\ldots,b^A_n$ of $A$, along with addition and scalar multiplication operations that satisfy a  list of axioms, with scalars taken from the ring $\Zd$.
The morphisms in $\CCat$ are $\Zd$-linear maps.

When $A$ has the form $\Zd^n$, we write $b_1,\ldots,b_n$ for its standard basis.
The \emph{rank} of $A$ is the size of its basis,
and the canonical isomorphism $\flatten{-} : A \to \Zd^{\rank{A}}$ maps $b^A_i$ to $b_i$. 
%
\cref{app:CCat-properties} establishes that $\CCat$ is a compact closed category with respect to the tensor product $\otimes$, and has biproducts $\oplus$.


Every $\CCat$-type $\alpha$ and typing context $\Delta$ corresponds to an object in $\CCat$.
\begin{align*}
  \begin{aligned}[t]
    \interpL{\R} &\triangleq \Zd \\
    \interpL{\alpha_1 \oplus \alpha_2} &\triangleq \interpL{\alpha_1} \oplus \interpL{\alpha_2} \\
    \interpL{\alpha_1 \lolli \alpha_2} &\triangleq \interpL{\alpha_1} \lolli \interpL{\alpha_2}
  \end{aligned}
  \qquad\qquad
  \begin{aligned}[t]
    \interpL{\cdot} &\triangleq \Zd \\
    \interpL{\Delta,x:\alpha} &\triangleq \interpL{\Delta} \otimes \interpL{\alpha}
  \end{aligned}
\end{align*}


For every well-typed expression $\Delta \vdash a : \alpha$ we associate a $\Zd$-linear transformation $\interpL{a} \in \CCat(\interpL{\Delta},\interpL{\alpha})$, defined in \cref{fig:cat-semantics}, on basis elements by induction on the typing rules of \cref{fig:typing-rules-C-simple}. Note that constants $r \in \Zd$ correspond to generalized elements $\interpL{r} \in \CCat(\Zd,\Zd)$ defined by $x \mapsto rx$. Between any two $\Zd$-modules in $\CCat$ there is a zero morphism $x \mapsto \zero$ where $\zero$ is the zero element.

\begin{figure}
    \centering
\begin{align*}
    \interpL{x}(s)
        &\triangleq s
        \\
    \interpL{\letin{x}{a}{a'}}(s \otimes s')
        &\triangleq \interpL{a'}\left(s' \otimes \interpL{a}(s)\right)
    \\
    \interpL{r}(s)
        &\triangleq r s
    \\
    \interpL{a \cdot a'}(s \otimes s')
        &\triangleq \interpL{a}(s) \cdot \interpL{a'}(s')
    \\
    \interpL{\zero}(s)
        &\triangleq 0
    \\
    \interpL{a_1+a_2}(s)
        &\triangleq \interpL{a_1}(s) + \interpL{a_2}(s)
    \\
    \interpL{[a_1,a_2]}(s)
        &\triangleq \interpL{a_1}(s) \oplus \interpL{a_2}(s)
    \\
    \interpL{\caseof{a}{\inl{x_1} \rightarrow a_1 \mid \inr{x_2} \rightarrow a_2}}(s \otimes s')
        &\triangleq \interpL{a_1}(s' \otimes c_1) + \interpL{a_2}(s' \otimes c_2)
        \tag{where $\interpL{a}(s)=c_1 \oplus c_2$}
    \\
    \interpL{\lambda x:\alpha.a}(s)
        &\triangleq
        b \mapsto \left(\interpL{a}(s \otimes b)\right)
    \\
    \interpL{a_1 a_2}(s_1 \otimes s_2)
        &\triangleq \left(\interpL{a_1}(s_1)\right)\left(\interpL{a_2}(s_2)\right)
\end{align*}
    \caption{Categorical semantics of $\CCat$-expressions $\Delta \vdash a : \alpha$ as $\Zd$-linear maps 
    $\interpL{a} \in \CCat(\interpL{\Delta},\interpL{\alpha})$, up to isomorphism of $\interpL{\Delta}$.
    For example, in the rule for $\letin{x}{a}{a'}$ typed by $\Delta,\Delta'$ where $\Delta \vdashL a : \alpha$ and $\Delta',x:\alpha \vdashL a' : \alpha'$, we assume we have $s \in \interpL{\Delta}$ and $s' \in \interpL{\Delta'}$.
    }
    \Description{A recursive definition of the categorical semantics of \(\PCat\)-expressions into morphisms in \(\PCat\).}
    \label{fig:cat-semantics}
\end{figure}

For a value map $\delta$ of $\Delta$, we define a morphism $\interpL{\delta}_{\Delta} \in \CCat(\Zd,\interpL{\Delta})$ by induction on $\Delta$:
\begin{align*}
    \interpL{\delta}_{\cdot}(s) &\triangleq s
    &
    \interpL{\delta}_{\Delta',x:\alpha}(s) &\triangleq \interpL{\delta}_{\Delta'}(s) \otimes \interpL{\delta(x)}(1)
\end{align*}

\subsubsection{Soundness}
Next we will prove that if $\Delta \vdash a_1 \lrequiv a_2 : \alpha$, then $\interpL{a_1} = \interpL{a_2}$.
We sketch the proof here and give the full details in \cref{app:lambdaC-soundness}.

\begin{theorem} \label{thm:categorical-equivalence-sound}
    If $\Delta \vdash a_1 \lrequiv a_2 : \alpha$ then $\interpL{a_1}=\interpL{a_2}$.
\end{theorem}
\begin{proof}[Proof sketch]
    First we prove two key lemmas:
    \begin{enumerate}
        \item (\cref{thm:step-sound}) If $\cdot \vdash a : \alpha$ and $a \rightarrow a'$, then $\interpL{a} = \interpL{a'}$.
        \item (\cref{lem:val-relation-sound}) If $(v_1,v_2) \in \ValRelation{\alpha}$ then $\interpL{v_1}=\interpL{v_2}$.
    \end{enumerate}

    To show $\interpL{a_1}=\interpL{a_2}$, it suffices to show $\interpL{a_1}(g)=\interpL{a_2}(g)$ for all basis elements $g \in \interpL{\Delta}$. So our first step is to show that for each such $g$ there exists a value context $\up{g}$ of $\Delta$
    such that $\interpL{\up{g}}(1)=g$ (\cref{lem:contexts-complete}). 
    Next, we show that composing $\interpL{a_i}$ with $\interpL{\up{g}}$ is the same as $\interpL{\up{g}(a_i)}$ (\cref{lem:substitution-composition}), so it suffices to prove  $\interpL{\up{g}(a_1)} = \interpL{\up{g}(a_2)}$.
    
    Because $\Delta \vdash a_1 \lrequiv a_2 : \alpha$, we know that $(\up{g}(a_1),\up{g}(a_2)) \in \ExpRelation{\alpha}$, so it must be the case that $\interpL{\up{g}(a_1)} = \interpL{\up{g}(a_2)}$ (\cref{lem:environment-interp-append}).
\end{proof}

\subsubsection{Completeness}
Finally, we prove that every vector (including linear maps) in $\CCat$ can be expressed in $\lambdaC$, and that $\lrequiv$ is complete with respect to the categorical semantics.

\begin{theorem}[Completeness of linear maps] \label{thm:lambdaC-completeness}
    For any $a \in \interpL{\alpha}$, there exists a $\lambdaC$-expression $\up{a}$ such that
    $\cdot \vdashC \up{a} : \alpha$ and $\interpL{\up{a}}(1) = a$.
\end{theorem}
\begin{proof}
    In \cref{app:lambdaC-completeness}, \cref{lem:values-complete} we prove the corresponding statements regarding basis elements. This extends easily to arbitrary vectors: we can always write $a$ as a finite linear combination $s_1 b_1 + \cdots + s_m b_m$ of basis elements,
    in which case we define $\up{a}$ as $s_1 \cdot \up{b_1} + \cdots + s_m \cdot \up{b_m}$.
\end{proof}

\begin{theorem}[Completeness of $\lrequiv$] \label{thm:lambdaC-lrequiv-completeness}
    If $\Delta \vdashC a_1,a_2 : \alpha$ such that $\interpL{a_1} = \interpL{a_2}$, then $\Delta \vdashC a_1 \lrequiv a_2$.
\end{theorem}
\begin{proof}
    See \cref{app:lambdaC-completeness-2}, \cref{app:lambdaC-lrequiv-completeness}.
\end{proof}

\subsection{The Symplectic Form in \(\lambdaC\)}

Recall from \cref{sec:background} that the condensed encoding of a projective Clifford is a pair of functions $(\mu,\psi)$ where $\mu$ is a $\Zd$-linear map and $\psi$ is a symplectomorphism---a linear map respecting the symplectic form $\omega : \Zd^{2n} \otimes \Zd^{2n} \to \Zd$. We can now formally define the symplectic form in $\lambdaC$.

We start by picking out a subset of $\CCat$-types corresponding to the types for which $\omega$ is defined.
\begin{align*}
    \sigma &::= \R \oplus \R \mid \sigma_1 \oplus \sigma_2
        \tag{symplectic types}
\end{align*}
Clearly, every such \emph{symplectic type} is also a linear type $\alpha$, and each can be given a symplectic form.

\begin{lemma}
    For each $\sigma$ there exists a linear expression $\cdot \vdashL \omega_\sigma : \sigma \lolli \sigma \lolli \R$ satisfying
    \begin{align*}
        \vdashL \omega_{\R \oplus \R} [r_x, r_z] [r_x', r_z'] &\lrequiv r_x' r_z - r_x r_z' \\
        \vdashL \omega_{\sigma_1 \oplus \sigma_2} [v_1, v_2] [v_1', v_2'] &\lrequiv \omega_{\sigma_1} v_1 v_1' + \omega_{\sigma_2} v_2 v_2'
    \end{align*}
\end{lemma}
\begin{proof}
    We define $\omega$ by induction on $\sigma$ as follows:    
    {\small \begin{align*}
        \omega_{\R \oplus \R} &\triangleq \lambda x.~\lambda x'.~
            \caseof{x}{\begin{aligned}
                \inl{x_x} &\rightarrow \caseof{x'}{\begin{aligned}
                    \inl{x_x'} &\rightarrow \zero \\
                    \inr{x_z'} &\rightarrow -x_x x_z'
                \end{aligned} } \\
                \inr{x_z} &\rightarrow \caseof{x'}{\begin{aligned}
                    \inl{x_x'} &\rightarrow x_z x_x' \\
                    \inr{x_z'} &\rightarrow \zero
                \end{aligned} }
            \end{aligned} }
        \\
        \omega_{\sigma_1 \oplus \sigma_2} &\triangleq \lambda x.~\lambda x'.~
            \caseof{x}{\begin{aligned}
                \inl{x_1} &\rightarrow \caseof{x'}{\begin{aligned}
                    \inl{x_1'} &\rightarrow \omega_{\sigma_1} x_1 x_1' \\
                    \inr{x_2'} &\rightarrow \zero
                \end{aligned} } \\
                \inr{x_2} &\rightarrow \caseof{x'}{\begin{aligned}
                    \inl{x_1'} &\rightarrow \zero \\
                    \inr{x_2'} &\rightarrow \omega_{\sigma_2} x_2 x_2'
                \end{aligned} }
            \end{aligned} }
            \qedhere
    \end{align*}
    \normalsize}
\end{proof}

\section{A Calculus for Projective Cliffords}
\label{sec:lambdaPC}

Now that we have a type system of $\Zd$-modules, we can use it to build up the type system for $\lambdaP$. In particular, closed terms of $\lambdaP$ correspond to Paulis in $\Qdn$ (\cref{sec:background:qdn}), and open terms correspond to condensed encodings of projective Cliffords.



With that in mind, $\PCat$-types $\tau$ are generated from single-\qudit Paulis ($\Pauli$) and $\ptensor$.
\begin{align*}
    \tau &::= \Pauli \mid \tau_1 \ptensor \tau_2
        \tag{$\PCat$-types}
\end{align*}
Every $\tau$ corresponds to a symplectic type $\sigma$.
\begin{align*}
    \overline{\Pauli} &\triangleq \R \oplus \R
    &
    \overline{\tau_1 \ptensor \tau_2} &\triangleq \overline{\tau_1} \oplus \overline{\tau_2}
\end{align*}
The syntax of $\PCat$-expressions is given by the following grammar:
\begin{align*}
    t &::= x \mid \letin{x}{t}{t'} \\
        &\mid a \mid \config{a} t \mid t_1 \cprod t_2 \mid \pow{t}{r} \\
        &\mid \caseof{t}{\PauliX \rightarrow t_x \mid \PauliZ \rightarrow t_z} \\
        &\mid \inl{t} \mid \inr{t} \mid \caseof{t}{\inl{x_1} \rightarrow t_1 \mid \inr{x_2} \rightarrow t_2}
        \tag{$\PCat$-expressions}
\end{align*}
Closed $\PCat$-expressions $t$ of type $\tau$ normalize to a pair of a $\lambdaC$ value $v \in \Val{\overline{\tau}}$ and a phase $r \in \R$, written $\config{r}{v}$. Intuitively, these correspond to Paulis $\zeta^r \Delta_{\underline{v}}$.
Every $\CCat$-expression $a$ is also a $\PCat$-expression with implicit phase $\config{0}$. The expression $\config{a}t$ adds a phase $a$ of $\CCat$-type $\R$ to the $\PCat$-expression $t$. The operator $t_1 \cprod t_2$ implements the condensed product (\cref{sec:encodings}), while $\pow{t}{r}$ scales both the Pauli representation and phase by the scalar $r$, as in $(\zeta^s \Delta_v)^r$.

\subsection{Typing Rules}

When a $\PCat$-expression has a free variable as in $x:\tau \vdashP t : \tau'$, it corresponds to a condensed encoding $\interpP{t} = (\mu,\psi)$. We can be explicit about $\psi$ in particular: for every such $t$ we can define an $\CCat$-expression $x:\overline{\tau} \vdashL \psiof{t} :\overline{\tau'}$, defined in \cref{fig:psiof}, that ignores the phase of $t$ and satisfies $\interpL{\psiof{t}} = \psi$.
Later we will prove that $\psiof{t}$ respects the symplectic form of its input variable $x$.

\begin{figure}
    \centering
    \begin{align*}
    \psiof{x} &\triangleq x \\
    \psiof*{\letin{x}{t}{t'}} &\triangleq \letin{x}{\psiof{t}}{\psiof*{t'}} \\
    \psiof{a} &\triangleq a \\
    \psiof*{\config{a}{t}} &\triangleq \psiof{t} \\
    \psiof*{t_1 \cprod t_2} &\triangleq \psiof{t_1} + \psiof{t_2} \\
    \psiof*{\pow{t}{r}} &\triangleq r \cdot \psiof{t} \\
    \psiof*{\caseof{t}{\PauliX \rightarrow t_x \mid \PauliZ \rightarrow t_z}}
        &\triangleq \caseof{\psiof{t}}{\inl{x_1} \rightarrow x_1 \cdot \psiof{t_x} 
                                        \mid
                                        \inr{x_2} \rightarrow x_2 \cdot \psiof{t_z}} \\
    \psiof*{\iota_i{t}} &\triangleq \iota_i{\psiof{t}} \\
    \psiof*{\caseof{t}{\inl{x_1} \rightarrow t_1 \mid \inr{x_2} \rightarrow t_2}}
        &\triangleq \caseof{\psiof{t}}{\inl{x_1} \rightarrow \psiof{t_1}
                                        \mid
                                        \inr{x_2} \rightarrow \psiof{t_2}}
\end{align*}
    \caption{Projecting out the non-phase component of a \(\PCat\)-expression to form a \(\CCat\)-expression.}
    \Description{A definition of a function \(\psiof{t}\) that takes \(\PCat\)-expressions corresponding to pairs \((\mu,\psi)\) to the linear morphism \(\psi\). The definition is given recursively on the expression \(t\).}
    \label{fig:psiof}
\end{figure}

The typing judgment for $\lambdaP$ has the form $\Theta \vdashP t : \tau$. The judgment is made up of two parts: a linearity check, which ensures that both the phase ($\mu$) and vector ($\psi$) components of the expression are linear in their inputs; and a symplectomorphism check, which ensures that the $\psi$ component respects the symplectic form.
\begin{align*}
    \inferrule*
    {\Theta \vdashPL t : \tau \\
     \Theta \vdashS t : \tau
    }
    { \Theta \vdashP t : \tau }
\end{align*}
The linearity check is straightforward, and the rules are shown in \cref{fig:lambdaPC-L-typing-rules}.

\begin{figure}
    \centering
    \[ \begin{array}{c}
        \inferrule*[right=L-var]
            {\Theta=x:\tau}
            {\Theta \vdashPL x : \tau}
        \qquad \qquad
        \inferrule*[right=L-let]
            {\Theta_1 \vdashPL t : \tau \\ \Theta_2,x:\tau \vdashPL t' : \tau'}
            {\Theta_1,\Theta_2 \vdashPL \letin{x}{t}{t'} : \tau'}
        \\ \\
        \inferrule*[right=L-phase]
            {\overline{\Theta} \vdashL a : \R \\
             \Theta \vdashPL t : \tau
            }
            {\Theta \vdashPL \config{a} t : \tau}
        \qquad \qquad
        \inferrule*[right=L-$\CCat$]
            {\overline{\Theta} \vdashL a : \overline{\tau}}
            {\Theta \vdashPL a : \tau}
        \\ \\
        \inferrule*[right=L-Pauli-E]
            {\Theta_1 \vdashPL t : \Pauli \\
             \Theta_2 \vdashPL t_x : \tau \\
             \Theta_2 \vdashPL t_z : \tau
            }
            {\Theta_1,\Theta_2 \vdashP \caseof{t}{\PauliX \rightarrow t_x \mid \PauliZ \rightarrow t_z} : \tau}
        \\ \\
        \inferrule*[right=L-$\cprod$]
            {\Theta \vdashPL t_1 : \tau \\ \Theta \vdashPL t_2 : \tau}
            {\Theta \vdashPL t_1 \cprod t_2 : \tau}
        \qquad \qquad
        \inferrule*[right=L-pow]
            { \Theta_1 \vdashPL t : \tau \\ \overline{\Theta_2} \vdashL a : \Zd }
            { \Theta_1,\Theta_2 \vdashPL \pow{t}{a} : \tau }
        \\ \\
        \inferrule*[right=L-$\ptensor$-I1]
            {\Theta \vdashPL t : \tau_1}
            {\Theta \vdashPL \inl{t} : \tau_1 \ptensor \tau_2}
        \qquad \qquad
        \inferrule*[right=L-$\ptensor$-I2]
            {\Theta \vdashPL t : \tau_2}
            {\Theta \vdashPL \inr{t} : \tau_1 \ptensor \tau_2}
        \\ \\
        \inferrule*[right=L-$\ptensor$-E]
            {\Theta_1 \vdashPL t : \tau_1 \ptensor \tau_2 \\
             \Theta_2,x_i:\tau_i \vdashPL t_i : \tau'
            }
            {\Theta_1,\Theta_2 \vdashPL \caseof{t}{\inl{x_1} \rightarrow t_1 \mid \inr{x_2} \rightarrow t_2} : \tau'}
    \end{array}\]
    \caption{Linearity typing rules for $\lambdaP$ expressions.}
    \Description{A set of ten typing rules for expressions in \(\lambdaP\).}
    \label{fig:lambdaPC-L-typing-rules}
\end{figure}

The symplectomorphism check is trivial when $\Theta$ is empty, and is defined in terms of the $\lambdaC$ equivalence relation when $\Theta$ is non-empty.
\begin{align*}
        \inferrule*[right=S0]
        {~}
        {\cdot \vdashS t : \tau}
    \qquad
        \inferrule*[right=S1]
            { x_1:\tau, x_2:\tau \vdash \omega(t^\psi\{x_1/x\}, t^\psi\{x_2/x\}) \lrequiv \omega(x_1,x_2) }
            { x:\tau \vdashS t : \tau }
\end{align*}
Note that the symplectomorphism check explicitly does not allow open expressions with more than one free variable. As we saw in the introduction, Cliffords on multi-qudit systems are captured by a variable $x : \tau_1 \ptensor \tau_2$, which is a coproduct in the category of condensed encodings, and not a tensor product in the sense of linear logic.
We hypothesize that multiple linear variables $x_1:\tau_1,x_2:\tau_2 \vdashP t : \tau$ might correspond not to a multi-qudit operation, but instead to operations in higher levels of the Clifford hierarchy---see \cref{sec:future-work} for more discussion.

It can be useful to derive typing rules that combine the linearity and symplectomorphism checks.

\begin{lemma}
    The following typing rules are valid:
\[\begin{array}{c}
        \inferrule*[right=$\PCat$-Pauli-E]
            {\Theta \vdashP t : \Pauli \\
             \cdot \vdashP t_x : \tau \\
             \cdot \vdashP t_z : \tau \\
             \cdot \vdashP \symplecticform[\tau]{t_z}{t_x} \lrequiv 1
            }
            {\Theta \vdashP \caseof{t}{\PauliX \rightarrow t_x \mid \PauliZ \rightarrow t_z} : \tau}
        \\ \\
        \inferrule*[right=$\PCat$-$\ptensor$-E]
            {\Theta \vdashP t : \tau_1 \ptensor \tau_2 \\
             x_i:\tau_i \vdashP t_i : \tau' \\
             x_1:\tau_1,x_2:\tau_2 \vdashP \symplecticform[\tau']{t_1}{t_2} \lrequiv \zero
            }
            {\Theta \vdashP \caseof{t}{\inl{x_1} \rightarrow t_1 \mid \inr{x_2} \rightarrow t_2} : \tau'}
\end{array}\]
\end{lemma}
\begin{proof}
    For the first rule, we start by checking that, for $\Theta = q_1:\Pauli,q_2:\Pauli$: 
    \[
        \Theta \vdashP \omega(\caseof{q_1}{\PauliX \rightarrow t_x \mid \PauliZ \rightarrow t_z}, \caseof{q_2}{\PauliX \rightarrow t_x \mid \PauliZ \rightarrow t_z}) \lrequiv \omega(q_1,q_2)
    \]
    which follows from instantiating $q_1$ and $q_2$ with arbitrary values $[r_x,r_z]$ and $[r_x',r_z']$ respectively.
    Then we can check that
    \begin{align*}
        q_1':\tau',q_2':\tau'
        &\vdash \omega\left(
            \caseof{t\{q_1'/q'\}}{\PauliX \rightarrow t_x \mid \PauliZ \rightarrow t_z},
            \caseof{t\{q_2'/q'\}}{\PauliX \rightarrow t_x \mid \PauliZ \rightarrow t_z}
        \right) \\
        &\lrequiv \omega(t\{q_1'/q\}, t\{q_2'/q\}) 
        \lrequiv \omega(q_1',q_2')
    \end{align*}

  The proof of the second rule is similar to the first.
\end{proof}

\begin{example} \label{eqn:h-example-well-typed}
    The projective Clifford $[H]$ corresponding to the qubit ($d=2$) Hadamard matrix can be expressed as $\caseof{x}{\PauliX \rightarrow [0,1] \mid \PauliZ \rightarrow [1,0]}$. It is well-typed thanks to the $\PCat-\Pauli-E$ rule because of the following symplectic form condition:
    \begin{align*}
        \symplecticform[\Pauli]{[1,0]}{[0,1]}
        &= \symplecticform[\R \oplus \R]{[1,0]}{[0,1]}
        = -1 = 1 \mod{2}
    \end{align*}
\end{example}

\subsection{Operational Semantics}

Next we will define operational semantics rules for closed Pauli expressions. The $\beta$-reduction rules are shown in \cref{fig:operational-semantics-Pauli}, and we also allow reduction under call-by-value evaluation contexts.

Three of these rules involve adding an extra phase $k$ during the evaluation, which comes from the composition and normalization rules of condensed encodings, specifically from the difference between $\Zd$ and $\Zd'$. As in \cref{sec:encodings}, for a value $r' \in \Zd'$ we write $\sgn{r'}$ for $\tfrac{1}{d}(r'-\underline{\overline{r'}}) \in \mathbb{Z}_{d'/d} \subseteq \{0,1\}$, where $\overline{r'}=r' \mod d$. In addition, for $v' \in \Zd'^n$ we write $\sgn{v'}$ for $\frac1d \omega'(v,\overline{\underline{v}}) \in \mathbb{Z}_{d'/d}$.

When $d$ is odd, $r'=r$ and so all the extra phases are trivial. When $d=2$, the phases for $\text{pow}$ and $\text{case}$ also become $0$, as for any $b_1,b_2 \in \mathbb{Z}_2$, we have $\underline{b_1}~\underline{b_2}=\underline{b_1 b_2}$.

\begin{example}
    Consider the qubit $[H]$ example from \cref{eqn:h-example-well-typed}. If we substitute the encoding of $Y=\Delta_{[1,1]}$ for $x$, we should obtain the encoding of $H Y H = -Y$.
    As a first step,
    \begin{align*}
        \caseof{\config{0}[1,1]}{\PauliX \rightarrow [0,1] \mid \PauliZ \rightarrow [1,0]}
        &\rightarrow_{\beta} \config{\sgn{1}} \pow{[1,0]}{1} \cprod \pow{[0,1]}{1} \\
        &= \pow{[1,0]}{1} \cprod \pow{[0,1]}{1}
    \end{align*}
    because $\sgn{1}=0$.
    Next, we can see that $\pow{v}{1} \rightarrow \config{0}{v}$ because the coefficient $\sgn{\underline{1}~\underline{v}}$ is always equal to $0$. 
    Finally, we can see that
    \begin{align*}
        \config{0}[1,0] \cprod \config{0}[0,1]
        &\rightarrow_\beta \config{k}{[1,1]}
        &\text{where}~k&=\sgn{\omega'(\underline{[1,0]},\underline{[0,1]})} + \sgn{\underline{[1,0]}+\underline{[0,1]}}
    \end{align*}
    In the expression of $k$ above, the second component is $\sgn{[1,1]}=0$, while the first is
    \begin{align*}
        \sgn{\omega'(\underline{[1,0]},\underline{[0,1]})} 
        &= \sgn{-1 \mod{d'}} = \sgn{3} = 1
    \end{align*}
    since $d'=4$. Thus our example normalizes to $\config{1}[1,1]$, which corresponds to the Pauli $i^1 X^1 Z^1 = -Y$.
\end{example}

\begin{figure}
    \centering
    \begin{align*}
        \letin{x}{\config{r}{v}}{t'}
            \rightarrow_\beta
            &\config{r} t'\{v/x\} \\
        \config{r'}\left( \config{r}{v} \right)
            \rightarrow_\beta
            &\config{r' + r} v \\
        (\config{r_1} v_1) \cprod (\config{r_2} v_2)
            \rightarrow_\beta
            &\config{r_1 + r_2 + k} (v_1+v_2) \\
            & k = \tfrac{d}{2}\left( \divd{\omega'(\underline{v_1},\underline{v_2})} + \divd{\underline{v_1} + \underline{v_2}} \right) \\
        \pow{\config{r}{v}}{r'}
            \rightarrow_\beta
            &\config{r' r + k} (r' \cdot v) \\
            &k = \tfrac{d}{2}\sgn{\underline{r'}~\underline{v}} \\
        \caseof{\config{r}{[r_x, r_z]}}{\PauliX \rightarrow t_x \mid \PauliZ \rightarrow t_z}
            \rightarrow_\beta
            &\config{r + k}
                \pow{t_z}{r_z} \cprod \pow{t_x}{r_x} \\
            &k = \tfrac d2 \sgn{\underline{r_x}~\underline{r_z}} \\
        \iota_i{\config{r}{v}}
            \rightarrow_\beta
            &\config{r}{\iota_i(v)} \\
        \caseof{\config{r}{[v_1,v_2]}}
                {\inl{x_1} \rightarrow t_1 \mid \inr{x_2} \rightarrow t_2}
            \rightarrow_\beta
            &\config{r} t_1\{v_1/x_1\} \cprod t_2\{v_2/x_2\}
            \\
        v \rightarrow_\eta &\config{0}{v}
    \end{align*}
    \caption{$\beta$-reduction rules for closed $\lambdaP$ expressions.}
    \Description{A set of eight beta reduction rules for \(\lambdaP\) expressions. Three of the rules add phase coefficients \(k\) which add addition plus or minus signs to the phase.}
    \label{fig:operational-semantics-Pauli}
\end{figure}

\begin{lemma} The following properties hold of $\lambdaP$:

    \begin{description}
        \item[Progress:] If $\cdot \vdashP t : \tau$ then either $t$ is a normal form of the form $\config{r}{v}$ for a $\CCat$-value $v$,
    or there exists some $t'$ such that $t \rightarrow t'$.
        \item[Preservation:] If $\cdot \vdashP t : \tau$ and $t \rightarrow t'$, then $\cdot \vdashP t' : \tau$.
        \item[Normalization:] If $\cdot \vdashP t : \tau$ there is a unique normal form $\config{r}{v}$ such that $t \rightarrow^\ast \config{r}{v}$.
    \end{description}
\end{lemma}
\begin{proof}
    Straightforward by induction.
\end{proof}

\subsection{The Category of Projective Cliffords}

In this section, we define the \textit{symplectic category} $\SCat$ and the \textit{Pauli category} $\PCat$ that arise in the categorical semantics of $\lambdaP$.

The \emph{symplectic category}\footnote{Note that our use of ``symplectic category'' varies slightly from Weinstein symplectic categories~\citep{weinstein2010symplectic} since we only include $\Zd$-linear functions that respect the symplectic form.} has objects $(V,\omega)$ where $V\in\obj{\CCat}$ and $\omega:V\otimes V\to\Zd$ is a symplectic form on $V$. Its morphisms $\psi:(V_1,\omega_1)\to(V_2,\omega_2)$ are \textit{symplectic morphisms}, that is, $\Zd$-linear maps that respect the symplectic form:
\begin{align*}
    \omega_1(v,v')=\omega_2(\psi(v),\psi(v'))&&\textrm{for all }v,v'\in V_1
\end{align*}
The composition of $\psi$ and $\psi'$ is just their composition in $\CCat$. Symplectic morphisms are necessarily injective, so if $\psi\in\CCat(V_1,V_2)$ then $\rank{V_1}\leq\rank{V_2}$. 

Objects of $\CCat$ of the form 
$(\Zd \oplus \Zd)^n$
have a canonical symplectic form $\omega$ defined in terms of their canonical ordered basis, which we will write $b^x_1,b^z_1,\ldots,b^x_n,b^z_n$.
Since this is the form of objects that will appear in the categorical semantics as the interpretation of symplectic types $\sigma$, such objects comprise the subcategory of $\SCat$ we are most concerned with.

We can see $\SCat$ as a subcategory of $\CCat$ by forgetting the symplectic form $\omega$ on each object $(V,\omega)\in\obj{\SCat}$. The underlying object of the coproduct $(V_1,\omega_1)\oplus(V_2,\omega_2)$ in $\SCat$ is the biproduct $V_1 \oplus V_2$ in $\CCat$. Note that while we still denote it using the $\oplus$ symbol, $\oplus$ is not in general a (bi)product in $\SCat$.

We now define the category $\PCat$ whose morphisms are condensed encodings of projective Cliffords.
For each object $V \in \obj{\SCat}$, we introduce a symbol $\Pobj{V} \in \obj{\PCat}$ as well as a special unit object $\Punit$.
\begin{align*}
    \obj{\PCat} &\triangleq \{\Pobj{V} \mid V \in \SCat \} \cup \{\Punit\}
\end{align*}
The morphisms of $\PCat$ are then defined as follows:
\begin{align*}
  \begin{aligned}
    \PCat(\Pobj{V},\Pobj{V'}) &\triangleq \CCat(V,\Zd) \times \SCat(V,V')
  \end{aligned}
  \qquad
  \begin{aligned}
    \PCat(\Punit,\Pobj{V}) &\triangleq \CCat(\Zd,\Zd) \times \CCat(\Zd,V) \\
    \PCat(\Punit,\Punit) &\triangleq \{ \idmorph\Zd \}
  \end{aligned}
\end{align*}
Intuitively, if $V$ is the type of $n$ qu$d$it phase space $(\Zd \oplus \Zd)^{n}$, then $\Pobj{V}$ should be thought of as the $\cprod$-closed subset $\Qdn\subseteq\Pauligroup$ defined in \cref{sec:encodings}. A morphism $\gamma : \Pobj{V} \to \Pobj{V}$ can be thought of as the encoding $(\mu,\psi)$ of a projective Clifford $[U]\in\Pcliffordgroup$. 
On the other hand, morphisms $\gamma : \Punit \to \Pobj{V}$ can be thought of as the global elements of $\Pobj{V}$, that is, the set of Paulis in $\Qdn$.

The identity morphism in $\PCat(\Pobj{V},\Pobj{V})$ is $(\zero,\idmorph{V})$, and in $\PCat(\Punit,\Punit)$ is $\idmorph\Zd$.

The composition of two morphisms $(\mu_2,\psi_2) \in \PCat(\Pobj{V},\Pobj{V'})$ with $(\mu_1,\psi_1) \in \PCat(A,\Pobj{V})$ (for $A$ an arbitrary object of $\PCat$) is defined as
\[
    (\mu_2,\psi_2) \circ (\mu_1,\psi_1) = (\mu_3,\psi_2 \circ \psi_1) \in \PCat(A,\Pobj{V'})
\]
where $\mu_3 \in \CCat(A, \Zd)$ is defined linearly on its basis elements by
    \begin{align*}
        \mu_3(b) &\triangleq \mu_1(b) + \mu_2(\psi_1(b)) + \kappa^{\psi_2}\flatten{\psi_1(b)}
    \end{align*}
and where $\kappa^\psi$ is the function (not a linear map) from $\Zd^{\rank{V}}$ to $\Zd$ defined in \cref{sec:encodings}:
    \begin{align*}
        \kappa^\psi(v)
            &\triangleq \frac{1}{d} \sum_{i=1}^n \left(
                \underline{x_i}~\underline{z_i}
                (1 + \omega'(\underline{\psi(b_i^x)}, \underline{\psi(b_i^z)}))
                + \underline{x_i} \omega'(\underline{\psi(b_i^x)}, \underline{\psi(v)})
                + \underline{z_i} \omega'(\underline{\psi(b_i^z)}, \underline{\psi(v)})
                \right) \\
        &\text{where}~v=x_1 \oplus z_1 \oplus \cdots \oplus x_n \oplus z_n
    \end{align*}

\subsection{Categorical Semantics of \(\lambdaP\)}
\label{sec:lambdaP-categorical-semantics}

\cref{fig:cat-semantics-pauli} identifies every $\PCat$ expression $\Theta \vdashPL t : \tau$ with a pair of morphisms 
\[
    \interpP{t} = (\mu,\psi) 
    \in \PCat'(\Theta,\tau) \triangleq \CCat(\interpL{\overline{\Theta}}, \Zd) \times \CCat(\interpL{\overline{\Theta}}, \interpL{\overline{\tau}})
\]

\begin{figure}
    \centering
    \begin{align*}
        \interpP{x} &\triangleq (\zero, id) \\
        \interpP{\letin{x}{t}{t'}} &\triangleq \interpP{t'} \circ (\interpP{t} \otimes \idmorph{~}) \\
        \interpP{a} &\triangleq (\zero, \interpL{a}) \\
        \interpP{\config{a}{t}} &\triangleq \config{\interpL{a}} \interpP{t} \\
        \interpP{t_1 \cprod t_2} &\triangleq \interpP{t_1} \cprod \interpP{t_2} \\
        \interpP{\pow{t}{a}} &\triangleq \pow{\interpP{t}}{\interpL{a}} \\
        \interpP{\caseof{t}{\PauliX \rightarrow t_x \mid \PauliZ \rightarrow t_z}}
            &\triangleq \left(\interpP{t_x} \boxplus \interpP{t_z} \right) \circ \interpP{t} \\
        \interpP{\iota_i(t)} &\triangleq (\zero, \iota_i) \circ \interpP{t} \\
        \interpP{\caseof{t}{\inl{x_1} \rightarrow t_1 \mid \inr{x_2} \rightarrow t_2}}
            &\triangleq \left(\interpP{t_1} \boxplus \interpP{t_2} \right) \circ \interpP{t}
    \end{align*}
    \caption{Categorical semantics of $\lambdaP$ expressions $\Theta \vdashPL t : \tau$ as morphisms in $ \CCat(\interpL{\overline{\Theta}}, \Zd) \times \CCat(\interpL{\overline{\Theta}}, \interpL{\overline{\tau}})$.}
    \Description{Recursive definition of a function that maps \(\lambdaP\) expressions into morphisms.}
    \label{fig:cat-semantics-pauli}
\end{figure}

Composition of pairs of morphisms $(\mu,\psi)$ in \cref{fig:cat-semantics-pauli} implicitly refers to  composition in $\PCat$. \cref{fig:cat-semantics-pauli} uses constructions on pairs of morphisms that mirror their respective programming abstractions:

\begin{enumerate}
    \item If $(\mu,\psi) \in \PCat'(\Theta,\tau)$ and $\mu' \in \CCat(\interpL{\overline{\Theta}},\Zd)$,
    then $\config{\mu'}(\mu,\psi) \triangleq (\mu'+\mu, \psi) \in \PCat'(\Theta, \tau)$.

    \item If $(\mu_i,\psi_i) \in \PCat'(\Theta,\tau)$, then $(\mu_1,\psi_1) \cprod (\mu_2,\psi_2) \triangleq (\mu_0,\psi_1+\psi_2) \in \PCat'(\Theta,\tau)$, where:
    \[
        \mu_0(b) \triangleq \mu_1(b) + \mu_2(b)
            + \tfrac{d}{2}\left( \divd{\omega'(\underline{\psi_1(b)},\underline{\psi_2(b)})}
                               + \divd{\underline{\psi_1(b)} + \underline{\psi_2(b)}}
                          \right)    
    \]

    \item If $(\mu,\psi) \in \PCat'(\Theta,\tau)$ and
        $a \in \CCat(\interpL{\overline{\Theta'}}, \Zd)$,
        then $\pow{(\mu,\psi)}{a} \triangleq (\mu_0,\psi_0) \in \PCat'((\Theta,\Theta'), \tau)$, where:
    \begin{align*}
        \psi_0(b \otimes b') &\triangleq a(b') \psi(b) 
      &
        \mu_0(b \otimes b') &\triangleq a(b') \mu(b) + \tfrac{d}{2} \divd{\underline{a(b')} ~\underline{\psi(b)}}
    \end{align*}

    \item For $\mu_i \in \CCat(\alpha_i \otimes \alpha, \Zd)$
        and $\psi_i \in \CCat(\alpha_i \otimes \alpha, \alpha')$, define
        $(\mu_1,\psi_1) \boxplus (\mu_2,\psi_2) \triangleq (\mu_1 \boxplus \mu_2,\psi_1 \boxplus \psi_2)$, where:
    \begin{align*}
        f_1 \boxplus f_2((b_1 \oplus b_2) \otimes b)
            &\triangleq f_1(b_1 \otimes b) + f_2(b_2 \otimes b)
    \end{align*}
\end{enumerate}

\begin{theorem}
    If $\Theta \vdashP t : \tau$ then $\interpP{t} \in \PCat(\interpP{\Theta}, \interpP{\tau})$ where
\begin{align*}
    \interpP{\tau} &\triangleq \Pobj{\interpL{\overline{\tau}}}
    &
    \interpP{\cdot} &\triangleq \Punit 
    &
    \interpP{x:\tau} &\triangleq \interpP{\tau}
\end{align*}
\end{theorem}
\begin{proof}
    If $\Theta$ is non-empty, the symplectomorphism judgment $\Theta \vdashS t : \tau$ and the soundness of the equivalence relation (\cref{thm:categorical-equivalence-sound}) ensures that $\psi \in \SCat$.
    If $\Theta$ is empty, the result is trivial.
\end{proof}

The categorical semantics as defined in \cref{fig:cat-semantics-pauli} preserves the following invariant:

\begin{lemma} \label{lem:psiof-sound}
    If $\Theta \vdashP t : \tau$ and $\interpP{t} = (\mu,\psi)$, then $\interpL{\psiof{t}} = \psi$.
\end{lemma}

Finally we can prove the soundness and completeness of $\lambdaP$.

\begin{theorem}[Soundness of $\lambdaP$] \label{thm:phase-types-semantics-sound}
    If $\cdot \vdashP t : \tau$ and $t \rightarrow t'$ then $\interpP{t} = \interpP{t'}$
\end{theorem}
\begin{proof}[Sketch]
    First, we prove that substitution corresponds to composition of morphisms in the category (\cref{lem:phase-semantics-substitution-sound})---specifically, that $\interpP{t\{v/x\}}=\interpP{t} \circ \interpP{v}$. From there, we proceed by case analysis on the $\beta$ reduction rules. The full details are shown in \cref{appendix:categorical-semantics-pauli}, \cref{app:thm:phase-types-semantics-sound}.
\end{proof}

\begin{theorem}[Completeness of $\lambdaP$] \label{thm:lambdaP-completeness}
    For every $(\mu,\psi) \in \PCat(\interpP{\tau},\interpP{\tau'})$ there is some $x : \tau \vdashP t : \tau'$ such that $\interpP{t} = (\mu,\psi)$.
    In other words, every projective Clifford can be represented in $\lambdaP$.
\end{theorem}
\begin{proof}
    Define $t=\config{\up{\mu}} \up{\psi}$, where $x:\overline{\tau} \vdashC \up{\mu} : \Zd$ and $x:\overline{\tau} \vdashC \up{\psi} : \overline{\up{\psi}}$ as in \cref{thm:lambdaC-completeness}. To show $x:\tau \vdashP t : \tau'$, we need to show that $\psiof{t}=\up{\psi}$ respects the symplectic form \ie
    \[
        x_1:\overline{\tau},x_2:\overline{\tau} \vdashC \omega(\up{\psi}x_1,\up{\psi}x_2) \lrequiv \omega(x_1,x_2)
    \]
    This follows from the completeness of $\lrequiv$ (\cref{thm:lambdaC-lrequiv-completeness}). 
\end{proof}

\section{Extensions of \(\lambdaP\)}
\label{sec:extensions}
\subsection{Linear/Non-linear Types}
\label{sec:lnl}

Up until now we have only worked with linear types and linear maps in the sense of linear logic. To make $\lambdaP$ into a realistic programming language that supports data structures (\cref{sec:data-structures}) and real-world examples (\cref{sec:examples}), we need to add in support for other kinds of programming features such as data structures, polymorphism, and modularity. To do this, we will incorporate non-linear data structures in the style of linear/non-linear logic~(LNL)~\citep{benton1994mixed}.

In LNL type systems, there are two kinds each of terms, types, typing contexts, and typing judgments: linear and non-linear. 
They are useful for adding programming features to a linear language in a way that does not conflict with properties of the linear system, and have specifically been used for quantum programming languages in several settings~\citep{paykin2017qwire,paykin2017linearity,paykin2019hott,jia2022semantics,fu2020linear,fu2023proto}.
%

We start by defining a set of non-linear types $\ntype$, which include ordinary type formers such as function types, units, products, and sums. Other data structures such as recursive types, polymorphism, and dependent types can also be added in a straightforward way; these features are orthogonal to the relationship between the linear and non-linear type systems~\citep{paykin2018linear}. Non-linear types also include \emph{lifted Pauli types} $\lift{\tau}$ and \emph{lifted projective Clifford types} $\PseudoCliffordFunction{\tau_1}{\tau_2}$, which are reminiscent of boxed circuit types from Quipper~\citep{green2013quipper} and \textsc{QWIRE}~\citep{paykin2017qwire} in that they capture a first-order quantum operation inside a classical sub-language.
\begin{align*}
    \ntype &::= () \mid \ntype_1 \times \ntype_2 \mid \ntype_1 + \ntype_2 \mid \ntype_1 \rightarrow \ntype_2
        \mid \lift{\tau} \mid \PseudoCliffordFunction{\tau_1}{\tau_2}
\end{align*}
The syntax of non-linear terms $n$ is standard for the non-linear type systems, with $n_1;n_2$ being the elimination rule for the unit $()$. We add introduction rules $\lift{t}$ for $\lift{\tau}$ and $\lambda \lift{x}.t$ for $\PseudoCliffordFunction{\tau_1}{\tau_2}$ that \emph{lift} closed or single-variable $\lambdaP$ expressions to non-linear terms.
\begin{align*}
    n &::= () \mid n_1;n_2 \mid (n_1,n_2) \mid \pi_1 n \mid \pi_2 n \\
    &\mid \inl{n} \mid \inr{n} \mid \caseof{n}{\inl{x_1} \rightarrow n_1 \mid \inr{x_2} \rightarrow n_2} \\
    &\mid \lambda x.t \mid n_1 n_2 \mid \lift{t} \mid \lambda \lift{x}.t
\end{align*}
Non-linear contexts  $\Gamma ::= \cdot \mid \Gamma,x:\ntype$ map variables to non-linear types $\ntype$.
The non-linear typing judgment has the form $\Gamma \vdashT t : \ntype$ and is entirely standard except the rules for $\lift{\tau}$ and $\PseudoCliffordFunction{\tau_1}{\tau_2}$:
    \[ \begin{array}{c}
        \inferrule*
            {\Gamma;\cdot \vdashP t : \tau}
            {\Gamma \vdashT \lift{t} : \lift{\tau}}
        \qquad \qquad
        \inferrule*
            {\Gamma;x:\tau \vdashP t : \tau'
            }
            {\Gamma \vdashT \lambda \lift{x}.t : \PseudoCliffordFunction{\tau}{\tau'}}
    \end{array} \]
We write $\CliffordFunction{\tau}{\tau'}$ for $\PseudoCliffordFunction{\tau}{\tau'}$ when $|\tau|=|\tau'|$, \ie when $\tau$ and $\tau'$ have the same rank.

As seen in these rules, we amend the $\PCat$ typing judgments with non-linear contexts, for example $\Gamma; \Theta \vdashPL t : \tau$. The existing typing rules do not change, except that the non-linear context is preserved between them. This is true even in cases where the linear context $\Theta$ is restricted, as in:
\begin{align*}
    \inferrule*
    {\Gamma;\Theta \vdashP t : \Pauli \\
     \Gamma;\Theta' \vdashP t_x : \tau \\
     \Gamma;\Theta' \vdashP t_z : \tau \\
    }
    {\Gamma;\Theta,\Theta' \vdashP \caseof{t}{\PauliX \rightarrow t_x \mid \PauliZ \rightarrow t_z} : \tau}
    \label{eqn:lnl-case}
\end{align*}
In addition, we add elimination rules for non-linear types $\ntype$ in $\lambdaP$.
\begin{align*}
    t &::= \cdots \mid \letin{x}{n}{t} \mid n;t \mid \caseof{n}{\inl{x_1} \rightarrow t_1 \mid \inr{x_2} \rightarrow t_2} \mid \force{n} \mid \apply{n}{t}
\end{align*}
The last two expressions $\force{n}$ and $\apply{n}{t}$ are elimination forms for $\lift{\tau}$ and $\PseudoCliffordFunction{\tau}{\tau'}$ respectively; they can be used in $\lambdaP$ expressions an unrestricted number of times. Their typing rules are shown here:
\[
        \inferrule*
            {\Gamma \vdashT n : \lift{\tau}}
            {\Gamma;\cdot \vdashP \force{n} : \tau}
        \qquad \qquad
        \inferrule*
            {\Gamma \vdashT n : \PseudoCliffordFunction{\tau_1}{\tau_2} \\
             \Gamma;\Theta \vdashP t : \tau_1
            }
            {\Gamma;\Theta \vdashP \apply{n}{t} : \tau_2}
\]
Non-linear values of type $\lift{\tau}$ are any expressions $\lift{t}$ (not necessarily normalized), and values of type $\PseudoCliffordFunction{\tau_1}{\tau_2}$ are any lambda abstraction $\lambda \lift{x}.t$.
The $\beta$-reduction rules are:
\begin{align*}
    \force{\lift{t}} &\rightarrow_\beta t
    &
    (\lambda \lift{x}.t)(\config{r}{v}) &\rightarrow_\beta \config{r} t\{v/x\}
\end{align*}



\subsection{Compiling to Pauli Tableaux and Circuits}
\label{sec:compilation}

The idea of encoding projective cliffords based on their basis elements is inspired by the use of Pauli tableaux~\citep{aaronson2004}.
A tableau can be defined as a non-linear data type $\Frame{\tau}{\tau'}$ as follows:
\begin{align*}
    \Frame{\Pauli}{\tau'} &\triangleq \lift{\tau'} \times \lift{\tau'}
    &\qquad
    \Frame{\tau_1 \ptensor \tau_2}{\tau'} &\triangleq \Frame{\tau_1}{\tau'} \times \Frame{\tau_2}{\tau'}
\end{align*}
Intuitively, a tableau is a list of pairs $(P_{i,x}, P_{i,z})$ of $n$-\qudit Paulis, indicating that $U X_i U^\dagger = P_{i,x}$ and $U Z_i U^\dagger = P_{i,z}$ for each $i$.

Not all terms of this type correspond to well-formed tableaux---it must be the case that each $(P_{i,x}, P_{i,z})$ pair satisfies
$\omega(P_{i,z}, P_{i,x})=1$ and, for $i \neq j$, $\omega(P_{i,x},P_{j,x})=\omega(P_{i,z},P_{j,z})=0$. For this reason, programming directly with Pauli tableaux in other languages can be dangerous for the programmer, as the type system does not ensure tableaux are well-formed. 
However, every projective Clifford $\Gamma \vdashT n : \PseudoCliffordFunction{\tau}{\tau'}$ can be compiled to a well-formed Pauli tableau $\Gamma \vdashT \compile_\tau(n) : \Frame{\tau}{\tau'}$ by induction on $\tau$:
\begin{align*}
    \compile_\Pauli(n) &\triangleq (\lift{\force{n} \PauliX}, \lift{\force{n} \PauliZ}) \\
    \compile_{\tau_1 \ptensor \tau_2}(n)
        &\triangleq \left(\compile_{\tau_1}\left(\lambda \lift{x_1}.n(\inl{x_1})\right),
                          \compile_{\tau_2}\left(\lambda \lift{x_2}.n(\inr{x_2})\right)\right)
\end{align*}
This compilation can then be used in conjunction with circuit synthesis from Pauli tableaux to compile Clifford functions all the way to circuits~\citep{aaronson2004,PaykinSchmitz2023pcoast,schneider2023sat,brandl2024efficient}.

\subsection{Inverses}
\label{sec:inverse}

The projective Clifford group is closed under inverses, meaning that if $\gamma$ is a projective Clifford, then so is $\gamma^{-1}$. In this section we explore how to express the inverse of a Clifford in $\lambdaP$.

We can characterize inverses in $\PCat$ as $(\mu,\psi)^{-1}=(\mu_{\mathrm{inv}},\psi^{-1})$ where
\begin{align*}
    \psi^{-1}(v) &= \flatten{ (\omega(\psi(b_1^z),v) \oplus \omega(v,\psi(b_1^x)))
                        \oplus \cdots \oplus
                        (\omega(\psi(b_n^z),v) \oplus \omega(v,\psi(b_n^x))) }^{-1} \\
    \mu_{\mathrm{inv}}(b) &= \frac d2K^\psi\flatten{\psi^{-1}(b)}-\mu(\psi^{-1}(b))
\end{align*}
where $\flatten{-}$ is the canonical isomorphism $A \xrightarrow{\sim} \Zd^{\rank{A}}$, and $\mu_{\mathrm{inv}}$ is defined by extending its action on standard basis vectors $b$.

We start out with defining the equivalent of $\psi^{-1}$ in $\lambdaC$. Let $\vdashT n : \PseudoCliffordFunction{\tau}{\tau'}$. We can always create a linear function from $\overline{\tau'}$ to $\overline{\tau}$ by induction on $\tau$ as follows:
\begin{align*}
    \cdot \vdashL \pseudoinverse[\tau]{n} &: \overline{\tau'} \lolli \overline{\tau} \\
    \pseudoinverse[\Pauli]{n} &\triangleq \lambda q'.[\omega(\psiof*{\apply{n}{\PauliZ}}, q'), \omega(\psiof*{\apply{n}{\PauliX}}, q')] \\
    \pseudoinverse[\tau_1 \ptensor \tau_2]{n} &\triangleq \lambda q'.
        [\pseudoinverse[\tau_1]{n_1}(q'), \pseudoinverse[\tau_2]{n_2}(q')]
\end{align*}
where $n_i : \PseudoCliffordFunction{\tau_i}{\tau'}$ is defined as $\lambda \lift{x}.n(\iota_i(x))$.
Recall that if $\Gamma;\Theta \vdashP t : \tau$ then $\Gamma;\overline{\Theta} \vdashL \psiof{t} : \overline{\tau}$.

Lifting this inverse operation to a Clifford (assuming $\rank{\tau}=\rank{\tau'}$) is a bit more complicated. We would like an operation with the following signature:
\begin{align*}
    \inferrule*
    {\Gamma \vdashT f : \CliffordFunction{\tau}{\tau'} \\ \Gamma;\Theta \vdashP t : \tau'}
    {\Gamma;\Theta \vdashP f^{-1}(t) : \tau}
\end{align*}
One approach to defining $f^{-1}(t)$ could proceed as follows:
\begin{enumerate}
    \item Define a translation $\muof{f}$ similar to $\psiof{f}$ that projects out the phase morphism of a Pauli term.
    \item Define a non-linear function $K : \CliffordFunction{\tau}{\tau'} \rightarrow \lift{\tau} \rightarrow \lift{\R}$ that computes the phase $K$ for a given condensed Clifford encoding.
    \item Define $\pseudomuinverse{f} : \overline{\tau'} \lolli \R$ by induction on $\tau'$ as in terms of $K$, $\muof{f}$, and $\pseudoinverse{f}$.
    \item Finally, define $f^{-1}(t) \triangleq \config{\pseudomuinverse{f}(t)} \pseudoinverse{f}(t)$.
\end{enumerate}

Alternatively, we could compute the phase of $\pseudoinverse{f}(t)$ on a case-by-case basis by defining
\begin{align*}
    f^{-1}(\config{r}v) \rightarrow_\beta \config{r-s}(\pseudoinverse{f}(v))
\end{align*}
where $s$ is the unique phase satisfying $f(\pseudoinverse{f}(v)) \rightarrow^\ast \config{s}v$.

\subsection{Data Structures}
\label{sec:data-structures}

We can layer data structures on top of $\lambdaP$ types using techniques similar to data structures in Quipper~\citep{green2013quipper} and QWIRE~\citep{paykin2017qwire}. For the sake of this paper, we will focus on the type $\tau^n \cong \tau \ptensor \cdots \ptensor \tau$, which is often useful in real applications. Notice that this type introduces some limited dependent type such as the type $\FinNat{n}$ of natural numbers less than $n$, which we assume are supported by the classical components of $\lambdaP$.

The typing rules for $\tau^n$ are just a generalization of $\cprod$. Here, $\delta_{i_1,i_2}$ is $1$ if $i_1=i_2$ and $0$ otherwise.
\begin{align*}\begin{array}{c}
    \inferrule*
    {\Gamma \vdashT m : \FinNat{n} \\ \Gamma;\Theta \vdashPL t : \tau}
    {\Gamma;\Theta \vdashPL \iota_m(t) : \tau^n}
  \quad\qquad
    \inferrule*
    {\Gamma; \Theta \vdashT t : \tau^n \\
     \Gamma,i:\FinNat{n}; x:\tau \vdashPL t' : \tau' \\
    }
    {\Gamma; \Theta \vdashPL \caseof{t}{\iota_i(x) \mapsto t'} : \tau' }
\end{array}\end{align*}

\section{Case Study: Stabilizer Error-correcting Codes}
\label{sec:examples}

In this section we explore how programming quantum algorithms with $\lambdaP$ arises from thinking about the action of Cliffords on Paulis rather than states. The results are cleaner and more compact programs, compared to their circuit analogues, that convey some intuition about the action of quantum algorithms on Paulis, which still being implementable on real devices. For this case study, we focus specifically on stabilizer error-correcting codes.

We start by re-introducing the more practical programming style syntax used in the introduction, but now  now in the case of general qudits. We provide a formal translation between the two syntaxes in \cref{appendix:glossary}. After that, we introduce the basics of stabilizer codes, which include three main unitary programming tasks, all of which are Clifford: encoding, performing logical operations, and preparing for the syndrome measurement.\footnote{Since $\lambdaP$ only encompasses unitary operations, we separate the preparation of syndrome measurement from the measurement itself. We also elide the details of decoding, which are generally classical.}

\subsection{Qu\(d\)it Standard Library}

The single-qudit Paulis are defined as follows:

\noindent
\begin{minipage}{0.32\textwidth}
\begin{lstlisting}
    X :: |^Pauli^|
    X =  |^[1,0]^|
\end{lstlisting}
\end{minipage}
\begin{minipage}{0.32\textwidth}
\begin{lstlisting}
    Y :: |^Pauli^|
    Y =  |^[1,1]^|
\end{lstlisting}
\end{minipage}
\begin{minipage}{0.32\textwidth}
\begin{lstlisting}
    Z :: |^Pauli^|
    Z =  |^[0,1]^|
\end{lstlisting}
\end{minipage}
%

The multi-qudit Clifford group is generated from: the quantum fourier transform, which generalizes the Hadamard gate; the phase-shift gate, which generalizes $S$; and the two-qudit SUM gate, which generalizes CNOT~\citep{farinholt2014ideal}. All three of these can be expressed naturally in $\lambdaP$.
\begin{center}
\begin{minipage}[t]{0.26\textwidth}
\begin{lstlisting}[basicstyle=\footnotesize]
qft :: |^Pauli -o Pauli^|
qft |^X^| = Z
qft |^Z^| = pow(X, -1)
\end{lstlisting}
\end{minipage}
\begin{minipage}[t]{0.28\textwidth}
\begin{lstlisting}[basicstyle=\footnotesize]
phase :: |^Pauli -o Pauli^|
phase |^X^| = Y
phase |^Z^| = Z
\end{lstlisting}
\end{minipage}
\begin{minipage}[t]{0.41\textwidth}
\begin{lstlisting}[basicstyle=\footnotesize]
sum :: |^Pauli ** Pauli -o Pauli ** Pauli^|
sum |^in1 X^| *= in2 X
sum |^in1 Z^| *= I
sum |^in2 X^| *= I
sum |^in2 Z^| *= in1 pow(Z,-1)
\end{lstlisting}
\end{minipage}
\end{center}
In the definition of \lstinline{sum}, we write \lstinline{f |^p^| *= P} as shorthand for \lstinline{f |^p^| = p * P}; this is a recurring pattern in many examples.

Other useful parametric operations compose projective Cliffords in sequence and in parallel. The parentheses indicate an infix operation.
\begin{lstlisting}[basicstyle=\footnotesize]
    (.) :: |^ tau2 -o tau3 ^| -> |^ tau1 -o tau2 ^| -> |^ tau1 -o tau3 ^|
    g . f |^q^| = g (f q)
        
    (**) :: |^ tau1 -o tau2 ^| -> |^ tau1' -o tau2' ^| -> |^ tau1 ** tau1' -o tau2 ** tau2' ^|
    f ** g |^in1 q ^| = in1 (f q)
    f ** g |^in2 q'^| = in2 (g q')
\end{lstlisting}

$\lambdaP$ functions don't need to be defined by pattern matching. Consider for example that every every Pauli unitary $P$ is also in the Clifford group because it satisfies $P Q P^{\dagger}=\tau^{\symplecticform{P}{Q}} Q$:
\begin{lstlisting}[basicstyle=\footnotesize]
    pauliToClifford :: |^tau^| -> |^ tau -o tau ^|
    pauliToClifford p |^q^| = <omega p q> q
\end{lstlisting}

Another useful observation is that for any $n$-qudit Pauli operator $P$, the block-diagonal matrix 
$\mathrm{diag}(I,P,\cdots,P^{d-1})$ is also Clifford, generalizing $CNOT$ and $SUM$ even further.
This ``controlled Pauli'' can be expressed in $\lambdaP$ as follows:
\begin{lstlisting}[basicstyle=\footnotesize]
    control-pauli :: |^ tau ^| -> |^ Pauli ** tau -o Pauli ** tau ^|
    control-pauli p |^ in1 X ^| *= in2 p
    control-pauli p |^ in1 Z ^| *= I
    control-pauli p |^ in2 q ^| *= in1 (pow(Z,omega p q))
\end{lstlisting}
Note that the symplectomorphism check succeeds for the final line because:
\begin{align*}
    \omega([[\omega(p,x_1), 0], x_1], [[\omega(p,x_2), 0], x_2])
    &\lrequiv \omega([\omega(p,x_1), 0], [\omega(p,x_2), 0]) + \omega(x_1, x_2)) \\
    &\lrequiv 0 + \omega(x_1,x_2) \lrequiv \omega(x_1,x_2)
\end{align*}

\subsection{Basics of Stabilizer Codes}

A stabilizer code is a quantum error-correcting code characterized by an Abelian group $\mathbf{S} \subseteq \Pauligroup$ of stabilizers, with $\tau^{s} I \in \mathbf{S}$ if and only if $s=0$. 
The fact that $\mathbf{S}$ is Abelian means that all of the Paulis in $\mathbf{S}$ commute with each other. 
The code space of such a stabilizer code is the set of states stabilized by $\mathbf{S}$, \ie $\mathcal{C}(\mathbf{S}) = \{ \ket{\phi} \mid \forall P \in \mathbf{S}.~ P \ket{\phi}=\ket{\phi} \}$.

We use the meta-variable $n$ to refer to the number of physical qubits, $k$ for the number of logical qubits, and $r$ for the number of stabilizer generators $S_0,\ldots,S_{r-1}$ in an $[[n,k]]$ code. Note that we always have $k+r=n$~\citep[Chapter~3]{gottesman2024surviving}.
For additional details, see \citet{gottesman2024surviving}.

\subsection{Encoders} \label{sec:encoder}

An encoder of an $[[n,k]]$ quantum error correction code $\mathbf{S}$ is an $n$-qudit unitary $U_e$ that maps $k$ logical qudits $\ket{\phi}$ to $n=r+k$ physical qudits  $\overline{\ket{\phi}} = U_e(\ket{0}^{\otimes r} \otimes \ket{\phi})$.

Consider a stabilizer code $\mathbf S$ and a Clifford $U$ on $\mathbb{C}^{2^n}$ that satisfies $U Z_i U^\dagger = S_i$ for each of the $r$ stabilizer generators $S_i$. Then for any $\ket{\phi} \in \mathbb{C}^{2^k}$, it follows that
\begin{align*}
    S_i (U (\ket{0}^{\otimes r} \otimes \ket{\phi}))
    = U Z_i (\ket{0}^{\otimes r} \otimes \ket{\phi}))
    = U (\ket{0}^{\otimes r} \otimes \ket{\phi}))
\end{align*}
Thus $U (\ket{0}^{\otimes r} \otimes \ket{\phi}) \in \mathcal{C}(\mathbf{S})$, and so $U$ is an encoder for $\mathbf S$.


In order to fully define $U$ as a projective Clifford, we also need to define its action on all $Z_i$ (not just for $i < r$) and $X_i$. In other words, we need to extend the set $\{S_1,\ldots,S_r\}$ of stabilizer generators with additional operators $\{S_{r+1},\ldots,S_n\}$ as well as $\{T_1,\ldots,T_{n}\}$ such that the sets $\{S_i\}$ and $\{T_i\}$ are each commutative and each $S_i$ anticommutes with each $T_i$.
The operators $S_{r+1},\cdots,S_n$ stabilize the logical $\ket{0}$ state, and determine the rest of the logical computational basis states up to phase.

Assume that in $\lambdaP$ we encode our operators $\{S_i\}$ and $\{T_i\}$ as functions from \lstinline{Nat n}---the type of natural numbers less than $n$---to a Pauli type.
Then we can define our encoder as follows:
\begin{center}
\begin{minipage}[t]{0.27\textwidth}
\begin{lstlisting}[basicstyle=\footnotesize]
stab :: Nat n -> |^ Pauli^n ^|
...
\end{lstlisting}
\end{minipage}
\begin{minipage}[t]{0.3\textwidth}
\begin{lstlisting}[basicstyle=\footnotesize]
destab :: Nat n -> |^ Pauli^n ^|
...
\end{lstlisting}
\end{minipage}
\begin{minipage}[t]{0.3\textwidth}
\begin{lstlisting}[basicstyle=\footnotesize]
encoder ::|^ Pauli^n -o Pauli^n ^|
encoder |^ in i Z ^| = stab i
encoder |^ in i X ^| = destab i
\end{lstlisting}
\end{minipage}
\end{center}


Note that it should be possible to give a partial definition for the Steane encoder by omitting the action on the $T_i$ destabilizers and have the $\lambdaP$ compiler generate suitable choices for the action of the encoder on $\iota_i Z$. Indeed, such a partial definition could be procedurally incorporated into $\lambdaP$ in a way that is consistent between invocations using the symplectic Gram-Schmidt algorithm~\citep{Sil08}, though we leave that for future work. In that case we would give the following partial definition:
\begin{lstlisting}[basicstyle=\footnotesize]
    encoderPartial :: |^ Pauli^k -o Pauli^n ^|
    encoderPartial |^ in i Z ^| = stab i
\end{lstlisting}

\begin{example}
    The Steane code is a $[[7,1]]$ CSS code for qubits, i.e. it encodes 1 logical qubit using seven physical qubits, over which all Clifford gates can be implemented transversally. 
    Its encoder is expressed in \cref{fig:steane-encoder} in terms of its six stabilizer generators.
\end{example}

\begin{figure}
    \centering
    \begin{center}
\begin{minipage}[t]{0.45\textwidth}
\begin{lstlisting}[basicstyle=\footnotesize]
steaneStabilizer :: Nat 7 -> |^ Pauli^7 ^|
steaneStabilizer 0 = X ** I ** I ** X ** X ** X ** I
steaneStabilizer 1 = I ** X ** I ** X ** I ** X ** X
steaneStabilizer 2 = I ** I ** X ** I ** X ** X ** X
steaneStabilizer 3 = Z ** I ** I ** Z ** Z ** Z ** I
steaneStabilizer 4 = I ** Z ** I ** Z ** I ** Z ** Z
steaneStabilizer 5 = I ** I ** Z ** I ** Z ** Z ** Z
-- S6 commutes with S0-S5
steaneStabilizer 6 = Z ** Z ** Z ** Z ** Z ** Z ** Z
\end{lstlisting}
\end{minipage}
\begin{minipage}[t]{0.51\textwidth}
\begin{lstlisting}[basicstyle=\footnotesize]
steaneEncoder :: |^ Pauli^7 -o Pauli^7 ^|
steaneEncoder |^in i Z^| = steaneStabilizer i
--destabilizers obtained from stabilizers by
-- exchanging X and Z on the last qubit
steaneEncoder |^in i X^| = inj qft 6 (steaneStabilizer i)
  where
    inj :: |^ tau1 -o tau2 ^| -> Nat n -> |^ tau_1^n -o tau_2^n ^|
    inj f i |^in j q^| = if i==j then in j (f q)
                         else in j q
\end{lstlisting}
\end{minipage}
\end{center}
    \caption{An encoder for the qudit Steane code}
    \Description{Left: a function \texttt{steaneStabilizer} of type \texttt{Nat 7 -> | Pauli\^^7 |}, which sends index $i$ to the $ith$ Pauli stabilizer. Right: a projective Clifford function \texttt{steaneEncoder} which sends \texttt{in i Z} to \texttt{steaneStabilizer i}.}
    \label{fig:steane-encoder}
\end{figure}

\subsection{Logical Operators}

Any projective Clifford $f$ on logical qubits can be implemented by inverting the encoder as follows:
\begin{lstlisting}[basicstyle=\footnotesize]
    logicalOperator :: |^ Pauli^k -o Pauli^k ^| -> |^ Pauli^n -o Pauli^n ^|
    logicalOperator f = encoder . f . encoder^-1
\end{lstlisting}

We can also define transversal operators directly, where a logical operator is implemented by applying that operator to every physical qubit.
\begin{lstlisting}[basicstyle=\footnotesize]
    transversal :: |^ tau_1^m -o tau_2^m ^| -> |^ (tau_1^n)^m -o (tau_2^n)^m ^|
    transversal f |^ in j (in i q) ^| = map (in i) (f (in j q))
      where
        map :: |^ tau_1 -o tau_2 ^| -> |^ tau_1^m -o tau_2^m ^|
        map g |^ in j q' ^| = in j (g q')
\end{lstlisting}

\subsection{Syndrome Preparation}

Lastly, in order to detect and correct errors we must prepare and then measure an error syndrome.

Syndrome measurement involves measuring the eigenvalue of each of the $r$ stabilizer generators $S_i$ into the $r$ ancilla qubits. In particular, suppose a codeword $\ket{\phi} \in \mathcal{C}(\mathbf{S})$ is subjected to a Pauli error $E$. Then measuring the eigenvalue of $S_i$ results in:
\begin{align*}
    \bra{\phi} E^\dagger S_i E \ket{\phi}
    = \bra{\phi} \zeta^{\omega(S_i, E)} E^\dagger E S_i \ket{\phi}
    =  \zeta^{\omega(S_i, E)} \braket{\phi \mid \phi} = \zeta^{\omega(S_i, E)}
\end{align*}

\begin{figure}
    \centering
\begin{subfigure}{0.39\textwidth}
\centering
\begin{tikzpicture} \node[scale=0.6] {
\begin{quantikz}[column sep=0.3cm,row sep=0.1cm]
    \lstick[3]{$r$} &\gate{\qft} & \ctrl{3}      & \gate{\qft^{-1}}&        &             &                & & \\
    \setwiretype{n} & \vdots     &               &                 & \ddots &             &                & & \\
                    &            &               &                 &        & \gate{\qft} & \ctrl{1}       & \gate{\qft^{-1}}& \\
                    &\qwbundle{n}& \gate{S_0}    &                 &        &             & \gate{S_{r-1}} & & \\
\end{quantikz}
}; \end{tikzpicture}
\end{subfigure}%
~
\begin{subfigure}{0.59\textwidth}
\begin{lstlisting}[basicstyle=\footnotesize]
    syndromePrep :: |^ Pauli^r ** Pauli^n -o Pauli^r ** Pauli^n ^|
    syndromePrep |^ in1 (in i X) ^| *= I
    syndromePrep |^ in1 (in i Z) ^| *= in2 (pow (stab i, d-1))
    syndromePrep |^ in2 q ^|        *=
        star r (\j -> pow(in1 (in j X), omega(stab j, q)))
\end{lstlisting}
\end{subfigure}
    \caption{Stabilizer syndrome preparation as an informal circuit diagram (left) and as a $\lambdaP$ program (right). 
        $\text{star}~r~f$ is the $r$-fold $\cprod$-product $f(0) \cprod \cdots \cprod f (m-1)$.
    }
    \Description{Left: A circuit implementing stabilizer syndrome preparation made up of \texttt{qft} gates and controlled-$S_i$ gates. Right: a projective Clifford function \texttt{syndromePrep} expressed in $\lambdaP$.}
    \label{fig:syndrome-preparation}
\end{figure}

\noindent
Measuring each stabilizer results in a vector $(\omega(S_0,E),\ldots,\omega(S_{r-1},E))$ called the \emph{syndrome}, which is used for detection and/or correction. 
Here we only consider the unitary part of this circuit (\cref{fig:syndrome-preparation}), which we call syndrome preparation. It is made up of repeated calls to controlled stabilizers.

If we expand out this operation as a projective Clifford, we see some patterns. For $0 \le i < r$, the action of this operator sends $X_i$ to $X_i$ and $Z_i$ to $Z_i \cprod S_i^{d-1}$. On the other hand, it sends Paulis  $I \otimes Q$ (with support entirely on the $n$ physical qudits), to $(I \otimes Q) \cprod X_0^{\omega(S_0,Q)} \cprod \cdots \cprod X_{r-1}^{\omega(S_{r-1}, Q)}$.

\section{Related and Future Work}
\label{sec:related-work}

\subsection{Related Work in Quantum Programming Languages}

The majority of quantum programming languages today are based on the quantum gate model, where programs are formed from primitive quantum gates together with classical~\citep{selinger2009quantum,green2013quipper,altenkirch2010quantum,khalate2022llvm} or quantum~\citep{Badescu2015quantum,chiribella2013quantum,voichick2023qunity} control flow. $\lambdaP$ takes a different approach, inspired by quantum programming languages based the linear-algebraic structure of quantum computing.

\paragraph{Linear-algebraic Programming Languages}

The $\lambdaC$ calculus is closely related to a line of work on linear-algebraic lambda calculi stemming from QML~\citep{altenkirch2005qml} and Lineal~\citep{arrighi2017lineal}, where terms in the calculus correspond to linear transformations over vectors spaces.
%
More recently, \citet{diazcaro2024linear} combined a linear-algebraic lambda calculus with a linear logic, which allows them to prove well-typed functions are linear instead of defining function application pointwise. In contrast to their work, $\lambdaC$ does not include all of the connectives of intuitionistic multiplicative linear logic, notably the tensor product $\otimes$. The tensor product is not necessary for programming projective Cliffords, and by excluding it we avoid much of the complexity of their calculus. For example, while sums of values of type $\alpha_1 \oplus \alpha_2$ can be combined e.g. $[v_1,v_2] + [v_1',v_2'] \rightarrow [v_1+v_1',v_2+v_2']$, the same cannot be done for sums of values of type $\alpha_1 \otimes \alpha_2$. As such, their type system does not allow reduction under $t_1 + t_2$ if $t_1$ and $t_2$ have type $\alpha_1 \otimes \alpha_2$.
%
Both type systems have equational theories based on logical relations, and both satisfy the fact that $f(a_1+a_2)=f(a_1)+f(a_2)$.



Linear-algebraic lambda calculi are closely related to quantum computing. QML~\citep{altenkirch2005qml}, Lambda-$S_1$~\citep{diaz2019realizability}, and symmetric pattern matching calculi~\citep{sabry2018symmetric} limit linear transformations to unitary ones by ensuring that the branches of a quantum case statement are appropriately orthogonal. These orthogonality checks mirror the symplectic form condition in $\lambdaP$---we too are restricting linear transformations to a particular shape. However, as these other languages target unitary transformations, their orthogonality checks involve simulating a linear transformation on vectors of size $2^n$. In contrast, the orthogonality check on symplectic encodings only involves vectors of size $2n$.

\paragraph{Pauli-based Programming and Optimization}

We can also compare this work to representations of quantum algorithms based on Paulis, such as Pauli tableaux~\citep{aaronson2004,schmitz2024graph}, Pauli exponentials/Pauli rotations~\citep{li2022paulihedral,Zhang2019,PaykinSchmitz2023pcoast}, phase polynomials~\citep{amy2018towards,amy2019formal}, and the ZX calculus~\citep{kissinger2020}. These representations have been used in the verification, optimization, and synthesis of quantum algorithms, but are not necessarily well-suited for programming. For one, they do not generally support abstractions like data structures and parametricity. For another, while phase polynomials and ZX-calculus processes can be used for circuit synthesis, they also include operations that are not physically realizable: for example, synthesizing a circuit from a ZX diagram is \#P-hard in general~\citep{debeaudrap2022circuit}. Pauli exponentials are synthesizable into circuits in a straightforward way, but they are only fully general for unitaries. Pauli exponentials can be extended beyond unitaries, for example in PCOAST~\citep{PaykinSchmitz2023pcoast,schmitz2023optimization}, but only for primitive preparation and measurement operators.

The lines of work most closely-related to $\lambdaP$ are Gottesman types~\citep{rand2020gottesman} and Heisenberg logic~\citep{sundaram2023hoare}. In those works, the projective action of Cliffords are used as a specification for Clifford circuits; for example, the Hadamard gate is given the specification that it maps $X$ to $Z$ and vice versa. The advantage of this approach is that it can be used for \emph{partial} specification, allowing for properties like separability as well as non-Clifford gates~\citep{rand2021extending}. While $\lambdaP$ focuses on describing high-level algorithms and allowing a compiler to efficiently generate the circuits, Gottesman types and Heisenberg logic focus on verification of circuits that already exist, which may have additional constraints such as limited connectivity or depth. These two approaches are complimentary however; $\lambdaP$ could be used to inspire a richer specification language for Gottesman types, and Heisenberg logic could be used to verify the correctness of a $\lambdaP$ compiler.

\paragraph{Qu\(d\)it Quantum Programming}

In recent years there has been significant interest in higher-dimensional (qudit) systems, inspired by the capabilities of real hardware~\citep{chi2022programmable,hrmo2023native,fischer2023universal}. Simulations confirm that qubit-based implementations of quantum algorithms can be significantly more efficient compared to their qubit versions~\citep{litteken2022communication,nadkarni2021quantum,keppens2025qudit}. As such,  researchers have investigated techniques for compiling qubit quantum programs to qudit circuits using a variety of techniques including gatewise transpilation~\citep{drozhzhin2024transpiling}, phase polynomials~\citep{heyfron2019quantum}, and unitary decomposition~\citep{mato2023compilation,volya2023qudcom,lysaght2024quantum,sharma2024compilation,kiktenko2025qudits}.

While some gate-based programming frameworks support qudit simulation and development~\citep{mato2023mixed,chatterjee2023QuDiet,volya2023qudcom,deSouzaFarias2025QuForge}, to our knowledge $\lambdaP$ is the first in the style of linear-algebraic programming languages.

\subsection{Future Work}
\label{sec:future-work}

\paragraph{Implementation}

The examples given in this paper are written in pseudocode, and an implementation of $\lambdaP$ is a next logical step. Such an implementation should include both an interpreter and a compiler to Pauli tableaux and thus circuits, as described in \cref{sec:compilation}.

An efficient type checker is necessary for practical use of $\lambdaP$, as checking the symplectomorphism relation is tricky by hand.  Consider a judgment $x_1:\tau_1,x_2:\tau_2 \vdashP \symplecticform[\tau']{t_1}{t_2} \lrequiv r$. Naively checking this condition would require evaluating the symplectic form on each of the $2^{|\tau_1|}\times 2^{|\tau_2|}$ basis values of type $\tau_1$ and $\tau_2$ respectively, where the cost of each evaluation is unknown.
Instead, we propose implementing symplectomorphism-checking via symbolic evaluation, which would require just a single call to an SMT solver. 

\paragraph{Beyond Cliffords}
\label{sec:beyond-cliffords}

As discussed in the introduction, $\lambdaP$ is not universal for quantum computing since it focuses solely on the Clifford group. However, we speculate that there are several directions in which to extend the language to encompass general quantum computation.

One avenue is to integrate $\lambdaP$ into representations of Pauli exponentials such as PCOAST~\citep{PaykinSchmitz2023pcoast}. PCOAST is an optimization framework for universal quantum computation where typical gates are replaced with Pauli-based components---Pauli rotations (unitaries of the form $e^{-i P \theta}$), Pauli tableaux, and generalized versions of Pauli measurements and state preparation. Existing PCOAST representations are not well-suited for programming, as they are not parametric on the number of qubits or modular over their inputs. We hypothesize that by replacing the tableaux with $\lambdaP$ expressions, we could very naturally express programs for universal quantum computing in the style of PCOAST terms but with the added benefit of compositionality, modularity, and type checking.

PCOAST supports measurement and preparation, which we note \emph{are} efficiently simulatable on Pauli tableaux, with subtleties for even dimensions $d$. Incorporating measurement into $\lambdaP$, even without arbitrary Pauli rotations, would be extremely useful for applications like error correction.

Another approach towards universaility is to extend $\lambdaP$ to be parametric in the dimension $d$ of \qudits (\cref{sec:why-qudits}), which would extend $\lambdaP$s computational power along the lines of mixed-dimensional quantum computing~\citep{mato2023mixed,mato2023compression,volya2023qudcom,chatterjee2023QuDiet,poor2023completeness}. 
Clifford unitaries in higher dimensions are not all Clifford in lower dimensions: for example, the 2 qubit Fourier transform $F_4$ is single-\qudit Clifford (with $d=4$) but not multi-qubit Clifford. In fact, $F_4$ together with $F_2=H$ is universal.


A third avenue to explore is moving up the Clifford hierarchy. The Clifford hierarchy starts with the Pauli group, and the $(k+1)$st level of the hierarchy maps the Pauli group to the $k$th level:
\begin{align*}
    \Clf~^0_{d,n} &\triangleq \Pauligroup 
    &
    \Clf~^{k+1}_{d,n} &\triangleq \{U \mid \forall P \in \Pauligroup.~ U P U^\dagger \in \Clf~^k_{d,n} \}
\end{align*}
Extending this to a programming abstraction, if a projective Clifford is a map from a single Pauli to the Pauli group, then a projective $k$th-level Clifford is a map from $k$ Paulis to the Pauli group.
We hypothesize that a $k$th-level Clifford could be an expression in a variant of $\lambdaP$ with $k$ free variables: $x_1:\tau_1,\ldots,x_k:\tau_k \vdashP t : \tau$. For example, the $T$ gate acts on two single-\qudit Paulis as follows:
\begin{lstlisting}
T :: |^ Pauli -o Pauli -o Pauli ^|
T |^Z^| |^q^| = pauliToClifford Z q 
T |^X^| |^X^| = Y
T |^X^| |^Z^| = <1>Z 
\end{lstlisting}

There are challenges with this approach, however, as the Clifford hierarchy is neither very well-behaved or well-understood. Condensed encodings would need to be extended up the Clifford hierarchy, perhaps using observations by \citet{chen2024characterizing}. Higher levels of the Clifford hierarchy do not form a group, and are not closed under composition. 
%
Regardless, the Clifford hierarchy \emph{is} universal in that compositions of Clifford hierarchy gates are enough to simulate any unitary to arbitrary accuracy.
It would be sufficient to extend the type system to a universal superset of the Clifford group by including a single non-Clifford operation like $T$, or extending it to a subclass such as 3-Cliffords or $k$-level semi-Cliffords~\citep{chen2024characterizing,de_Silva_2021}.


\subsection{Conclusion}

This paper presents a type system for programming projective Clifford unitaries as functions on the Pauli group. We establish a Curry-Howard correspondence with the category $\PCat$ of projective Clifford encodings built out of $\Zd$-linear maps and symplectic morphisms. We present the type systems of $\lambdaC$ and $\lambdaP$ and prove they are well-behaved with respect to the categorical semantics. Finally, we explore what it is like to program with $\lambdaP$ via extensions and examples.

\begin{acks} 
Many thanks to Albert Schmitz and Jon Yard for numerous fruitful discussions about this work.
In addition, we would like to thank the anonymous reviewers for their extremely constructive feedback.
The work of Sam Winnick was supported in part by the NSERC Discovery under Grant No. RGPIN-2018-04742 and the NSERC project FoQaCiA under Grant No. ALLRP-569582-21.
\end{acks}

\bibliographystyle{ACM-Reference-Format}
\bibliography{bibliography}


\section*{Supplementary Material}
\appendix

\section{Type System for \(\Zd\)-Modules}
\label{sec:r-modules-appendix}
\subsection{Contextual Reduction Rules}
\label{app:lambdaC-reduction-rules}

In this section we supplement the $\beta$ reduction rules given in \cref{sec:lambdaC} with a call-by-value evaluation scheme for reducing under a context. The additional rules are shown in \cref{fig:operational-semantics-C-contextual}

\begin{figure}[b]
    \centering
    \[\begin{array}{c}
        \inferrule*
            {a \rightarrow_\beta a'}
            {a \rightarrow a'}
        \qquad \qquad
        \inferrule*
            {a \rightarrow a'}
            {\letin{x}{a}{a''}
             \rightarrow
             \letin{x}{a'}{a''}
            }
        \\ \\
        \inferrule*
            {a \rightarrow a'}
            {a \cdot a''
             \rightarrow
             a' \cdot a''
            }
        \qquad \qquad
        \inferrule*
            {a \rightarrow a'}
            {r \cdot a
             \rightarrow
             r \cdot a'
            }
        \\ \\
        \inferrule*
            {a_1 \rightarrow a_1'}
            {a_1 + a_2 \rightarrow a_1' + a_2}
        \qquad \qquad
        \inferrule*
            {a_2 \rightarrow a_2'}
            {v_1 + a_2 \rightarrow v_1 + a_2'}
        \\ \\
        \inferrule*
            {a_1 \rightarrow a_1'}
            {[a_1,a_2] \rightarrow [a_1',a_2]}
        \qquad \qquad 
        \inferrule*
            {a_2 \rightarrow a_2'}
            {[v_1,a_2] \rightarrow [v_1,a_2']}
        \\ \\
        \inferrule*
            {a \rightarrow a'}
            {\caseof{a}{\inl{x_1} \rightarrow a_1 \mid \inr{x_2} \rightarrow a_2}
                \rightarrow
                \caseof{a'}{\inl{x_1} \rightarrow a_1 \mid \inr{x_2} \rightarrow a_2}
            }
        \\ \\
        \inferrule*
            {a_1 \rightarrow a_1'}
            {a_1 a_2 \rightarrow a_1' a_2}
        \qquad \qquad
        \inferrule*
            {a_2 \rightarrow a_2'}
            {v a_2 \rightarrow v a_2}
    \end{array}\]
    \caption{Operational semantics for reducing $\CCat$-expressions under a context}
    \Description{A set of eleven inference rules that show contextual reductions of sub-expressions inside a $\lambdaC$ expression.}
    \label{fig:operational-semantics-C-contextual}
\end{figure}

\subsection{Strong Normalization}
\label{app:strong-normalization}

Next, we want to show that every well-typed expression $\cdot \vdashL a : \alpha$ (1) always evaluates to a value (does not get stuck or run forever); and (2) evaluates to a \emph{unique} value.

Uniqueness is easy for this operational semantics, because it has a fixed evaluation order and is thus deterministic.


To prove that every expression evaluates to a value, we define the following logical relation on values and closed expressions respectively:
\begin{align*}
    \Val{\R} &\triangleq \R \\
    \Val{\alpha_1 \oplus \alpha_2} &\triangleq \{[v_1,v_2] \mid v_1 \in \Val{\alpha_1} \wedge v_2 \in \Val{\alpha_2} \} \\
    \Val{\alpha_1 \lolli \alpha_2} &\triangleq \{\lambda x.a \mid \forall v \in \Val{\alpha_1}.~ a\{v/x\} \in \Norm{\alpha_2} \} \\
    \Norm{\alpha} &\triangleq \{a \mid \cdot \vdashL a : \alpha \wedge \exists v \in \Val{\alpha}.~ a \rightarrow^\ast v\}
\end{align*}
We can extend this relation to open expressions. Let $\delta$ be a map from variables to (closed) values; in other words, it is a partial function from the set of variables to the set of values. We write $x \mapsto v \in \delta$ to mean that $x \in \dom(\delta)$ and $\delta(x)=v$. In addition, we write $\delta(a)$ for the expression obtained by substituting values $\delta(x)$ for free variables $x$ in $a$ in the usual capture-avoiding way.

In the same way, typing contexts $\Delta$ can be thought of as partial functions from the set of variables to the set of types. For example, $\Delta_0 = x_0:\alpha_0$ is a partial function with domain $\{x_0\}$ defined by $\Delta_0(x_0)=\alpha_0$.
We now define $\ValRelation{\Delta}$ as follows:
\begin{align} \label{eqn:VDelta-defn}
    \ValRelation{\Delta}
    &\triangleq \left\{
        \delta
        \mid
        \forall x \in \dom(\Delta),~
        x \in \dom(\delta) \wedge \delta(x) \in \Val{\Delta(x)}
    \right\}
\end{align}
In that case, define:
\begin{align*}
    \Norm{\Delta \vdashL \alpha} &\triangleq \{ a \mid \Delta \vdashL a : \alpha \wedge \forall \delta \in \ValRelation{\Delta}.~ \delta(a) \in \Norm{\alpha} \}.
\end{align*}

It is trivial to see that the following lemma holds about $\Norm{\alpha}$:
\begin{lemma} \label{lem:norm-step-inversion}
    If $\cdot \vdashL a : \alpha$ and $a \rightarrow a'$ such that $a' \in \Norm{\alpha}$, then $a \in \Norm{\alpha}$.
\end{lemma}
\begin{proof}
    Follows from the definition of $\Norm{\alpha}$.
\end{proof}

We can also prove several \emph{compatibility lemmas} about the normalization relation.

\begin{lemma}[Compatibility of constants and tuples] \label{lem:norm-compatibility}
    The following rules are sound for $\Norm{\alpha}$:
    \[ \begin{array}{c}
        \inferrule*
            {~}{r \in \Norm{\cdot \vdashL \R}}
        \qquad \qquad
        \inferrule*
            {a_1 \in \Norm{\Delta \vdashL \alpha_1} \\ a_2 \in \Norm{\Delta \vdashL \alpha_2}}
            {[a_1, a_2] \in \Norm{\Delta \vdashL \alpha_1 \oplus \alpha_2}}
    \end{array}\]
\end{lemma}
\begin{proof}
    The rule $r \in \Norm{\R}$ is trivial since $r$ is already a value in $\Val{\R}$.

    For the $[a_1,a_2]$ case: Let $\delta \in \ValRelation{\Delta}$. We want to show that $\delta([a_1,a_2])=[\delta(a_1),\delta(a_2)] \in \Norm{\alpha_1 \oplus \alpha_2}$. By the hypotheses we know that $\delta(a_1) \rightarrow v_1$ and $\delta(a_2) \rightarrow v_2$, which means
    \begin{align*}
        [\delta(a_1),\delta(a_2)]
            &\rightarrow^\ast [v_1,\delta(a_2)]
            \rightarrow^\ast [v_1,v_2]
    \end{align*}
    which is a value.
\end{proof}

\begin{lemma}[Compatibility of $+$] \label{lem:norm-compatibility-plus}
    If $a_1,a_2 \in \Norm{\Delta \vdashL \alpha}$ then $a_1 + a_2 \in \Norm{\Delta \vdashL \alpha}$.
\end{lemma}
\begin{proof}
    Let $\delta \in \ValRelation{\Delta}$. We proceed by induction on $\alpha$.

    If $\alpha=\R$ then we know there exist constants $r_1$ and $r_2$ such that 
    $\delta(a_i) \rightarrow^\ast r_i$ and 
    \begin{align*}
        r_1 + r_2
            \rightarrow_\beta r'
    \end{align*}
    where $r'=r_1+r_2 \in \R$.

    If $\alpha=\alpha_1 \oplus \alpha_2$ then we know $\delta(a_i) \rightarrow^\ast [v_{i,1},v_{i,2}]$ for $v_{i,j} \in \Val{\alpha_j}$.
    Then
    \[
        \delta(a_1) + \delta(a_2) \rightarrow^\ast [v_{1,1},v_{1,2}] + [v_{2,1},v_{2,2}]
        \rightarrow_\beta [v_{1,1} + v_{2,1}, v_{1,2} + v_{2,2}]
    \]
    By the induction hypotheses, $v_{1,j} + v_{2,j} \in \Norm{\cdot \vdashL \alpha_j}$, which completes the proof.

    Finally, if $\alpha = \alpha_1 \lolli \alpha_2$ then $\delta(a_i) \rightarrow^\ast \lambda x.a_i'$ for $a_i' \in \Norm{x:\alpha_1 \vdashL \alpha_2}$. Then
    \[
        \delta(a_1) + \delta(a_2) \rightarrow^\ast (\lambda x.a_1') + (\lambda x.a_2')
        \rightarrow_\beta \lambda x.(a_1' + a_2')
    \]
    By the induction hypothesis for $\alpha_2$, we have that $\alpha_1' + \alpha_2' \in \Norm{x:\alpha_1 \vdashL \alpha_2}$, which completes the proof.
\end{proof}

\begin{lemma}[Compatibility of scalar multiplication]
    If $a_1 \in \Norm{\Delta_1 \vdashL \R}$ and $a_2 \in \Norm{\Delta_2 \vdashL \alpha}$ then $a_1 \cdot a_2 \in \Norm{\Delta_1,\Delta_2 \vdashL \alpha}$.
\end{lemma}
\begin{proof}
    It suffices to only consider closed expressions:
    If $a_1 \in \Norm{\R}$ and $a_2 \in \Norm{\alpha}$ then $a_1 \cdot a_2 \in \Norm{\alpha}$.

    In that case, we know there exist values $r_1 \in \Zd$ and $\cdot \vdashL v_2 : \alpha$ such that 
    \[
        a_1 \cdot a_2 \rightarrow^\ast r_1 \cdot v_2
    \]
    We proceed by induction on $\alpha$.

    If $\alpha=\R$ then $v_2$ is also a constant in $\Zd$, in which case $r_1 \cdot v_2 \rightarrow_\beta r'$ where $r' = r_1 v_2 \in \Zd$.

    If $\alpha = \alpha_1 \oplus \alpha_2$ on the other hand, then $v_2 = [v_{2,1},v_{2,2}]$ in which case
    \[
        r_1 \cdot [v_{2,1},v_{2,2}]
            \rightarrow_\beta [r_1 \cdot v_{2,1}, r_2 \cdot v_{2,2}]
    \]
    Note that by induction, we know that $r_1 \cdot v_{2,1} \in \Norm{\alpha_1}$ and $r_1 \cdot v_{2,1} \in \Norm{\alpha_1}$;
    as a result of \cref{lem:norm-compatibility-plus} we can thus conclude that 
    $[r_1 \cdot v_{2,1}, r_2 \cdot v_{2,2}] \in \Norm{\alpha_1 \oplus \alpha_2}$.

    Finally, if $\alpha=\alpha_1 \lolli \alpha_2$ then $v_2 = \lambda x.a_2$ and so
    \[
        r_1 \cdot \lambda x.a_2 \rightarrow \lambda x.r_1 \cdot a_2.
    \]
    Now, since $\lambda x.a_2 \in \Val{\alpha_1 \lolli \alpha_2}$ it must be the case that $a_2 \in \Norm{x:\alpha_1 \vdashL \alpha_2}$.
    So by the induction hypothesis, $x \cdot a_2 \in \Norm{x:\alpha_1 \vdashL \alpha_2}$ and so $\lambda x.r \cdot a_2 \in \Val{\alpha_1 \lolli \alpha_2}$
\end{proof}

\begin{lemma}[Compatibility of case]
    Let $a \in \Norm{\alpha_1 \oplus \alpha_2}$ and suppose that $a_i \in \Norm{v_i:\alpha_i \vdashL \alpha'}$.
    Then $\caseof{a}{\inl{x_1} \rightarrow a_1 \mid \inr{x_2} \rightarrow a_2} \in \Norm{\alpha'}$.
    
    Let $a \in \Norm{\Delta \vdashL \alpha_1 \oplus \alpha_2}$ and suppose that $a_i \in \Norm{\Delta,x_i:\alpha_i \vdashL \alpha'}$.
    Then 
    \[\caseof{a}{\inl{x_1} \rightarrow a_1 \mid \inr{x_2} \rightarrow a_2} \in \Norm{\Delta \vdashL \alpha'}\]
\end{lemma}
\begin{proof}
    \proofpart[Closed expressions]
    Since $a \in \Norm{\alpha_1 \oplus \alpha_2}$ we know there exists values $\cdot \vdashL v_i : \alpha_i$ such that $a \rightarrow^\ast [v_1,v_2]$. Then
    \begin{align*}
        \caseof{a}{\inl{x_1} \rightarrow a_1 \mid \inr{x_2} \rightarrow a_2}
        &\rightarrow^\ast
        \caseof{[v_1,v_2]}{\inl{x_1} \rightarrow a_1 \mid \inr{x_2} \rightarrow a_2} \\
        &\rightarrow_\beta
        a_1\{v_1/x_1\} + a_2\{v_2/x_2\}
    \end{align*}
    By the assumption in the lemma statement, we know that $a_i\{v_i/x_i\} \in \Norm{\alpha'}$, and so the result follows from the compatibility of $+$ (\cref{lem:norm-compatibility-plus}).

    \proofpart[Open expressions] Let $\delta \in \ValRelation{\Delta}$.
    We want to show that
    \begin{align*}
        &\delta\left(\caseof{a}{\inl{x_1} \rightarrow a_1 \mid \inr{x_2} \rightarrow a_2}\right)
        \\
        &= \caseof{\delta(a)}{\inl{x_1} \rightarrow \delta(a_1) \mid \inr{x_2} \rightarrow \delta(a_2)} \in \Norm{\Delta \vdashL \alpha'}
    \end{align*}
    Then clearly $\delta(a) \in \Norm{\alpha_1 \oplus \alpha_2}$. I further claim that $\delta(a_i) \in \Norm{x_i:\alpha_i \vdashL \alpha'}$, since for any
    $\delta_i \in \ValRelation{x_i:\alpha_i}$ we have $\delta_i(\delta(a_i)) = (\delta,x_i \mapsto \delta_i(x))(a_i) \in \Norm{\alpha'}$ where $\delta,x\mapsto \delta_i(x) \in \ValRelation{\Delta,v_i:\alpha_i}$. The result follows from the first half of this proof.
\end{proof}





\begin{lemma}[Compatibility of $\lambda$]
    If $a \in \Norm{\Delta,x:\alpha \vdashL \alpha'}$ then $\lambda x.a \in \Norm{\Delta \vdashL \alpha \lolli \alpha'}$.
\end{lemma}
\begin{proof}
    Let $\delta \in \ValRelation{\Delta}$. We want to show that $\delta(\lambda x.a) = \lambda x.\delta(a) \in \Val{\alpha \lolli \alpha'}$; in other words, that for all $v \in \Val{\alpha}$ we have $\delta(a)\{v/x\} \in \Norm{\alpha'}$. This follows directly from the hypothesis that $a \in \Norm{\Delta,x:\alpha \vdashL \alpha'}$.
\end{proof}

\begin{lemma}[Compatibility of application]
    If $a_1 \in \Norm{\Delta_1 \vdashL \alpha \lolli \alpha'}$ and $a_1 \in \Norm{\Delta_2 \vdashL \alpha}$ then $a_1 a_2 \in \Norm{\Delta_1,\Delta_2 \vdashL \alpha'}$.
\end{lemma}
\begin{proof}
    Let $\delta \in \ValRelation{\Delta_1,\Delta_2}$. Since $a_1 \in \Norm{\Delta_1 \vdashL \alpha \lolli \alpha'}$ there exists a value $\lambda x.a_1' \in \Val{\alpha \lolli \alpha'}$ such that $\delta(a_1) \rightarrow^\ast \lambda x.a_1'$.
    In particular, this means $a_1' \in \Norm{x:\alpha \vdashL \alpha'}$. Similarly there eixsts some $v \in \Val{\alpha}$ such that
    $\delta(a_2) \rightarrow^\ast v$. Then
    \[
        \delta(a_1 a_2) \rightarrow^\ast (\lambda x.a_1') v \rightarrow_\beta a_1'\{v/x\} \in \Norm{\alpha'}
    \]
\end{proof}

\begin{lemma} \label{lem:norm1}
    If $\Delta \vdashL a : \alpha$ then $a \in \Norm{\Delta \vdashL \alpha}$.
\end{lemma}
\begin{proof}
    Follows from the compatibility lemmas by induction on the typing judgment.
\end{proof}

Finally, we can prove our main theorem about strong normalization:

\begin{theorem}[Strong normalization, \cref{thm:normalization}]
    If $\cdot \vdashL a : \alpha$ then there exists a unique value $v$ (up to the usual notions of $\alpha$-equivalence) such that $a \rightarrow^\ast v$.
\end{theorem}
\begin{proof}
    If $\cdot \vdashL a : \alpha$ then, by \cref{lem:norm1} we know $a \in \Norm{\alpha}$ so, by definition, there exists a value $v$ such that $a \rightarrow^\ast v$.
    The uniqueness of $v$ comes from the fact that the step relation is deterministic.
\end{proof}

\subsection{Equivalence Relation}
\label{app:lambdaC-equivalence-relation}

In this section we prove the fundamental property of $\lrequiv$ (\cref{thm:lrequiv-fundamental-property}), which is that $\Delta \vdashL a : \alpha$ implies $\Delta \vdashL a \lrequiv a : \alpha$. The property follows directly from \cref{lem:lrequiv-compatability}, which proves that the compatibility lemmas shown in \cref{fig:r-modules-compatibility} are sound.

\begin{lemma} \label{lem:deterministic-ExprRelation}
    If $a_1 \rightarrow^\ast a_2$ and $a_1' \rightarrow^\ast a_2'$ and $(a_2,a_2') \in \ExpRelation{\alpha}$ then $(a_1,a_1') \in \ExpRelation{\alpha}$.
\end{lemma}
\begin{proof}
    A consequence of strong normalization (\cref{thm:normalization}) is that $a_1 \rightarrow^\ast a_2$ implies that $a_1 \rightarrow^\ast v$ if and only if $a_2 \rightarrow^\ast v$. Since $(a_2,a_2') \in \ExpRelation{\alpha}$, this means there exist values $(v,v') \in \ValRelation{\alpha}$ such that $a_2 \rightarrow^\ast v$ and $a_2' \rightarrow^\ast v'$. But then it is also the case that $a_1 \rightarrow^\ast v$ and $a_1' \rightarrow^\ast v'$, which completes the proof.
\end{proof}

\begin{figure}
    \centering
    \begin{mathpar}
        \inferrule*[right=$\lrequiv$-var]
            {~}
            {x:\alpha \vdashL x \lrequiv x : \alpha}
            
        \inferrule*[right=$\lrequiv$-let]
            {\Delta_1 \vdashL a_1 \lrequiv  a_1' : \alpha \\
            \Delta_2,x:\alpha \vdashL a_2 \lrequiv a_2' : \alpha'
            }
            {\Delta_1,\Delta_2 \vdashL \letin{x}{a_1}{a_2} \lrequiv \letin{x}{a_1'}{a_2'} : \alpha'
            }
            
        \inferrule*[right=$\lrequiv$-const]
            {r \in \Zd}
            {\cdot \vdashL r \lrequiv r : \R}

        \inferrule*[right=$\lrequiv$-$\cdot$]
            {\Delta_1 \vdashL a_1 \lrequiv  a_1' : \R \\
            \Delta_2 \vdashL a_2 \lrequiv a_2' : \alpha
            }
            {\Delta_1,\Delta_2 \vdashL a_1 \cdot a_2 \lrequiv a_1' \cdot a_2' : \alpha}

        \inferrule*[right=$\lrequiv$-$\zero$]
            {~}
            {\Delta \vdashL \zero \lrequiv \zero : \alpha}

        \inferrule*[right=$\lrequiv$-$+$]
            {\Delta \vdashL a_1 \lrequiv a_1' : \alpha \\
             \Delta \vdashL a_2 \lrequiv a_2' : \alpha
            }
            {\Delta \vdashL a_1 + a_2 \lrequiv a_1' + a_2' : \alpha}

        \inferrule*[right=$\lrequiv$-$\oplus$]
            {\Delta \vdashL a_1 \lrequiv a_1' : \alpha_1 \\
             \Delta \vdashL a_2 \lrequiv a_2' : \alpha_2
            }
            {\Delta \vdashL [a_1, a_2] \lrequiv [a_1', a_2'] : \alpha_1 \oplus \alpha_2}

        \inferrule*[right=$\lrequiv$-$\oplus$]
            {\Delta \vdashL a \lrequiv a' : \alpha_1 \oplus \alpha_2 \\
             \Delta',x_1:\alpha_1 \vdashL a_1 \lrequiv a_1' : \alpha' \\
             \Delta',x_2:\alpha_2 \vdashL a_2 \lrequiv a_2' : \alpha'
            }
            {\Delta, \Delta' \vdashL \caseof{a}{\inl{x_1} \rightarrow a_1 \mid \inr{x_2} \rightarrow a_2} 
                    \lrequiv  \caseof{a'}{\inl{x_1} \rightarrow a_1' \mid \inr{x_2} \rightarrow a_2'}
                    : \alpha'}

        \inferrule*[right=$\lrequiv$-$\lambda$]
            {\Delta,x:\alpha \vdashL a \lrequiv a' : \alpha'}
            {\Delta \vdashL \lambda x.a \lrequiv \lambda x.a' : \alpha \lolli \alpha'}

        \inferrule*[right=$\lrequiv$-app]
            {\Delta_1 \vdashL a_1 \lrequiv a_1' : \alpha \lolli \alpha' \\
             \Delta_2 \vdashL a_2 \lrequiv a_2' : \alpha
            }
            {\Delta_1,\Delta_2 \vdashL a_1 a_2 \lrequiv a_1' a_2' : \alpha'}
    \end{mathpar}
    \caption{Compatibility lemmas for $\lrequiv$.}
    \Description{A set of ten inference rules describing recursively the fact that, if all subterms of two different expressions are equivalent, then the expressions themselves are equivalent.}
    \label{fig:r-modules-compatibility}
\end{figure}

\begin{lemma}[Compatibility Lemmas] \label{lem:lrequiv-compatability}
    The rules in \cref{fig:r-modules-compatibility} are sound for $\lrequiv$.
\end{lemma}
\begin{proof}
    \proofpart[\textsc{c-let}] \label{proof:lrequiv-compatibility:let}

    Let $(\delta,\delta') \in \ValRelation{\Delta_1,\Delta_2}$; it is thus also true that $(\delta,\delta') \in \ValRelation{\Delta_i}$ for both $\Delta_1$ and $\Delta_2$ individually. By the assumption $a_1 \lrequiv a_1'$, we know that
    there exist $(v,v') \in \ValRelation{\alpha}$ such that $\delta(a_1) \rightarrow^\ast v$ and $\delta'(a_2) \rightarrow^\ast v'$.
    So in that case,
    \begin{align*}
        \delta(\letin{x}{a_1}{a_2})
            &\rightarrow^\ast \letin{x}{v}{\delta(a_2)} \rightarrow \delta(a_2)\{v/x\} \\
        \delta'(\letin{x}{a_1'}{a_2'})
            &\rightarrow^\ast \letin{x}{v'}{\delta'(a_2')} \rightarrow \delta'(a_2')\{v'/x\}
    \end{align*}
    Observe that $((\delta,x \mapsto v), (\delta',x \mapsto v')) \in \ValRelation{\Delta_2,x:\alpha}$, and so
    \begin{align*}
         (\delta(a_2)\{v/x\}, \delta'(a_2)\{v'/x\}) \in \ExpRelation{\alpha'}
    \end{align*}
    by the assumption that $a_2 \lrequiv a_2'$. The result follows, therefore, from \cref{lem:deterministic-ExprRelation}.

    \proofpart[\textsc{c-$\lambda$}] Let $(\delta,\delta') \in \ValRelation{\Delta}$. To show
    $(\lambda x.\delta(a), \lambda x.\delta'(a')) \in \ValRelation{\alpha \lolli \alpha'}$, we need to show that, for all $(v,v') \in \ValRelation{\alpha}$, we have $(\delta(a)\{v/x\}, \delta(a')\{v'/x\}) \in \ExpRelation{\alpha'}$. However, this follows from the assumption that $a \lrequiv a'$ and that $((\delta,x \mapsto v), (\delta',x \mapsto v')) \in \ValRelation{\Delta,x:\alpha}$.

    \proofpart[\textsc{c-app}] Let $(\delta,\delta') \in \ValRelation{\Delta}$. By the assumption that $a_1 \lrequiv a_1'$, we know there exist function values $(\lambda x.a_0, \lambda x.a_0') \in \ValRelation{\alpha \lolli \alpha'}$ such that
    $\gamma(a_1) \rightarrow^\ast \lambda x.a_0$ and $\gamma(a_1') \rightarrow^\ast \lambda x.a_0'$. Furthermore, since    
    $a_2 \lrequiv a_2'$, there must exist values $(v,v') \in \ValRelation{\alpha}$ such that $\gamma(a_2) \rightarrow^\ast v$ and $\gamma(a_2') \rightarrow^\ast v'$. Thus
    \begin{align*}
        \gamma(a_1 a_2) &\rightarrow^\ast (\lambda x.a_0) v \rightarrow a_0\{v/x\} \\
        \gamma(a_1' a_2') &\rightarrow^\ast (\lambda x.a_0') v' \rightarrow a_0'\{v'/x\}
    \end{align*}
    The fact that $(\lambda x.a_0, \lambda x.a_0') \in \ExpRelation{\alpha \lolli \alpha'}$ means exactly that $(a_0\{v/x\},a_0'\{v'/x\}) \in \ExpRelation{\alpha'}$, which is all that is required to show that $(\gamma(a_1 a_2), \gamma(a_1' a_2')) \in \ExpRelation{\alpha'}$.

    \proofpart The proofs of the remaining rules are similar to the other cases.

\end{proof}

\section{Categorical Soundness for \(\lambdaC\)}
\label{appendix:categorical-semantics}
\subsection{Properties of \(\CCat\)}
\label{app:CCat-properties}

The (outer) direct sum $A \oplus A'$ is defined as the set of symbols $\{a \oplus a' \mid a\in A ~\text{and}~a' \in A'\}$ 
with basis $\{b^A_i \oplus b^{A'}_i\}$, and 
with addition and scalar multiplication defined by:
\begin{align*}
    r(a \oplus a') &= (ra) \oplus (ra')
    &
    (a \oplus a') + (b \oplus b') &= (a + b) \oplus (a' + b')
\end{align*}
The rank of $A \oplus A'$ is $\rank{A}+\rank{A'}$,
and it is easy to verify that $A \oplus A'$ is the categorical biproduct of $A$ and $B$.

The tensor product $A\otimes B$ of $\Zd$-modules $A$ and $B$ forms a monoidal product on $\CCat$.
The universal property of the tensor product in $\CCat$ allows us to regard elements of $A\otimes A'$ as $\Zd$-linear combinations of symbols $a\otimes a'$ with $a\in A$ and $b\in A'$, modulo the following relations:
\begin{align*}
    r(a \otimes b) &= (ra) \otimes b = a \otimes (rb) \\
    (a \otimes b) + (a' \otimes b) &= (a + a') \otimes b \\
    (a \otimes b) + (a \otimes b') &= a \otimes (b + b')
\end{align*}
The rank of $A \otimes B$ is multiplicative in the ranks of $A$ and $B$, and the basis of $A \otimes A'$ is the Kronecker basis $\{b^A_i \otimes b^{A'}_j\}$.

The \textit{associator} of $\otimes$ is taken to be the isomorphism $A\otimes(B\otimes C)\stackrel\sim\to(A\otimes B)\otimes C$ given by linearly extending $a\otimes(b\otimes c)\mapsto(a\otimes b)\otimes c$; this is natural in $A$, $B$, and $C$. Likewise, the \textit{unitors} are given by $a\mapsto1\otimes a$ and $a\mapsto a\otimes1$. Finally, we can define an involutive natural isomorphism $A\otimes B\cong B\otimes A$, thus making $(\CCat,\otimes)$ into a \textit{symmetric monoidal category}.

Next, we define $A\lolli A'$ to be the $\Zd$-module of all $\Zd$-linear maps $f:A\to A'$, with addition and scalar multiplication defined elementwise. The character group $A^*$ is defined as $A\lolli\Zd$. 

If $A$ has a basis $\{b^A_i\}$, then $A^\ast$ has a corresponding dual basis $\{(b^A_i)^\ast\}$ given by $(b^A_i)^\ast(b^A_j)=\delta_{i,j}$. Likewise, if $C$ has a basis $\{b^C_j\}$, then $A \lolli C$ has a basis consisting of functions $\{b^A_i \lolli b^C_j\}$ given by $(b^A_i \lolli b^C_j)(b^A_k)=\delta_{i,k} b^C_j$.

For each object $B$, the functor  $B\lolli-:\CCat\to C$ is right adjoint to $-\otimes B:\CCat\to\CCat$, which means that there is a natural bijection $\CCat(A\otimes B,C)\cong\CCat(A,B\lolli C)$ that is natural in $A$ and $C$, as well as $B$.
This implies $(\CCat,\otimes)$ is a symmetric closed monoidal category.
The counit of the currying adjunction is the evaluation map $(B\lolli C)\otimes B\to C$, which is given elementwise by $f \otimes b \mapsto f(b)$.
From this we obtain a natural transformation $B\otimes A^* \to A\lolli B$ defined by
\begin{align}\label{eqn-nuclearizer}
    b\otimes\alpha&\mapsto (a\mapsto\alpha(a)b)
\end{align}
Since \cref{eqn-nuclearizer} takes the Kronecker basis $b_i\otimes a_i^*$ to the function basis $a_i\multimap b_i$, it follows by considering the rank that \cref{eqn-nuclearizer} is an isomorphism. In other words, $\CCat$ is a compact closed category. In such a category, it follows that the canonical natural transformation $A \to A^{**}$ into the double dual is an isomorphism
and that the tensor product is de Morgan self-dual, which means the canonical de Morgan natural transformation $B^*\otimes A^* \to (A\otimes B)^*$ is in fact an isomorphism.

Since left adjoints preserve colimits and the biproduct is a coproduct hence a colimit, there is a natural isomorphism of the following form:
\begin{align*} 
    (A\oplus B)\otimes(C\oplus D)&\stackrel\sim\to A\otimes C\,\,\oplus\,\, A\otimes D\,\,\oplus\,\,B\otimes C\,\,\oplus\,\,B\otimes D\\
    (a\oplus b)\otimes(c\oplus d)&\mapsto a\otimes c\,\,\oplus\,\, a\otimes d\,\,\oplus\,\,b\otimes c\,\,\oplus\,\,b\otimes d
\end{align*}
These features imply that $\CCat$ is a model of multiplicative additive linear logic where $\oplus=\&$ is a biproduct and $\Zd$ is the unit of $\otimes=\parr$~\citep{mellies2009categorical}.

\subsection{Completeness}
\label{app:lambdaC-completeness}

In order to prove completeness of $\lambdaC$ (\cref{thm:lambdaC-completeness}), it suffices to prove the property for basis elements only.

\begin{lemma}\label{lem:values-complete}
    If $a \in \interpL{\alpha}$ is a basis element of the $\Zd$-module $\interpL{\alpha}$, then there exist values $\cdot \vdash \up{a} : \alpha$ and $\cdot \vdash \down{a} : \alpha \lolli \R$ such that
    \begin{align*}
        \interpL{\up{a}}(1) &= a 
        && \interpL{\down{a}}(1) = a^\ast 
            = b \mapsto \begin{cases}
                1 & b=a \\
                0 & \text{otherwise}
            \end{cases}
    \end{align*}
\end{lemma}
\begin{proof}
    By induction on $\alpha$. 
    
    \proofpart[$\alpha=\R$] If $\alpha=\R$ then define $\up{a}=a \in \Zd$. If $a=1$ then define $\down{a}=\lambda x.x$; otherwise define $\down{a}=\lambda x.\zero$. The proof follows straightforwardly from definition.

    \proofpart[$\alpha=\alpha_1 \oplus \alpha_2$]
    If $\alpha=\alpha_1 \oplus \alpha_2$ then $a$ has the form $a_1 \oplus a_2$ where $a_i \in \interpL{\alpha_i}$ and either $a_1=0$ or $a_2=0$. Define $\up{a}=[\up{a_1},\up{a_2}]$ and 
    \[
        \down{a} = \lambda x.\caseof{x}{\inl{x_1} \rightarrow \down{a_1}(x) \mid \inr{x_2} \rightarrow \down{a_2}(x)}.
    \]
    Then
    \begin{align*}
        \interpL{[\up{a_1},\up{a_2}]}(1) &= \interpL{\up{a_1}}(1) \oplus \interpL{\up{a_2}}(1)
        = a_1 \oplus a_2
    \end{align*}
    by the induction hypothesis, and
    \begin{align*}
        \interpL{\down{a}}(1)
            &= (b_1 \oplus b_2) \mapsto \left( \interpL{\down{a_1}} \boxplus \interpL{\down{a_2}} \right)(b_1 \oplus b_2) \\
            &= (b_1 \oplus b_2) \mapsto (a_1^\ast \boxplus a_2^\ast) (b_1 \oplus b_2)
                \tag{induction hypothesis} \\
            &= (b_1 \oplus b_2) \mapsto a_1^\ast(b_1) + a_2^\ast(b_2)
    \end{align*}
    Since we assumed that either $a_1=0$ or $a_2=0$, one of the two terms above will always be zero. For example, in the case of $a_2=0$, the equation above reduces to
    \begin{align*}
        (b_1 \oplus b_2) \mapsto \begin{cases}
            1 & b_1=a_1 \wedge b_2=0 \\
            0 & \text{otherwise}
        \end{cases}
        &= (a_1 \oplus a_2)^\ast
    \end{align*}
    Similar reasoning holds for $a_1=0$.

    \proofpart[$\alpha=\alpha_1 \lolli \alpha_2$]
    If $\alpha = \alpha_1 \lolli \alpha_2$, then $a$ has the form $a_1 \lolli a_2$ for $a_i$ a basis element of $\interpL{\alpha_i}$. Then define 
    \begin{align*}
        \up{a} &= \lambda x. \down{a_1}(x) \cdot \up{a_2} \\
        \down{a} &= \lambda f. \down{a_2}(f(\up{a_1}))
    \end{align*}
    For $\up{a}$, unfolding definitions we can see that
    \begin{align*}
        \interpL{\lambda x.\down{a_1}(x) \cdot \up{a_2}}(1)
        &= b \mapsto (\interpL{\down{a_1}}(b)) (\interpL{\up{a_2}}(1)) \\
        &= b \mapsto (a_1^\ast(b)) (a_2) 
            \tag{induction hypothesis} \\
        &= b \mapsto \begin{cases}
            a_2 & b=a_1 \\
            0 & \text{otherwise}
        \end{cases} \\
        &= a_1 \lolli a_2
    \end{align*}

    Next, consider the semantics of $\down{a}$. Let $b_1 \lolli b_2$ be a basis element of $\interpL{\alpha_1 \lolli \alpha_2}$. Then
    \begin{align*}
        \interpL{\lambda f. \down{a_2}(f(\up{a_1}))}(1)
        &= (b_1 \lolli b_2) \mapsto \interpL{\down{a_2}} ((b_1 \lolli b_2)\interpL{\up{a_1}}(1)) \\
        &=  (b_1 \lolli b_2) \mapsto a_2^\ast ((b_1 \lolli b_2)a_1)
            \tag{induction hypothesis}
    \end{align*}
    By the definition of $(b_1 \lolli b_2)$, this is equal to
    \begin{align*}
        &(b_1 \lolli b_2) \mapsto \begin{cases}
            a_2^\ast b_2
                & a_1=b_1 \\
            a_2^\ast 0
                &\text{otherwise}
        \end{cases} \\
        &=(b_1 \lolli b_2) \mapsto \begin{cases}
            1
                & a_1=b_1 \wedge a_2=b_2 \\
            0
                &\text{otherwise}
        \end{cases} \\
        &= (a_1 \lolli a_2)^\ast
        \qedhere
    \end{align*}
\end{proof}

For the proofs in \cref{app:lambdaC-soundness}, it is also useful to lift basis elements $g \in \interpL{\Delta}$ to substitutions $\up{g}$ of the variables in $\Delta$ for values.

\begin{lemma} \label{lem:contexts-complete}
    For every basis element $g \in \interpL{\Delta}$ there exists a substitution $\up{g} \in \Val{\Delta}$ such that
    $\interpL{\up{g}}(1)=g$, where $\delta \in \Val{\Delta}$ if and only if, for all $x_i:\alpha_i \in \Delta$, there exists a value $\cdot \vdash v_i : \alpha_i$ such that $\delta(x_i)=v_i$.
\end{lemma}
\begin{proof}
    By induction on $\Delta$. If $\Delta=\cdot$ and $g$ is a basis element of $\interpL{\Delta}=\Zd$, then $g=1$. Define $\up{g}=\cdot$.
        Then clearly $\interpL{\up{g}}(1)=g$.

        If $\Delta=\Delta',x:\alpha$ then $g=g' \otimes a$ where $g' \in \interpL{\Delta'}$ and $a \in \interpL{\alpha}$.
        Define $\up{g} = \up{g'},x \mapsto \up{a}$. Then
        \begin{align*}
            \interpL{\up{g}}(1) &= \interpL{\up{g'}}(1) \otimes \interpL{\up{a}}(1) = g' \otimes a = g.
        \end{align*}
\end{proof}

\subsection{Soundness}
\label{app:lambdaC-soundness}

The goal of this section is to prove that if $\Delta \vdash a_1 \lrequiv a_2 : \alpha$, then $\interpL{a_1} = \interpL{a_2}$.

\begin{lemma} \label{lem:environment-interp-append}
    If $\Delta=\Delta_1,\Delta_2$ then $\interpL{\delta}_{\Delta}(1)=\interpL{\delta}_{\Delta_1}(1) \otimes \interpL{\delta}_{\Delta_2}(1)$.
\end{lemma}
\begin{proof}
    By induction on $\Delta_2$.
\end{proof}

Next, we show that substitution corresponds to composition of morphisms in the category.

\begin{lemma} \label{lem:substitution-composition}
    If $\Delta \vdashL a : \alpha$ and $g \in \Val{\Delta}$, then $\interpL{g(a)} = \interpL{a} \circ \interpL{g}$.
\end{lemma}
\begin{proof}
    It is useful to generalize the statement of this lemma to make it easier to prove by induction.
    Suppose $\Delta,\Delta' \vdash a : \alpha$ and $g \in \Val{\Delta}$. Then it suffices to prove that, for all $g' \in \interpL{\Delta'}$ we have
    \begin{align}
        \interpL{a}(\interpL{g}_{\Delta}(1) \otimes g') = \interpL{g(a)}(g').
    \end{align}
    We will proceed by induction on $\Delta,\Delta' \vdash a : \alpha$.

    \proofpart[$a=x$]
    Suppose $\Delta,\Delta' \vdash x : \alpha$. If $x \in \dom(\Delta)$ then $\Delta=x:\alpha$ and $\Delta'=\cdot$.
    Then
    \begin{align*}
        \interpL{a}(\interpL{g}(1) \otimes g')
        &= g' \cdot \interpL{g(x)}(1) = \interpL{g(x)}(g')
    \end{align*}
    by the linearity of $\interpL{g(x)}$.

    On the other hand, if $x \in \dom(\Delta')$ then $\Delta'=x:\alpha$ and $\Delta=\cdot$.
    Then $g$ has no action on $x$ and so
    \begin{align*}
        \interpL{a}(\interpL{g}_{\Delta}(1) \otimes g')
        = \interpL{\cdot,x:\alpha \vdash x : \alpha}(1 \otimes g') = g' = \interpL{x:\alpha \vdash x : \alpha}(g').
    \end{align*}

    \proofpart[$a=\letin{x}{a_1}{a_2}$]
    If $a=\letin{x}{a_1}{a_2}$ then without loss of generality we assume we can write $\Delta=\Delta_1,\Delta_2$ and $\Delta'=\Delta_1',\Delta_2'$ such that $\Delta_1,\Delta_1' \vdash a_1 : \alpha_1$ and $\Delta_2,\Delta_2',x:\alpha_1 \vdash a_2 : \alpha$. Furthermore, it is the case that $g \in \Val{\Delta_1}$ and $g \in \Val{\Delta_2}$. A basis element of $\interpL{\Delta'}$ has the form $g_1' \otimes g_2'$ for $g_i' \in \interpL{\Delta_i'}$. Then
    \begin{align*}
        \interpL{a}(\interpL{g}_{\Delta}(1) \otimes (g_1' \otimes g_2'))
        &= \interpL{a}\left(
            \left(\interpL{g}_{\Delta_1}(1) \otimes \interpL{g}_{\Delta_2}(1)\right)
            \otimes
            \left(g_1' \otimes g_2'\right)
            \right) \\
        &= \interpL{a_2}\left(
                \interpL{g}_{\Delta_2}(1)
                \otimes
                g_2'
                \otimes
                \interpL{a_1}\left( \interpL{g}_{\Delta_1}(1) \otimes g_1' \right)
            \right)
    \end{align*}
    By the induction hypothesis for $a_1$, this is equal to
    \begin{align*}
        \interpL{a_2}\left(
                \interpL{g}_{\Delta_2}(1)
                \otimes
                g_2'
                \otimes
                \interpL{g(a_1)}\left( g_1' \right)
            \right)
    \end{align*}
    and by the induction hypothesis for $a_2$, equal to
    \begin{align*}
        \interpL{g(a_2)}\left(
                g_2'
                \otimes
                \interpL{g(a_1)}\left( g_1' \right)
            \right)
        &= \interpL{\letin{x}{g(a_1)}{g(a_2)}}(g_1' \otimes g_2') = \interpL{g(a)}(g_1' \otimes g_2').
    \end{align*}

\proofpart[$a=r$] If $a=r$ is a scalar then both $\Delta$ and $\Delta'$ are empty, in which case
    \begin{align*}
        \interpL{a}(\interpL{g}_{\cdot}(1) \otimes g')
        &= \interpL{r}(1 \otimes g') = r \cdot g' = \interpL{g(r)}(g')
    \end{align*}

\proofpart[$a=a_1 \cdot a_2$]
    We can assume $\Delta=\Delta_1,\Delta_2$ and $\Delta'=\Delta_1',\Delta_2'$ such that
    $\Delta_1,\Delta_1' \vdash a_1 : \R$ and $\Delta_2, \Delta_2' \vdash a_2 : \alpha$.
    Furthermore, we can write $g'=g_1' \otimes g_2'$ for $g_1' \in \interpL{\Delta_1'}$ and $g_2' \in \interpL{\Delta_2'}$
    Then
    \begin{align*}
        \interpL{a_1 \cdot a_2}\left(
                \interpL{g}_{\Delta} \otimes (g_1' \otimes g_2')
            \right)
        &= \interpL{a_1 \cdot a_2}\left(
                \left(\interpL{g}_{\Delta_1} \otimes \interpL{g}_{\Delta_2}\right)
                \otimes
                \left(g_1' \otimes g_2'\right)
            \right)
            \tag{\cref{lem:environment-interp-append}}
            \\
        &= \interpL{a_1}\left( \interpL{g}_{\Delta_1} \otimes g_1' \right)
            \cdot
            \interpL{a_2}\left( \interpL{g}_{\Delta_2} \otimes g_2' \right)
        \\
        &= \interpL{g(a_1)}(g_1')
            \cdot
            \interpL{g(a_2)}(g_2')
            \tag{induction hypothesis} \\
        &= \interpL{g(a_1) \cdot g(a_2)}(g_1' \otimes g_2')
    \end{align*}

\proofpart[$a=\zero$] Trivial as both $a$ and $g(a)$ are the zero morphism.

\proofpart[$a=a_1 + a_2$]
    It must be the case that $\Delta \vdash a_i : \alpha$, and so
    \begin{align*}
        \interpL{a_1+a_2}(\interpL{g}_{\Delta} \otimes g')
        &= \interpL{a_1}(\interpL{g}_{\Delta} \otimes g')
            +
            \interpL{a_2}(\interpL{g}_{\Delta} \otimes g')
            \\
        &= \interpL{g(a_1)}(g') + \interpL{g(a_2)}(g')
            \tag{induction hypothesis} \\
        &= \interpL{g(a_1) + g(a_2)}(g')
    \end{align*}


\proofpart[\ensuremath{a=[a_1,a_2]}] Similar to the previous case:
    \begin{align*}
        \interpL{[a_1,a_2]}(\interpL{g}_{\Delta}(1) \otimes g')
        &= \left( \interpL{a_1}(\interpL{g}_{\Delta}(1) \otimes g') \right)
            \oplus
            \left( \interpL{a_1}(\interpL{g}_{\Delta}(1) \otimes g') \right)
        \\
        &= \interpL{g(a_1)}( g')
            \oplus
            \interpL{g(a_2)}( g')
            \tag{induction hypothesis}
        \\
        &= \interpL{[g(a_1), g(a_2)]}(g')
    \end{align*}

\proofpart[$a=\caseof{a'}{\inl{x_1} \rightarrow a_1 \mid \inr{x_2} \rightarrow a_2}$]
    Without loss of generality, write $\Delta=\Delta_1,\Delta_2$ and $\Delta'=\Delta_1',\Delta_2'$ such that
    $\Delta_1,\Delta_1' \vdash a' : \alpha_1 \oplus \alpha_2$ and $\Delta_2,\Delta_2',x_i:\alpha_i \vdash a_i : \alpha$.
    Let $g_i' \in \interpL{\Delta_i'}$. Then:
    \begin{align}
        \interpL{a}\left(\interpL{g}_{\Delta_1,\Delta_1'}(1) \otimes (g_1' \otimes g_2')\right)
        &= \interpL{a}\left( \left(\interpL{g}_{\Delta_1} \otimes \interpL{g}_{\Delta_2}\right)(1)
                            \otimes 
                            (g_1' \otimes g_2')
                     \right)
            \notag
            \\
        &= \interpL{a_1}\left(\interpL{g}_{\Delta_2}(1) \otimes g_2' \otimes c_1\right)
            +
            \interpL{a_2}\left(\interpL{g}_{\Delta_2}(1) \otimes g_2' \otimes c_2\right)
            \label{eqn:substitution-sound-case}
    \end{align}
    where $\interpL{a'}(\interpL{g}_{\Delta_1}(1) \otimes g_1') = c_1 \oplus c_2$.
    By the induction hypothesis for $a'$, we also know that $\interpL{g(a')}(g_1')=c_1 \oplus c_2$. Then, by the induction hypotheses for $t_1$ and $t_2$, \cref{eqn:substitution-sound-case} is equal to
    \begin{align*}
        &\interpL{g(a_1)}\left(g_2' \otimes c_1\right)
        +
        \interpL{g(a_2)}\left(g_2' \otimes c_2\right)
        \\
        &=
        \interpL{\caseof{g(a')}{\inl{x_1} \rightarrow g(a_1) \mid \inr{x_2} \rightarrow g(a_2)}}(g_1' \otimes g_2') \\
        &=
        \interpL{g(a)}(g_1' \otimes g_2')
    \end{align*}

\proofpart[$a=\lambda x.a'$]
    We want to show that $\interpL{\lambda x.g(a')}(g') = \interpL{\lambda x.a'}(\interpL{g}(1) \otimes g')$.
    Unfolding definitions:
    \begin{align*}
        \interpL{\lambda x.g(a')}(g')
        &= \sum_b \delta_b \parr \interpL{g(a')}(g' \otimes b) \\
        &= \sum_b \delta_b \parr \interpL{a'}(\interpL{g}(1) \otimes g' \otimes b) 
            \tag{induction hypothesis} \\
        &= \interpL{\lambda x.a'}(\interpL{g}(1) \otimes g')
    \end{align*}

\proofpart[$a=a_1 a_2$] If $\Delta,\Delta' \vdash a_1 a_2 : \alpha'$, then without loss of generality we can write $\Delta=\Delta_1,\Delta_2$ and $\Delta'=\Delta_1',\Delta_2'$ such that $\Delta_1,\Delta_1' \vdashL a_1 : \alpha \lolli \alpha'$ and $\Delta_2,\Delta_2' \vdashL a_2 : \alpha$. Further, we can write $g'=g_1' \otimes g_2'$ and $\interpL{g}_{\Delta}(1) = \interpL{g_1}_{\Delta_1}(1) \otimes \interpL{g_2}_{\Delta_2}(1)$. Then
    \begin{align*}
    \interpL{a_1 a_2}(\interpL{g}_{\Delta_1,\Delta_2}(1) \otimes g_1' \otimes g_2')
        &=  \interpL{a_1 a_2}(\interpL{g}_{\Delta_1}(1) \otimes \interpL{d}_{\Delta_2}(1) \otimes g_1' \otimes g_2' \\
        &=\left(\interpL{a_1}(\interpL{g}_{\Delta_1}(1) \otimes g_1')\right)
            \left(\interpL{a_2}(\interpL{g}_{\Delta_2}(1) \otimes g_2')\right) \\
        &=\left(\interpL{g(a_1)}(g_1')\right)
            \left(\interpL{g(a_2)}(g_2')\right)
             \tag{induction hypothesis} \\
        &= \interpL{g(a_1 a_2)}(g_1' \otimes g_2')
    \end{align*}
\end{proof}

Next we show that the operational semantics is preserved by the categorical semantics.

\begin{theorem} \label{thm:step-sound}
    If $\cdot \vdash a : \alpha$ and $a \rightarrow a'$ then $\interpL{a}=\interpL{a'}$.
\end{theorem}
\begin{proof}
    It suffices to show the result for just the $\beta$-reduction rules.

    \proofpart[$\letin{x}{v}{a'} \rightarrow_\beta a'\{v/x\}$]
    Unfolding definitions, we have that
    \begin{align*}
        \interpL{\letin{x}{v}{a'}}(1)
        &= \interpL{a'}(1 \otimes \interpL{v}(1))
        = \interpL{a'\{v/x\}}(1)
    \end{align*}
    by \cref{lem:substitution-composition}.
    
    \proofpart[$\zero_{\R} \rightarrow_\beta 0$]
    Follows by unfolding definitions---both produce the zero map.

    \proofpart[\ensuremath{\zero_{\alpha_1 \oplus \alpha_2} \rightarrow_\beta [\zero_{\alpha_1}, \zero_{\alpha_2}]}]
    It suffices to see that if $\zero_{\alpha_i}$ is the zero map, then $\interpL{[\zero_{\alpha_1}, \zero_{\alpha_2}]}$ is the zero map on all inputs.

    \proofpart[\ensuremath{\zero_{\alpha \lolli \alpha'} \rightarrow_\beta \lambda x.\zero_{\alpha'}}]
    Similarly, it suffices to see that $\interpL{\lambda x.\zero}(1)$ is the zero map.

    \proofpart[$r_1 \cdot r_2 \rightarrow_\beta r'=r_1r_2 \in \R$] Trivial.

    \proofpart[\ensuremath{r \cdot [v_1,v_2] \rightarrow_\beta [r \cdot v_1, r \cdot v_2]}]
    \begin{align*}
        \interpL{r \cdot [v_1,v_2]}(1)
        &= r \cdot \left( \interpL{v_1}(1) \oplus \interpL{v_2}(1) \right) \\
        &= \left( r \cdot \interpL{v_1}(1) \right)
            \oplus 
            \left( r \cdot \interpL{v_2}(1) \right) \\
        &= \interpL{[r \cdot v_1, r \cdot v_2]}(1)
    \end{align*}

    \proofpart[$r \cdot \lambda x.a \rightarrow_\beta \lambda x. r \cdot a$]
    \begin{align*}
        \interpL{r \cdot \lambda x.a}(1)
            &= r \interpL{\lambda x.a}(1) \\
            &= r \interpL{a} = \interpL{\lambda x.r \cdot a}(1)
    \end{align*}

    \proofpart[$r_1+r_2 \rightarrow_\beta r'=r_1+r_2 \in \Zd$] Trivial.

    \proofpart[\ensuremath{[v_1,v_2]+[v_1',v_2'] \rightarrow_\beta [v_1+v_1',v_2+v_2']}]
    \begin{align*}
        \interpL{[v_1,v_2]+[v_1',v_2']}(1)
        &= \interpL{[v_1,v_2]}(1) + \interpL{[v_1',v_2']}(1) \\
        &= \left( \interpL{v_1}(1) \oplus \interpL{v_2}(1) \right)
            +
            \left( \interpL{v_1'}(1) \oplus \interpL{v_2'}(1) \right) \\
        &= \left( \interpL{v_1}(1) + \interpL{v_1'}(1)\right)
            \oplus
            \left( \interpL{v_2}(1) + \interpL{v_2'}(1) \right) \\
        &= \left( \interpL{v_1 + v_1'}(1)\right)
            \oplus
            \left( \interpL{v_2 + v_2'}(1) \right) \\
        &= \interpL{[v_1+v_1', v_2+v_2']}(1)
    \end{align*}

    \proofpart[$(\lambda x.a_1) + (\lambda x.a_2) \rightarrow_\beta \lambda x.a_1 + a_2$]
    \begin{align*}
        \interpL{(\lambda x.a_1) + (\lambda x.a_2)}(1)
        &= \left( b \mapsto \interpL{a_1}(b)\right)
            +
            \left( b \mapsto \interpL{a_2}(b) \right) \\
        &= b \mapsto  \interpL{a_1}(b) +  \interpL{a_2}(b) \\
        &= \interpL{a_1 + a_2}(1)
    \end{align*}
    by linearity.

    \proofpart[$\caseof{[v_1,v_2]}{\inl{x_1} \rightarrow a_1 \mid \inr{x_2} \rightarrow a_2} \rightarrow_\beta a_1\{v_1/x_1\} + a_2\{v_2/x_2\}$]
    From \cref{fig:cat-semantics} we know that
    \begin{align}
        \interpL{\caseof{[v_1,v_2]}{\inl{x_1} \rightarrow a_1 \mid \inr{x_2} \rightarrow a_2}}(1)
        &= \interpL{a_1}(c_1) + \interpL{a_2}(c_2)
        \label{eqn:beta-sound-case}
    \end{align}
    where $\interpL{[v_1,v_2]}(1) =\interpL{v_1}(1) \oplus \interpL{v_2}(1) = c_1 \oplus c_2$.
    Then, by \cref{lem:substitution-composition} we know 
    \begin{align*}
        \interpL{a_i}(c_i) = \interpL{a_i}(\interpL{v_i}(1)) = \interpL{a_i\{v_i/x_i\}}(1)
    \end{align*}
    Thus, \cref{eqn:beta-sound-case} is equal to
    \begin{align*}
        \interpL{a_1\{v_1/x_1\}}(1) + \interpL{a_2\{v_2/x_2\}}(1)
        &= \interpL{a_1\{v_1/x_1\} + a_2\{v_2/x_2\}}(1)
    \end{align*}
    as expected.

    \proofpart[$(\lambda x.a) v \rightarrow_\beta a\{v/x\}$]
    \begin{align*}
        \interpL{(\lambda x.a) v}(1)
        &= \interpL{a}(\interpL{v}(1)) \\
        &= \interpL{a\{v/x\}}(1) 
            \tag{\cref{lem:substitution-composition}}
    \end{align*}
    
\end{proof}

To prove soundness, we first need to establish soundness for the logical relation on values.

\begin{lemma} \label{lem:val-relation-sound}
    If $(v_1,v_2) \in \ValRelation{\alpha}$ then $\interpL{v_1}=\interpL{v_2}$.
\end{lemma}
\begin{proof}
    By induction on $\alpha$. If $\alpha=\R$ or $\alpha=\alpha_1 \oplus \alpha_2$, the result is easy to derive from definitions.

    In the case that $\alpha=\alpha_1 \lolli \alpha_2$, we know we can write $v_i=\lambda x.a_i$ such that, for all $(v_1',v_2') \in \ValRelation{\alpha_1}$, we have $(a_1\{v_1'/x\}, a_2\{v_2'/x\}) \in \ExpRelation{\alpha_2}$.
    In other words, there exist $(v_1'',v_2'') \in \ValRelation{\alpha_2}$ such that
    $a_i\{v_i'/x\} \rightarrow^\ast v_i''$. By the induction hypothesis we know that $\interpL{v_1''}=\interpL{v_2''}$, and so by \cref{thm:step-sound} we know that
    \[
        \interpL{a_1\{v_1'/x\}} = \interpL{v_1''} = \interpL{v_2''} = \interpL{a_2\{v_2'/x\}}.
    \]
    To show $\interpL{\lambda x.a_1}=\interpL{\lambda x.a_2}$, it suffices to show that for all basis elements $b$ of $\interpL{\alpha_1}$, we have
    $\interpL{a_1}(b)=\interpL{a_2}(b)$. This follows from \cref{lem:substitution-composition} and the fact that $(\up{b},\up{b}) \in \ValRelation{\alpha_1}$:
    \begin{align*}
        \interpL{a_1}(b) = \interpL{a_1\{\up{b}/x\}}=\interpL{a_2\{\up{b}/x\}}=\interpL{a_2}(b).
    \end{align*}
\end{proof}

Finally, we can prove the main soundness theorem. 

\begin{theorem}[\cref{thm:categorical-equivalence-sound}]
    If $\Delta \vdash t_1 \lrequiv t_2 : \alpha$ then $\interpL{t_1}=\interpL{t_2}$.
\end{theorem}
\begin{proof}
    To show $\interpL{t_1}=\interpL{t_2}$, it suffices to show $\interpL{t_1}(g)=\interpL{t_2}(g)$ for all basis elements $g \in \interpL{\Delta}$. By \cref{lem:values-complete}, for each such $g$ there exists some $\up{g} \in \ValRelation{\Delta}$ such that
    $\interpL{\up{g}}(1)=g$. Then
    \begin{align*}
        \interpL{t_i}(g)&=\interpL{t_i}(\interpL{\up{g}}(1)) \\
        &= \interpL{\up{g}(t_i)}(1) \tag{\cref{lem:substitution-composition}}
    \end{align*}
    so it suffices to show  $\interpL{\up{g}(t_1)} = \interpL{\up{g}(t_2)}$.
    Because $\Delta \vdash t_1 \lrequiv t_2 : \alpha$, we know $(\up{g}(t_1),\up{g}(t_2)) \in \ExpRelation{\alpha}$. Thus, the result follows from \cref{lem:val-relation-sound}.
\end{proof}

\subsection{Completeness of the Equivalence Relation}
\label{app:lambdaC-completeness-2}

Finally, we can prove completeness of  the equivalence relation. Below is a slightly strengthened statement of \cref{thm:lambdaC-lrequiv-completeness}.

\begin{theorem}[Completeness of $\lrequiv$] \label{app:lambdaC-lrequiv-completeness}
~
  \begin{enumerate}
      \item If $\Delta \vdashC a_1,a_2 : \alpha$ such that $\interpL{a_1} = \interpL{a_2}$, then $\Delta \vdashC a_1 \lrequiv a_2 : \alpha$.
      \item If $v_1,v_2$ are values of type $\alpha$ such that $\interpL{v_1} = \interpL{v_2}$, then $(v_1,v_2) \in \Val{\alpha}$.
  \end{enumerate}
\end{theorem}
\begin{proof}
    \proofpart
    Let $\delta$ be a value map for $\Delta$. 
    By strong normalization, there exist values $v_i$ such that $\delta(a_i) \rightarrow^\ast v_i$.
    We prove in \cref{lem:substitution-composition} that $\interpL{\delta(a_i)} =\interpL{a_i} \circ \interpL{\delta}$. Then, from the soundness of $\rightarrow^\ast$:
    \[
        \interpL{v_1} = \interpL{a_1} \circ \interpL{\delta}
        = \interpL{a_2} \circ \interpL{\delta} = \interpL{v_2}.
    \]
    and the result follows from \cref{case:lrequiv-completeness-values}.

    \proofpart \label{case:lrequiv-completeness-values}
    By induction on $\alpha$.

    The only nontrivial case is when $\alpha=\alpha_1 \lolli \alpha_2$, in which case we can write $v_i=\lambda x.a_i$. Then $\interpL{a_1}=\interpL{a_2}$, and so by the induction hypothesis, $x:\alpha_1 \vdashC a_1 \lrequiv a_2 : \alpha_2$, which completes the proof.
\end{proof}

\section{Categorical Semantics for \(\lambdaP\)}
\label{appendix:categorical-semantics-pauli}

\jennifer{In this section the proofs are sound, but could be made more elegant/unified.}

In this section we aim to prove that the operational semantics of $\lambdaP$ is sound with respect to the categorical semantics into $\PCat$.

\subsection{Categorical Structure of \(\PCat\)}

The fact that $\PCat$ forms a category at all can be derived from the equivalence between condensed encodings and projective Cliffords as described in \cref{sec:background}. Specifically, the fact that composition of encodings in $\PCat$ corresponds to composition of projective Cliffords implies that composition is associative and respects the identity.

Next we will establish some useful lemmas about the categorical structure of morphisms in $\PCat$, including the constructions introduced in \cref{sec:lambdaP-categorical-semantics}.

\begin{lemma} \label{lem:config-composition-0}
    If $(\mu,\psi) \in \PCat(\Pobj{V},\Pobj{V'})$ and $(\mu',\psi') \in \PCat(\alpha,\Pobj{V})$ then
    \begin{align*}
        (\mu',\psi') \circ (\mu,\psi) = \config{\mu} (\mu',\psi') \circ (0,\psi)
    \end{align*}
\end{lemma}
\begin{proof}
    By definition of composition, $(\mu',\psi') \circ (\mu,\psi) = (\mu_0, \psi' \circ \psi)$ where $\mu_0(b)=\mu(b) + \mu'(\psi(b)) + \kappa^{\psi'}(\psi(b))$. Similarly, $(\mu',\psi') \circ (0,\psi) = (\mu_0', \psi' \circ \psi)$ where $\mu_0'(b)=\mu'(\psi(b)) + \kappa^{\psi'}(\psi(b))$. Thus, it is clear that 
    \begin{align*}
        \config{\mu} (\mu',\psi') \circ (0,\psi)
        &= (\mu + \mu_0', \psi' \circ \psi) 
        = (\mu',\psi') \circ (\mu,\psi) 
    \end{align*}
\end{proof}

\begin{lemma} \label{lem:config-composition}
    Let $(\mu,\psi) \in \PCat(\alpha,\Pobj{V})$ and $(\mu',\psi') \in \PCat(\Pobj{V},\Pobj{V'})$, and let $a \in \CCat(V,\Zd)$.
    Then
    \begin{align*}
        (\config{a}(\mu',\psi')) \circ (\mu,\psi) = \config{a \circ \psi} (\mu',\psi') \circ (\mu,\psi)
    \end{align*}
\end{lemma}
\begin{proof}
    On the one hand, $(\mu',\psi') \circ (\mu,\psi) = (\mu_0,\psi' \circ \psi)$
    where $\mu_0$ is defined on basis elements by
    \begin{align*}
        \mu_0(b) = \mu(b) + \mu'(\psi(b)) + \kappa^{\psi'}(\psi(b))
    \end{align*}
    On the other hand,
    \begin{align*}
        (\config{a}(\mu',\psi')) \circ (\mu,\psi)
        &= (a + \mu', \psi') \circ (\mu,\psi)
        = (\mu_0',\psi' \circ \psi)
    \end{align*}
    where 
    \begin{align*}
        \mu_0'(b)=\mu(b) + (a + \mu')(\psi(b)) + \kappa^{\psi'}(\psi(b))
    \end{align*}
    Thus $\mu_0'=\mu_0 + a \circ \psi$, which completes the proof.
\end{proof}

Next we can prove that $\iota_i$ distributes over composition in the expected way, as has no effect on phase.

\begin{lemma} \label{lem:iota-P-composition}
    $(\zero, \iota_i) \circ (\mu,\psi) = (\mu, \iota_i \circ \psi)$
\end{lemma}
\begin{proof}
    It suffices to show that $\kappa^{\iota_i}(v)=0$ for all $v$, which follows from $\cref{lem:kappa-R'-symplectomorphism}$ because,
    for all $v_1$ and $v_2$ we have
    \begin{align*}
        \omega'(\underline{\iota_i(v_1)}, \underline{\iota_i(v_2)})
        &= \omega'(\iota_i(\underline{v_1}), \iota_i(\underline{v_2}))
        = \omega'(\underline{v_1},\underline{v_2}).
        \qedhere
    \end{align*}
\end{proof}

We can also show that $\boxplus$ commutes with $\iota_1$ and $\iota_2$ as expected.
    
\begin{lemma} \label{lem:boxplus-iota}
    $
        \left((\mu_1,\psi_1) \boxplus (\mu_2,\psi_2)\right) \circ \iota_i
        = (\mu_i,\psi_i).
    $
\end{lemma}
\begin{proof}
    Unfolding the definition of composition, we have that
    \begin{align*}
        (\mu_1 \boxplus \mu_2, \psi_1 \boxplus \psi_2) \circ (\zero, \iota_i) 
        = (\mu_0, (\psi_1 \boxplus \psi_2) \circ \iota_i) = (\mu_0, \psi_i)
    \end{align*}
    where $\mu_0$ is defined on standard basis elements $b$ by
    \begin{align*}
        \mu_0(b) &= 0 + (\mu_1 \boxplus \mu_2) (\iota_i b) + \kappa^{\psi_1 \boxplus \psi_2}(\iota_i b) 
        = \mu_i b + \kappa^{\psi_1 \boxplus \psi_2}(\iota_i b)
    \end{align*}
    \cref{lem:kappa-basis-0} states that the function $\kappa$ on a standard basis value is always $0$.
    But if $b$ is a standard basis vector then so is $\iota_i$. Thus $\mu_0(b)=\mu_i(b)$, which completes the proof.
\end{proof}

Finally, we establish the associativity properties of $\textrm{pow}$ and $\cprod$ on morphisms.

\begin{lemma} \label{app:lem:pow-spec}
    Let $[U]$ be a projective Clifford with condensed encoding $(\mu,\psi) \in \PCat(\Pobj{V},\Pobj{V'})$, and let $[V]$ be a projective Clifford such that $V P V^\dagger = (U P U^\dagger)^r$.
    Then the condensed encoding of $V$ is $\pow{(\mu,\psi)}{r}$.
\end{lemma}
\begin{proof}
    Let $b$ be a basis element of $V$; it suffices to check the action of $[V]$ on a basis element $\Delta_{\underline b}$:
    \begin{align*}
        (U \Delta_{\underline{b}} U^\dagger)^{\underline r}
        &= (\zeta^{\mu(b)} \Delta_{\underline{\psi(b)}})^{\underline{r}} \\
        &= \zeta^{r \mu(b)} \Delta_{\underline{r} ~\underline{v}} \\
        &= \zeta^{r \mu(b)} \zeta^{\tfrac{1}{d} \sgn{\underline{r}~\underline{v}}} \Delta_{\underline{r v}} \\
        &= V \Delta_{\underline{v}} V^\dagger
    \end{align*}
\end{proof}

\begin{lemma} \label{app:lem:cprod-spec}
    Let $[U_1],[U_2]$ be projective Cliffords with condensed encodings $(\mu_1,\psi_1)$ and $(\mu_2,\psi_2)$ respectively. If $[V]$ be a projective Clifford satisfying $V P V^\dagger = U_1 P U_1^\dagger \cprod U_2 P U_2^\dagger$,
    then the encoding of $[V]$ is $(\mu_1,\psi_1) \cprod (\mu_2,\psi_2)$.
\end{lemma}
\begin{proof}
    As in the proof of \cref{app:lem:pow-spec}, it suffices to check the action of $[V]$ on a basis element $\Delta_{\underline{b}}$:
    \begin{align*}
        &(U_1 \Delta_{\underline{v}} U_1^\dagger) \cprod (U_2 \Delta_{\underline{b}} U_2^\dagger) \\
        &= \zeta^{{\mu_1(b)}} \Delta_{\underline{\psi_1(b)}} \cprod \zeta^{{\mu_2(b)}} \Delta_{\underline{\psi_2(b)}} \\
        &= \zeta^{{\mu_1(b)} + {\mu_2(b)}} (-1)^{\sgn{\omega'(\underline{\psi_1(b)},\underline{\psi_2(b)})} + \sgn{\underline{\psi_1(b)} + \underline{\psi_2(b)}}} \Delta_{\underline{\psi_1(b) + \psi_2(b)}}
        \tag{\cref{eqn:condensed-product-defn}} \\
        &= V \Delta_{\underline{b}} V^\dagger
    \end{align*}
\end{proof}

\subsection{Soundness}
To prove the categorical semantics is sound, we first prove that substitution corresponds to composition.

\begin{lemma} \label{lem:phase-semantics-substitution-sound}
    If $x:\tau \vdashP t : \tau'$ and $\cdot \vdashL v : \overline{\tau}$, then $\interpP{t\{v/x\}} = \interpP{t} \circ \interpP{v}$.
\end{lemma}
\begin{proof}
    By induction on $x:\tau \vdashP t : \tau'$.

    \proofpart[$t=x$]
        If $t=x$ then $\interpP{t\{v/x\}}=\interpP{v}$ and $\interpP{t}=(\zero, id)$ is the identity, so $\interpP{x} \circ \interpP{v} = \interpP{v}$.

    \proofpart[$t=\letin{y}{t_1}{t_2}$] According to the typing rule for let statements, $x$ must be in the domain of $t_1$. Thus
        \begin{align*}
            \interpP{t\{v/x\}}
            &= \interpP{\letin{y}{t_1\{v/x\}}{t_2}} \\
            &= \interpP{t_2} \circ \interpP{t_1\{v/x\}}
                \tag{\cref{fig:cat-semantics-pauli}} \\
            &= \interpP{t_2} \circ \interpP{t_1} \circ \interpP{v}
                \tag{induction hypothesis} \\
            &= \interpP{\letin{y}{t_1}{t_2}} \circ \interpP{v}
                \tag{\cref{fig:cat-semantics-pauli}}
        \end{align*}

    \proofpart[$t=\config{a} t'$]
        In this case, $x$ occurs in the domain of both $a$ and $t'$.
        \begin{align*}
            \interpP{t\{v/x\}}
            &= \interpP{\config{a\{v/x\}} (t'\{v/x\})} \\
            &= \config{\interpL{a\{v/x\}}} \interpP{t'\{v/x\}} \\
            &= \config{\interpL{a} \circ \interpL{v}} \interpP{t'} \circ \interpP{v}
                \tag{\cref{lem:substitution-composition} and induction hypothesis}
        \end{align*}
        In \cref{lem:config-composition} we prove that we can pull $\interpL{v}$ from the phase $\interpL{a} \circ \interpL{v}$ since $\interpP{v} = (\zero,\interpL{v})$, to obtain
        \begin{align*}
            \left( \config{\interpL{a}} \interpP{t'} \right) \circ \interpP{v}
            &= \interpP{\config{a} t'} \circ \interpP{v}
        \end{align*}
    
    \proofpart[$t=\pow{t'}{a}$]
        It suffices to show that 
        $\pow{(\mu,\psi)}{a} \circ (r, v_1 \otimes v_2)
        =
        \pow{(\mu,\psi) \circ (r, v_1)}{a \circ v_2}$,
    which follows from \cref{app:lem:pow-spec}.

    \proofpart[$t = t_1 \cprod t_2$]
        It suffices to show that $(g_1 \cprod g_2) \circ f
    = (g_1 \circ f) \cprod (g_2 \circ f)$, which follows from \cref{app:lem:cprod-spec}.

    \proofpart[$t=\caseof{t'}{\PauliX \rightarrow t_x \mid \PauliZ \rightarrow t_z}$] \label{proof:P-substitution-sound-case}
        Follows directly from the induction hypothesis and associativity of morphisms in $\PCat$.
        \begin{align*}
            \interpP{t\{v/x\}}
            &= \interpP{\caseof{t'\{v/x\}}{\PauliX \rightarrow t_x \mid \PauliZ \rightarrow t_z}} \\
            &= \left(\interpP{t_x} \boxplus \interpP{t_z}\right) \circ \interpP{t'\{v/x\}} \\
            &= \left(\interpP{t_x} \boxplus \interpP{t_z}\right) \circ \interpP{t'} \circ \interpP{v}
                \tag{induction hypothesis} \\
            &= \interpP{\caseof{t'}{\PauliX \rightarrow t_x \mid \PauliZ \rightarrow t_z}} \circ \interpP{v}
        \end{align*}

    \proofpart[$t=\caseof{t'}{\inl{x_1} \rightarrow t_1 \mid \inr{x_2} \rightarrow t_2}$]
        Similar to the previous case.

    \proofpart{$t=\iota_i t'$}
        Also follows from associativity, as in \cref{proof:P-substitution-sound-case}.
\end{proof}

We can now prove the main soundness theorem.

\begin{theorem}[\cref{thm:phase-types-semantics-sound}] \label{app:thm:phase-types-semantics-sound}
    If $\cdot \vdashP t : \tau$ and $t \rightarrow t'$ then $\interpP{t} = \interpP{t'}$
\end{theorem}
\begin{proof}
    By case analysis on the $\beta$ reduction rules.

    \proofpart[$\letin{x}{\config{r}{v}}{t'} \rightarrow_\beta \config{r} t'\{v/x\}$]

    Unfolding definitions, we see that
    \begin{align*}
        \interpP{\letin{x}{\config{r}{v}}{t'}}
        &= \interpP{t'} \circ \interpP{\config{r}{v}} \\
        &= \interpP{t'} \circ (\interpL{r}, \interpL{v})
    \end{align*}
    \cref{lem:config-composition-0} shows that the phase $\interpL{r}$ can be pulled to the front of the equation:
    \begin{align*}
        \config{\interpL{r}} \interpP{t'} \circ (\zero, \interpL{v})
    \end{align*}
    But this is equal to $\interpP{\config{r} t'\{v/x\}}$ by \cref{lem:phase-semantics-substitution-sound}.

    \proofpart[$\config{r'}\config{r}v \rightarrow_\beta \config{r'+r} v$]

    Follows from unfolding definitions:
    \begin{align*}
        \interpP{\config{r'} \config{r} v} 
        &= \config{\interpL{r'}} \config{\interpL{r}} (\zero,\interpL{v}) \\
        &= (\interpL{r'} + \interpL{r}, \interpL{v}) \\
        &= \interpP{\config{r'+r} v}
    \end{align*}

    \proofpart[$(\config{r_1}v_1) \cprod (\config{r_2} v_2) \rightarrow_\beta \config{r_1+r_2+k}(v_1+v_2)$]
    Follows directly from the definition of $\cprod$ on morphisms in $\PCat$.

    \proofpart[$\pow{\config{r}v}{r'} \rightarrow_\beta \config{r'r + k}{(r' \cdot v)}$]
    Also follows from definitions.
    \begin{align*}
        \interpP{\pow{\config{r}{v}}{r'}}
        &= (\interpL{r}, \interpL{v}) \circ \interpL{r'} \\
        &= (r' \cdot \interpL{r} + k, r' \cdot \interpL{v}) \\
        &= \interpP{\config{r' r + k}{(r' \cdot v)}}
    \end{align*}

    \proofpart[$\iota_i(\config{r}{v}) \rightarrow_\beta \config{r}{\iota_i(v)} $]
    Follows from \cref{lem:iota-P-composition}.

    \proofpart[$\caseof{\config{r}[r_x,r_z]}{\PauliX \rightarrow t_x \mid \PauliZ \rightarrow t_z}
    \rightarrow_\beta \config{r + k} \pow{t_x}{r_x} \cprod \pow{t_z}{r_z}$]
    Unfolding definitions, we see that
    \begin{align*}
        \interpP{\caseof{\config{r}[r_x,r_z]}{\PauliX \rightarrow t_x \mid \PauliZ \rightarrow t_z}}
        &= (\interpP{t_x} \boxplus \interpP{t_z}) \circ (\interpL{r}, \interpP{[r_x,r_z]}) \\
        &= \config{\interpL{r}} (\interpP{t_x} \boxplus \interpP{t_z}) \circ (\zero, 1 \mapsto r_x \oplus r_z)
    \end{align*}
    Let $\interpP{t_x}=(r_x',v_x)$ and $\interpP{t_z}=(r_z',v_z)$ both in $\PCat(\Punit,\Pobj{V})$, 
    and let $\gamma^\boxplus$ be the projective Clifford associated with $\interpP{t_x} \boxplus \interpP{t_z} = (r_x' \boxplus r_z', v_x \boxplus v_z)$.
    By the correspondence with the Pauli group, it suffices to show that
    \begin{align}
        \gamma^{\boxplus}(\Delta_{\underline{[r_x,r_z]}}) =
            \left(\zeta^{r_x'} \Delta_{\underline{v_x}}\right)^{r_x}
            \cprod 
            \left(\zeta^{r_z'} \Delta_{\underline{v_z}}\right)^{r_z}
            \label{eqn:Psoundness-pauli-case}
    \end{align}
    First, observe that $\Delta_{\underline{[r_x,r_z]}} = (-1)^{\sgn{\underline{r_x}~\underline{r_z}}} \Delta_{\underline{[0,r_z]}} \cprod \Delta_{\underline{[r_x,0]}}$. Furthermore, we have that
    \begin{align*}
        \Delta_{\underline{[r_x,0]}}
        &= \Delta_{\underline{r_x}~\underline{[1,0]}} = (\Delta_{\underline{[1,0]}})^{\underline{r_x}} \\
        \Delta_{\underline{[0,r_z]}}
        &= \Delta_{\underline{r_z}~\underline{[0,1]}} = (\Delta_{\underline{[0,1]}})^{\underline{r_z}}
    \end{align*}
    We can then distribute $\gamma^\boxplus$ inside the $\cprod$ operator due to \cref{lem:gamma-distributes-cprod} as well as under the exponent, since $U \Delta_v^r U^\dagger = (U \Delta_v U^\dagger)^r$.
    Thus we can see that \cref{eqn:Psoundness-pauli-case} is equal to
    \begin{align*}
        \zeta^{\tfrac d 2 \sgn{\underline{r_x}~\underline{r_z}}}
        \left( \gamma^{\boxplus} \Delta_{\underline{[0,1]}}\right)^{r_z}
        \cprod
        \left(\gamma^{\boxplus}\Delta_{\underline{[1,0]}}\right)^{r_x}
    \end{align*}
    Finally, it suffices to observe that $\gamma^{\boxplus}\Delta_{\underline{[1,0]}}$ is equal to $\zeta^{r_x'} \Delta_{\underline{v_x}}$ and similarly for $\gamma^{\boxplus}\Delta_{\underline{[0,1]}}$, which completes the proof.

    \proofpart[$\caseof{\config{r}[v_1,v_2]}{\inl{x_1} \rightarrow t_1 \mid \inr{x_2} \rightarrow t_2}
    \rightarrow_\beta \config{r} t_1\{v_1/x_1\} \cprod t_2\{v_2/x_2\}$]

    Let $\gamma^i$ be the projective Clifford associated with $\interpP{t_i}$ and let $\gamma^\boxplus$ be $\gamma^1 \boxplus \gamma^2$ As in the previous case, it suffices to show that
    \begin{align*}
        \gamma^{\boxplus} (\Delta_{\underline{v_1 \oplus v_2}}) = \gamma^1(\Delta_{\underline{v_1}}) \cprod \gamma^2(\Delta_{\underline{v_2}})
    \end{align*}
    Observe that $\Delta_{\underline{v_1 \oplus v_2}}=\Delta_{\underline{v_1 \oplus 0} + \underline{0 \oplus v_2}} = \Delta_{\underline{v_1 \oplus 0}} \cprod \Delta_{\underline{0 \oplus v_2}}$.
    Then, since $\gamma^\boxplus$ distributes over $\cprod$ by \cref{lem:gamma-distributes-cprod}, we have that
    \begin{align*}
        \gamma^{\boxplus} (\Delta_{\underline{v_1 \oplus v_2}})
        &= \gamma^{\boxplus} (\Delta_{\underline{v_1 \oplus 0}}) \cprod
            \gamma^{\boxplus} (\Delta_{\underline{0 \oplus v_2}})
    \end{align*}
    From there, it suffices to observe that $\gamma^{\boxplus} (\Delta_{\underline{v_1 \oplus 0}})=\gamma^1(\Delta_{\underline{v_1}})$ and vice versa for $0 \oplus v_2$.
\end{proof}

\section{Glossary for Haskell-style Pattern-matching Syntax}
\label{appendix:glossary}
\cref{fig:glossary} contains a translation between the Haskell-style pattern-matching syntax used in the examples in \cref{sec:introduction,sec:examples} and the formal $\lambda$-calculus used in the rest of the paper.
The relationship between the two is based on the definition of a generalized pattern-matching operation of the form $\caseof{e}{p_1 \rightarrow e_1 \mid \cdots \mid p_n \rightarrow e_n}$. In this context, a \emph{pattern} is a simple expression defined as follows:
\begin{align*}
    p &::= x \mid \PauliX \mid \PauliZ \mid \inl{p} \mid \inr{p} \mid \iota_i(p)  \tag{patterns} \\
\end{align*}
The generalized pattern-matching syntax is syntactic sugar defined as follows:

\begin{align*}
    \caseof{e}{x \mapsto e'} &\triangleq \letin{x}{e}{e'} \\ \\
    \caseof{e}{\begin{aligned}
        \inl{p_1} &\rightarrow e_1 \\
        \cdots \\
        \inl{p_m} &\rightarrow e_m \\
        \inr{p_1'} &\rightarrow e_1' \\
        \cdots \\
        \inr{p_{n}' &\rightarrow e_{n}'}
    \end{aligned}}
        &\triangleq 
            \caseof{e}{\begin{aligned}
                \inl{x_1} &\rightarrow \caseof{x_1}{\begin{aligned} 
                        p_1 &\rightarrow e_1 \\
                        &\cdots \\
                        p_m &\rightarrow e_m
                    \end{aligned}} \\
                \inr{x_2} &\rightarrow \caseof{x_2}{\begin{aligned} 
                        p_1' &\rightarrow e_1' \\
                        &\cdots \\
                        p_n' &\rightarrow e_n'
                    \end{aligned}}
            \end{aligned}}
            \tag{$x_1$, $x_2$ fresh}
        \\ \\
    \caseof{e}{\begin{aligned}
        \iota_i(p_1) &\rightarrow e_1 \\
        \cdots \\
        \iota_i(p_n) &\rightarrow e_n
    \end{aligned}}
        &\triangleq
            \caseof{e}{\iota_i(x) \rightarrow \caseof{x}{\begin{aligned} 
                    p_1 &\rightarrow e_1 \\
                    &\cdots \\
                    p_n &\rightarrow e_n
                \end{aligned}}}
            \tag{$x$ fresh}
\end{align*}

\begin{figure}
    \centering
    \begin{minipage}{0.3\textwidth}
\begin{lstlisting}
    foo :: |^ tau1 -o tau2 ^|
    foo |^ p1 ^| = e1
    ...
    foo |^ pn ^| = en
\end{lstlisting}
    \end{minipage}
    \begin{minipage}{0.1\textwidth}
        \centering
        \Large $\Leftrightarrow$
    \end{minipage}
    \begin{minipage}{0.5\textwidth}
    \begin{align*}
        \textrm{foo} = \lambda \lift{x}. \caseof{x}{\begin{aligned}
            p_1 &\rightarrow e_1 \\
            &\cdots \\
            p_n &\rightarrow e_n
        \end{aligned}}
        \tag{$x$ fresh}
    \end{align*}
    \end{minipage}
    
    \medskip \medskip
    \begin{minipage}{0.3\textwidth}
\begin{lstlisting}
    foo :: |^ tau1 -o tau2 ^|
    foo |^ p1 ^| *= e1
    ...
    foo |^ pn ^| *= en
\end{lstlisting}
    \end{minipage}
    \begin{minipage}{0.1\textwidth}
        \centering
        \Large $\Leftrightarrow$
    \end{minipage}
    \begin{minipage}{0.5\textwidth}
\begin{lstlisting}
    foo :: |^ tau1 -o tau2 ^|
    foo |^ p1 ^| = p1 * e1
    ...
    foo |^ pn ^| = p2 * en
\end{lstlisting}
    \end{minipage}
    \caption{Correspondence between Haskell-style pattern-matching syntax and $\lambdaP$ terms with generalized pattern-matching.}
    \label{fig:glossary}
    \Description{Top: A Haskell-style function \texttt{foo} expressed using nested pattern-matching and the equivalent $\lambdaP$ term using case-style nested pattern-matching. Bottom: A Haskell-style function \texttt{foo} expressed using star-equal notation and the equivalent function expressed without using star-equal notation.}
\end{figure}

\end{document}